\documentclass[a4paper]{amsart}

\usepackage{xcolor}
\definecolor{webgreen}{rgb}{0,.5,0}
\definecolor{webbrown}{rgb}{.6,0,0}
\definecolor{RoyalBlue}{cmyk}{1, 0.50, 0, 0}
\usepackage[colorlinks=true, breaklinks=true, urlcolor=webbrown, linkcolor=RoyalBlue, citecolor=webgreen,backref=page]{hyperref}

\usepackage{epsfig, graphicx, subfigure,tikz}
\usepackage{verbatim, setspace, soul}
\usepackage{amsmath, amssymb}

\usepackage{newtxtext,newtxmath}

\usepackage{mathabx}
\usetikzlibrary{decorations.markings}

\newcommand{\R}		{\mathbb{R}}
\newcommand{\C}		{\mathbb{C}}
\newcommand{\N}		{\mathbb{N}}
\newcommand{\Z}		{\mathbb{Z}}

\newcommand{\supp}{\mathrm{supp}}

\newcommand{\re}{\mathrm{Re}}
\newcommand{\im}{\mathrm{Im}}

\renewcommand{\arg}{\mathrm{arg}}
\renewcommand{\det}{\mathrm{det}}

\newcommand{\qasq}{\quad \text{as} \quad}
\newcommand{\qandq}{\quad \text{and} \quad}

\newcommand{\dd}{\mathrm{d}}
\newcommand{\ic}{\mathrm{i}}

\newcommand{\RS}{\mathfrak S}
\newcommand{\z}{{\boldsymbol z}}

\newcommand{\s}{{\boldsymbol s}}

\newcommand{\be}{\beta}

\newcommand{\ga}{\gamma}
\newcommand{\Ga}{\Gamma}

\newcommand{\la}{\lambda}
\newcommand{\ep}{\varepsilon}

\newcommand{\Tr}{{\mathsf{Tr}}\,}

\newcommand{\rhy}   {\textnormal{RHP}-${\boldsymbol Y}$}
\newcommand{\rht}   {\textnormal{RHP}-${\boldsymbol T}$}
\newcommand{\rhs}   {\textnormal{RHP}-${\boldsymbol S}$}
\newcommand{\rhr}   {\textnormal{RHP}-${\boldsymbol R}$}
\newcommand{\rhn}   {\textnormal{RHP}-${\boldsymbol N}$}
\newcommand{\rhp}   {\textnormal{RHP}-${\boldsymbol P}$}

\newtheorem{theorem}{Theorem}[section]

\newtheorem{proposition}{Proposition}[section]
\newtheorem{lemma}[proposition]{Lemma}
\newtheorem{definition}[proposition]{Definition}

\theoremstyle{definition}
\newtheorem{notation}{Convention}

\numberwithin{equation}{section}

\begin{document}

\title[Complex cubic ensemble: two-cut case]{Investigation of the two-cut phase region  in the \\ complex cubic ensemble of random matrices}

\date{\today}

\author{Ahmad Barhoumi}
\address{Department of Mathematics, University of Michigan, East Hall, 530~Church Street, Ann Arbor, MI 48109, USA}
\email{barhoumi@umich.edu}

\author{Pavel  Bleher}
\address{Department of Mathematical Sciences, Indiana University-Purdue University Indianapolis, 402~North Blackford Street, Indianapolis, IN 46202, USA}
\email{pbleher@iupui.edu}

\author{Alfredo Dea\~no}
\address{Department of Mathematics, Universidad Carlos III de Madrid, Avda. de la Universidad 30, 28911 Legan\'es, Madrid, Spain}
\email{alfredo.deanho@uc3m.es}

\author{Maxim Yattselev}
\address{Department of Mathematical Sciences, Indiana University-Purdue University Indianapolis, 402~North Blackford Street, Indianapolis, IN 46202, USA}
\email{maxyatts@iupui.edu}

\begin{abstract}
We investigate the phase diagram of  the complex cubic unitary ensemble of random matrices with the potential $V(M)=-\frac{1}{3}M^3+tM$ where $t$ is a complex parameter. As proven in our previous paper \cite{BlDeaY17}, the whole phase space of the model, $t\in\C$, is partitioned into two phase regions, $O_{\mathsf{one-cut}}$ and $O_{\mathsf{two-cut}}$, such that in $O_{\mathsf{one-cut}}$ the equilibrium measure is supported by one Jordan arc (cut) and in $O_{\mathsf{two-cut}}$ by two cuts. The regions $O_{\mathsf{one-cut}}$ and $O_{\mathsf{two-cut}}$ are separated by critical curves, which can be calculated in terms of critical trajectories of an auxiliary quadratic differential. In \cite{BlDeaY17} the one-cut phase region was investigated in detail. In the present paper  we investigate the two-cut region. We prove that in the two-cut region the endpoints of the cuts are analytic functions of the real and imaginary parts of the parameter $t$, but not of the parameter $t$ itself (so that the Cauchy--Riemann equations are violated for the endpoints). We also obtain the semiclassical asymptotics of the orthogonal polynomials associated with the ensemble of random matrices and their recurrence coefficients.  The proofs are based on the Riemann--Hilbert approach to semiclassical asymptotics of the orthogonal polynomials and the theory of $S$-curves and quadratic differentials.
\end{abstract}

\keywords{cubic random matrix model, partition function, equilibrium measure, S-curve, quadratic differential, orthogonal polynomials, non-Hermitian orthogonality, Riemann--Hilbert problem, nonlinear steepest descent method.}

\subjclass[2010]{33C47, 30E15, 31A25,  15B52} 

\thanks{The first author (A.B.) was partially supported by the National Science Foundation (NSF) under DMS-1812625. The work of the second author (P.B.) is supported in part by the NSF Grants DMS-1265172 and DMS-1565602. The second author (P.B.) acknowledges hospitality and support from MSRI (Mathematical Sciences Research Institute) during the research program ``Universality and Integrability in Random Matrix Theory and Interacting Particle Systems" in September--December 2021. The third author (A.D.) acknowledges financial support from the EPSRC grant ``Painlev\'e equations: analytical properties and numerical computation", reference EP/P026532/1, as well as support by the Madrid Government (Comunidad de Madrid-Spain) under the Multiannual Agreement with UC3M in the line of Excellence of University Professors (EPUC3M23), and in the context of the V PRICIT (Regional Programme of Research and Technological Innovation). The research of the fourth author (M.Y.) is supported by a grant from the Simons Foundation, CGM-706591.}

\dedicatory{Dedicated to the memory of Freeman Dyson}

\maketitle

\section{Introduction}
\label{s:intro}

This work is a continuation of the works of Bleher and Dea\~no \cite{BD1,BD2} and Bleher, Dea\~no, and  Yattselev \cite{BlDeaY17}. The main goal of the whole project is to analyze the large $N$ asymptotics of the partition function and correlation functions in the unitary ensemble of random matrices with the cubic potential $V(M)=-\frac{1}{3}M^3+tM$, where $t$ is a complex parameter. The partition function of this ensemble is formally defined by the matrix integral over the space of $N\times N$ Hermitian matrices,
\[
Z_N^{\rm RM}(t)=\int_{\mathcal H_N} e^{-N\Tr 
\left(-\frac{1}{3}M^3+tM\right)} \mathrm dM.
\]
The integral is formal because it is divergent, and it needs a regularization. To define such a regularization consider the (formal) identity
\[
Z_N^{\rm RM}(t)=\frac{\pi^{N(N-1)/2}}{\prod_{k=1}^N k!}\int_{\R}\ldots\int_{\R} \prod_{1\leq j<k\leq N}(z_j-z_k)^2\,
\prod_{j=1}^N e^{-N \left(-\frac{1}{3}z_j^3+tz_j\right)}\mathrm dz_1\ldots \mathrm dz_N
\]
see, \cite[formula (1.2.18)]{BL}. The random matrix partition function $Z_N^{\rm RM}(t)$ can now be regularized by changing the contour of integration from the real line $\R$ to the contour $\Ga$ depicted in Figure \ref{Gamma_t1}, which goes from $(-\infty)$ to $(e^{\pi \ic/3}\infty)$. Hence, we are interested in studying the regularized eigenvalue partition function defined as
\begin{equation}
\label{ZN}
Z_N(t) := \int_{\Ga}\ldots\int_{\Ga} \prod_{1\leq j<k\leq N}(z_j-z_k)^2\,
\prod_{j=1}^N e^{-N \left(-\frac{1}{3}z_j^3+tz_j\right)}\mathrm dz_1\ldots \mathrm dz_N,
\end{equation}
which is well defined (the integrals are convergent) for all values $t\in \C$. 

\begin{figure}[h]
\label{Gamma_t1}
\includegraphics[scale=0.4]{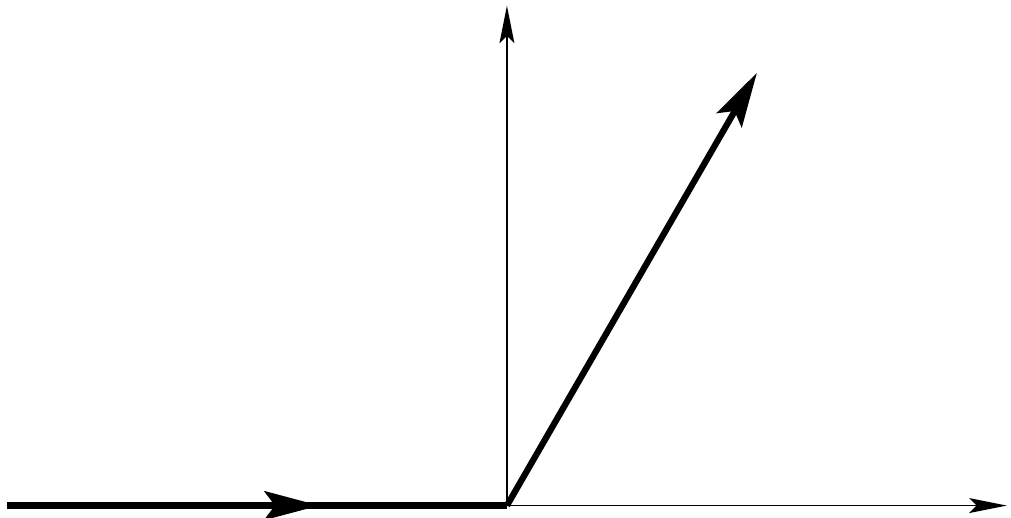}
\begin{picture}(0,0)
\put(-46,40){$\Ga$}
\put(-67,3){$0$}
\end{picture}
\caption[sectors ]{\small The  contour  $\Ga$ of integration.}
 \end{figure}
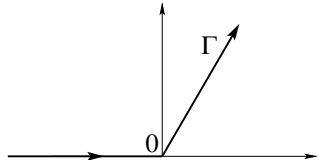

As shown in \cite{BD1}, the partition function $Z_N(t)$ can be used to study the topological expansion in the cubic ensemble of random matrices and enumeration of regular graphs of degree 3 on Riemannian surfaces. To be more precise, let us consider the cubic potential in the form $M^2/2-uM^3$, where $u>0$, and the corresponding eigenvalue matrix integral is given by
\begin{equation}
\label{ZNu}
\Xi_N(u)=\int_\Ga\ldots\int_\Ga \prod_{1\leq j<k\leq N}(\zeta_j-\zeta_k)^2\,
\prod_{j=1}^N e^{-N \left(\frac12\zeta_j^2-u\zeta_j^3\right)}\mathrm d\zeta_1\ldots \mathrm d\zeta_N.
\end{equation}
We would like to consider all the possible values, including the complex ones, of the parameter $u$. Formula \eqref{ZNu} is not very convenient for this purpose  because the contour of integration $\Ga$ should be rotated to secure the convergence of the integral depending on the argument of \( u \). Instead, let us make the change of variables
\[
\zeta_j=(3u)^{-1/3}z_j+\frac{1}{6u},
\]
which yields that
\[
\frac{\zeta_j^2}{2}-u\zeta_j^3-\frac{1}{108u^2}=-\frac{z_j^3}{3}+tz_j,
\]
where
\begin{equation}\label{tu}
t=\frac{1}{4(3u)^{4/3}}\,.
\end{equation}
This implies the relation between the partition functions,
\[
\Xi_N(u) =(4t)^{N^2/4}e^{-2t^{3/2}N/3} Z_N(t). 
\]

As proven in  \cite{BD1}, the free energy of the cubic model $\mathcal F_N(u)$   admits an asymptotic expansion as $N\to\infty$ in powers of ${N^{-2}}$, that is,
\begin{equation}\label{free}
\mathcal F_N(u):=
\frac{1}{N^2} \ln \frac{\Xi_N(u)}{\Xi_N(0)}\sim \sum_{g=0}^{\infty}\frac{ F^{(2g)}(u)}{N^{2g}}
\end{equation}
for any $u$ in the interval $0\le u< u_c$, where \( u_c:=3^{1/4} 18^{-1} \) is a critical point. In addition, the functions $F^{(2g)}(u)$ admit an analytic continuation to the disk $|u|<u_c$ on the complex plane,  and if we expand them in powers of $u$,
\begin{equation}\label{top2}
F^{(2g)}(u)=\sum_{j=1}^\infty \frac{f^{(2g)}_{2j}u^{2j}}{(2j)!},
\end{equation}
then the coefficient $f^{(2g)}_{2j}$ is a positive integer that counts the number of $3$-valent connected graphs with $2j$ vertices on a Riemann surface of genus $g$. Asymptotic expansion \eqref{free} is called a {\it topological expansion}. For more details on this aspect of the theory,  we refer the reader to the classical papers of Bessis, Itzykson and Zuber \cite{BIZ}, Br\'{e}zin, Itzykson, Parisi and Zuber \cite{BIPZ}, the monograph of Forrester \cite[Section 1.6]{Forrester}, the works of  Mulase \cite{Mulase}, Di Francesco \cite{DiFrancesco}, Ercolani and McLaughlin \cite{EMcL,EMcLP}, Eynard \cite{Eyn}, and references therein, or a very readable introduction by Zvonkin \cite{Zvonkin}. It is noteworthy that the general idea of a topological expansion goes back to the classical work of 't~Hooft \cite{tH74}.

As shown in \cite{BD1}, the coefficients $f^{(2g)}_{2j}/(2j)!$ of power series \eqref{top2} behave, when $j\to\infty$, as
\[
\begin{aligned}
\frac{f^{(2g)}_{2j}}{(2j)!} &=\frac{K_{2g}  j^{\frac{5g-7}{2}}}{u_c^{2j}}\left(1+\mathcal{O}\left(j^{-1/2}\right)\right),
\quad K_{2g}>0.
\end{aligned}
\]
This implies that  $u_c$ is the radius of convergence of each power series \eqref{top2}. In fact, $u=u_c$ is a singular point of each of the functions $\eqref{top2}$. The topological expansion in a neighborhood of the critical point
$u_c$ has been obtained in the work of Bleher and Dea\~no \cite{BD2} and it is closely related to the Painlev\'e I equation as follows:
consider a formal series
\begin{equation}
\label{genfun}
\sum_{g=0}^{\infty} \Gamma\left(\frac{5g-1}{2}\right)\frac{u_c^{g}K_{2g}}{6\cdot 3^{1/4}} \lambda^{\frac{1-5g}{2}}\,,
\end{equation}
it can be shown that there exists a one-parameter family of solutions $\{y_{\alpha}(\lambda)\}$, $\alpha\in\R,$ of the Painlev\'e~I differential equation
\[
y''(\lambda)=a_0y^2(\lambda)-a_1\lambda,
\]
with $a_0=2^{\frac{5}{2}}3^{\frac{9}{4}}$ and $a_1=2^{\frac{3}{2}}3^{-\frac{5}{4}}$, that admits an asymptotic expansion \eqref{genfun} as $\lambda\to-\infty$. This family consists of the Boutroux \emph{tronqu\'ee} solutions, which means that they are asymptotically free of poles in one of the five canonical sectors of angle $2\pi/5$ in the complex plane that appear naturally when considering rotational symmetries of the Painlev\'e~I equation. The parameter $\alpha$ appears in the Painlev\'e~I Riemann--Hilbert problem, as described by Kapaev in \cite{Kapaev}, and in the context of the cubic model, it is related to the choice of the original contour of integration $\Gamma$, see \cite{BD1,BD2}. In the present case, we have $\alpha=1$, which makes the solution of Painlev\'e~I \emph{tritronqu\'ee}, that is, asymptotically free of poles in four of the five canonical sectors of angle $2\pi/5$ in the complex plane. We refer the reader to \cite[Section 2.1]{BD2} and also the work of Joshi and Kitaev \cite{JK} for more details.
  
It is noteworthy that the key ingredient in the proof of topological expansion \eqref{free} in \cite{BD1} is  the derivation of semiclassical asymptotic formulae for the recurrence coefficients $\ga_n^2$, $\be_n$, see \eqref{tu} and \eqref{cm4} further below, of the corresponding monic orthogonal polynomials $P_n(z)=z^n+\ldots$ defined via orthogonality relations
\[
\int_{\Gamma} P_n(z)z^k e^{-N\left(\frac12z^2-uz^3\right)}\mathrm dz=0, \quad k\in\{0,1,\ldots,n-1\}.
\]
The idea of using the orthogonal polynomials in ensembles of random matrices goes back to the classical works of Dyson and Mehta, see the works \cite{Dyson} and \cite{Mehta} and references therein. Namely, as proven in \cite{BD1}, for any $u\in[0,u_c)$, there exists $\ep>0$ such that as $N,n\to \infty$ with $1-\ep\le \frac{n}{N}\le 1+\ep$, the recurrence coefficients $\ga_n^2$ and $\be_n$ admit asymptotic expansions in powers of $\frac{1}{N^2}$:
\begin{equation}\label{asympgb}
\left\{
\begin{aligned}
\gamma_n^2 & \sim \sum_{k=0}^{\infty}\frac{1}{N^{2k}}g_{2k}\left(\frac{n}{N},u\right),\\
\beta_n & \sim \sum_{k=0}^{\infty}\frac{1}{N^{2k}}
b_{2k}\left(\frac{n}{N}+\frac{1}{2N},u\right),
\end{aligned}
\right.
\end{equation}
where the functions $g_{2k}(s,u)$, $b_{2k}(s,u)$ do not depend on $n$ and $N$ and are analytic in $s$ at $s=1$.

In \cite{BD2} this asymptotic expansion is extended to the double scaling asymptotic expansion  of the recurrence coefficients at the critical point $u_c$. In the double scaling regime we set 
\[
\frac{n}{N}=1+vN^{-4/5}.
\]
where $v\in\mathbb{R}$ is a scaling variable. Then as proven in \cite{BD2}, at $u=u_c$ the recurrence coefficients $\ga_n^2$ and $\be_n$ admit asymptotic expansions in powers of $N^{-2/5}$ as $N\to\infty$:
\begin{equation} \label{dsc}
\left\{
\begin{aligned}
\ga_{n}^2&\sim \ga^2_c+\sum_{k=1}^{\infty}\frac{1}{N^{2k/5}}\, {p}_{2k}(v),\\
\be_{n}& \sim \be_c+\sum_{k=1}^{\infty}\frac{1}{N^{2k/5}}\,{q}_{2k}(\tilde{v}),
\end{aligned}
\right.
\end{equation}
where the functions ${p}_{2k}(v)$, ${q}_{2k}(\tilde{v})$ are expressed in terms of the Boutroux tritronqu\'ee solution to Painlev\'e I equation, \( \ga_c^2,\be_c \) are known explicitly, and $\tilde v=v+N^{-1/5}/2\,$. As shown in \cite{BD2}, expansions \eqref{asympgb} and \eqref{dsc} can be extended for large $N$ to $u$ in overlapping intervals: $[0,u_c-N^{-0.79}]$ for \eqref{asympgb} and $[u_c-N^{-0.65},u_c]$ for \eqref{dsc}, and this can be used to obtain the double scaling asymptotic formula for
the partition function at \( u_c \).  Namely, let
\[
u-u_c=2^{-\frac{12}{5}}3^{-\frac{7}{4}}\la N^{-\frac{4}{5}},
\]
where $\la$ is a complex scaling variable in the double scaling regime. Then for $\la$ outside of a neighborhood
of the poles of \( y(\lambda) \), the Boutroux tritronqu\'ee solution to Painlev\'e~I equation, the partition function $\Xi_N(u)$ can be written as
\begin{equation}\label{ZNuu}
\Xi_N(u)=\Xi_N^{\rm reg}(u)\Xi_N^{\rm sing}(\lambda)
\left(1+\mathcal{O}(N^{-\varepsilon})\right),\quad \varepsilon>0,
\end{equation}
where the regular and singular factors are given by
\[
\Xi_N^{\rm reg}(u)=e^{N^2 [a+b(u-u_c)+c(u-u_c)^2]+d} \quad \text{and} \quad \Xi_N^{\rm sing}(\lambda)=e^{-Y(\lambda)},
\]
where $a,b,c,d$ are some explicit constants and $Y(\lambda)$ is a solution of the differential equation \( Y''(\lambda)=y(\lambda) \) with  the boundary condition
\[
Y(\lambda)=\frac{2\sqrt{6}}{45}(-\lambda)^{5/2}-\frac{1}{48}\log(-\lambda)+\mathcal{O}((-\lambda)^{-5/2}), \qquad \lambda\to-\infty.
\]
Asymptotic formula \eqref{ZNuu} is used in \cite{BD2} to prove the conjecture of David \cite{Dav1,Dav2} that the poles of the tritronqu\'ee solution $y(\la)$ give rise to zeros of $\Xi_N(u)$.  
 
As we have stressed before, the topological expansions from \cite{BD1,BD2,BlDeaY17} were obtained by first analyzing the associated orthogonal polynomials defined by
\begin{equation}
\label{cm3}
\int_{\Ga} z^kP_n(z;t,N)e^{-NV(z;t)} \mathrm dz=0,\quad k\in\{0,\ldots, n-1\},
\end{equation} 
where 
\begin{equation}
\label{cm2}
V(z;t) := -\dfrac{z^3}{3} + tz, \quad t \in \C.
\end{equation}
Due to the non-Hermitian character of the above relations, it might happen that polynomial satisfying \eqref{cm3} is non-unique. In this case we understand by $P_n(z;t,N)$ the monic non-identically zero polynomial of the smallest degree (such a polynomial is always unique). One way of connecting the partition function to orthogonal polynomials is via {\it three term recurrence relations}. More precisely, let
\begin{equation}
\label{cm5a}
h_n(t,N) := \displaystyle \int_{\Ga} P_n^2(z;t,N)e^{-NV(z;t)} \mathrm dz = \frac{D_n(t,N)}{D_{n-1}(t,N)},
\end{equation}
where \( D_{-1}(t,N)\equiv 1 \) and \( D_n(t,N) := \det\big[\int_\Gamma z^{i+j}e^{-NV(z;t)}\dd z\big]_{i,j=0}^n \) is the Hankel determinant of the moments of the measure of integration in \eqref{cm3}. It easy to see that each \( D_n(t,N) \) is an entire function of \( t \) and therefore each \( h_n(t,N) \) is meromorphic in \( \C \). Hence, given \( n \), the set of the values \( t \) for which there exists \( k\in\{0,\ldots,n \} \) such that \( h_k(t,N) =0 \) is countable with no limit points in the finite plane. Outside of this set, the standard argument using orthogonality \eqref{cm3} shows that
\begin{equation}
\label{cm4}
zP_n(z;t,N) = P_{n+1}(z;t,N)+\be_n(t,N) P_n(z;t,N)+\ga_n^2(t,N) P_{n-1}(z;t,N),
\end{equation}
and by analytic continuation \eqref{cm4} extends to those values of \( t \) for which \( D_{n-1}(t,N)D_n(t,N) \neq 0 \) (that is, \( n+1 \)-st and \( n \)-th polynomials appearing \eqref{cm4} have the prescribed degrees), where
\begin{equation}
\label{cm5b}
\ga_n^2(t,N) = h_n(t,N)/h_{n-1}(t,N).
\end{equation}
Notice that \eqref{cm4} remains meaningful even if \( D_{n-2}(t,N)=0 \). In this case \( h_{n-2}(t,N)=0 \), \( h_{n-1}(t,N)=\infty \), and \( h_n(t,N) \) is finite, which means that \( P_{n-2}(z;t,N) \) is orthogonal to itself and therefore is equal to \( P_{n-1}(z;t,N) \) but not to \( P_n(z;t,N) \). Thus, \( \gamma_n^2(t,N)=0 \) and the last term in \eqref{cm4} is zero. Notice also that if \( h_{n-1}(t,N) \neq 0 \) while \( h_n(t,N)=0 \) (that is, \( D_{n-1}(t,N)\neq0 \) and \( D_n(t,N)=0 \)), then we simply have that \( \ga_N^2(t,N)=0 \) by analytic continuation. However, this means that the polynomial \( P_n(z;t,N) \) is orthogonal to itself and therefore \( P_{n+1}(z;t,N) = P_n(z;t,N) \). Then the analytic continuation argument necessitates that \( \beta_n(t,N)=\infty \), which also can be seen from the determinantal representation of \( \beta_n(t,N) \) as it is the difference of the subleading coefficients of \( P_n(z;t,N) \) and \( P_{n+1}(z;t,N) \). It is further known that the recurrence coefficients $\ga_N^2(t,N)$ satisfy the {\it Toda equation}:
\[
\frac{\partial^2 F_N(t)}{\partial t^2}=\ga_N^2(t,N), \qquad F_N(t)=\frac{1}{N^2}\log Z_N(t).
\]
Another way of connecting partition function to orthogonal polynomials is through the formula
\[
Z_N(t) = N!\prod_{n=0}^{N-1}h_n(t,N) \big( = N! D_{N-1}(t,N) \big).
\]
As it turns out, there are two different regions in the complex \( t \)-plane for which the behavior of the polynomials \( P_n(z;t,N) \), and therefore of the partition function, is different. Colloquially, we call them one-cut and two-cut regions, a distinction that will become clear later. In \cite{BD1,BD2,BlDeaY17} the partition function has been analyzed in and on the boundary of the one-cut region. The goal of this work is to start the analysis in the two-cut region. More precisely, here we only consider the asymptotics of the orthogonal polynomials and their recurrence coefficients and postpone the analysis of the partition function for the future project. The structure of the paper is as follows:
\begin{itemize}
\item In Sections \ref{s:Sprop} and \ref{s:Structure} we describe equilibrium measures and corresponding $S$-curves for the cubic model under consideration (those describe asymptotic behavior of the normalized counting measures of zeros of the orthogonal polynomials in the weak$^*$  sense). This leads us to a precise description of the phase diagram (i.e., of the one- and two-cut regions) of the cubic model on the complex $t$-plane. 
\item In Section~\ref{s:main} we present the main results of the paper:  asymptotic formulae for the orthogonal polynomials and their recurrence coefficients in the two-cut phase region.
\item Section \ref{s:pr-geom}--\ref{s:ae} are devoted to the proof of our main results, in them we obtain various results about the detailed structure of the $S$-curves, derive main properties of the dominant terms of the expansions, apply the Riemann--Hilbert approach to asymptotics of the orthogonal polynomials, and finally prove the expansions themselves.
\end{itemize}

\section{Equilibrium Measures and S-Property}\label{s:Sprop}

It is well understood that the zeros of polynomials satisfying \eqref{cm3} asymptotically distribute as a certain \emph{weighted equilibrium measure} on an \emph{S-contour} corresponding to the weight function \eqref{cm2}. In this section we discuss these notions in greater detail. Our consideration will use the works of Huybrechs, Kuijlaars, and Lejon \cite{HKL} and Kuijlaars and Silva~\cite{KS}. Let us start with some
definitions.

\begin{definition} 
\label{def:eq}
Let $V$ be an entire function. The logarithmic energy in the external field $\re V$ of a measure $\nu$ in the complex plane
is defined as
\[
E_V(\nu)=\iint \log \frac{1}{|s-t|}\mathrm d\nu(s) \mathrm d\nu(t)+\int \re V(s)\mathrm d\nu(s).
\]
The equilibrium energy of a contour $\Ga$ in the external field $\re V$ is equal to
\begin{equation}
\label{em1}
\mathcal E_V(\Ga)=\inf_{\nu\in \mathcal M(\Ga)} E_V(\nu),
\end{equation}
where $\mathcal M(\Ga)$ denotes the space of Borel probability measures on $\Ga$.
\end{definition}

When $\re V(s)-\log|s|\to+\infty$ as $\Ga\ni s\to\infty$, there exists a unique minimizing measure for \eqref{em1}, which is called the {\it weighted equilibrium measure} of $\Ga$ in the external field $\re V$, say $\mu_\Ga$, see \cite[Theorem I.1.3]{SaffTotik} or \cite{HKL}. The support of $\mu_\Ga$, say $J_\Ga$, is a compact subset of $\Ga$. The equilibrium measure  $\mu=\mu_\Ga$ is characterized by the Euler--Lagrange variational conditions:
\begin{equation}
\label{em2}
2U^\mu(z)+\re V(z)\;
\left\{
\begin{aligned}
&= \ell,\qquad z\in J_\Ga,\\
&\ge \ell,\qquad z\in \Ga\setminus J_\Ga,
\end{aligned}
\right.
\end{equation}
where $\ell=\ell_\Ga$ is a constant, the Lagrange multiplier, and \( U^\mu(z):=-\int\log|z-s|\dd\mu(s) \) is the logarithmic potential of $\mu$, see \cite[Theorem~I.3.3]{SaffTotik}. 

 Observe that due to the analyticity of the integrand in \eqref{ZN}, \( \Ga \) can be varied without changing the value of \( Z_N(t) \). Henceforth, we suppose that partition function \eqref{ZN} and orthogonal polynomials \eqref{cm3} are defined with $\Ga\in\mathcal T$, where \( \mathcal T \) is the following class of contours. 
\begin{definition} 
 We shall denote by $\mathcal T$ the collection of all piecewise smooth contours that extend to infinity in both directions, each admiting a parametrization $z(s)$, $s\in\R$, for which there exists $\epsilon\in(0,\pi/6)$ and $s_0>0$ such that
\begin{equation}
\label{cm2b}
\left\{
\begin{array}{ll}
|\arg(z(s))-\pi/3|\leq \epsilon, & s\geq s_0, \medskip \\
|\arg(z(s))-\pi|\leq \epsilon, & s\leq -s_0,
\end{array}
\right.
\end{equation}
where $\arg(z(s))\in[0,2\pi)$. 
\end{definition}
Despite the above flexibility, it is well understood in the theory of non-Hermitian orthogonal polynomials, starting with the works of Stahl \cite{St85,St85b,St86} and Gonchar and Rakhmanov \cite{GRakh87}, that one should use the contour whose equilibrium measure has support  \emph{symmetric} (with the \emph{S-property}) in the external field $\re V$. We make this idea precise in the following definition.
\begin{definition} 
The support $J_\Ga$ has the S-property in the external field $\re V$, if it consists of a finite number of open analytic arcs and their endpoints, and  on each arc it holds that
\begin{equation}
\label{em4}
\frac{\partial }{\partial n_+}\,\big(2U^{\mu_\Gamma}+\re V\big)=
\frac{\partial }{\partial n_-}\,\big(2U^{\mu_\Gamma}+\re V\big),
\end{equation}
where $\frac{\partial }{\partial n_+}$ and  $\frac{\partial }{\partial n_-}$ are the normal derivatives from the $(+)$- and $(-)$-side of $\Ga$. We shall say that a curve $\Ga\in\mathcal T$ is an S-curve in the field $\re V$, if $J_\Ga$ has the S-property in this field.
\end{definition}

It is also understood that geometrically $J_\Ga$ is comprised of \emph{critical trajectories} of a certain quadratic differential. Recall that if $Q$ is a meromorphic function, a \emph{trajectory} (resp. \emph{orthogonal trajectory}) of a quadratic differential $-Q(z)\mathrm dz^2$ is a maximal regular arc on which
\[
-Q(z(s))\big(z^\prime(s)\big)^2>0 \quad \big(\text{resp.} \quad -Q(z(s))\big(z^\prime(s)\big)^2<0\big)
\]
for any local uniformizing parameter. A trajectory is called \emph{critical} if it is incident with a \emph{finite critical point} (a zero or a simple pole of $-Q(z)\mathrm dz^2$) and it is called \emph{short} if it is incident only with finite critical points. We designate the expression \emph{critical (orthogonal) graph of $-Q(z)\mathrm dz^2$} for the totality of the critical (orthogonal) trajectories $-Q(z)\mathrm dz^2$.

The following theorem is a specialization to $V(z;t)$ of  \cite[Theorems~2.3 and~2.4]{KS}.

\begin{theorem}
Let $V(z;t)$ be given by \eqref{cm2}.
\label{fundamental} 
\begin{enumerate}
  \item There exists a contour $\Ga_t\in\mathcal T$ such that
\begin{equation}
\label{em3}
\mathcal E_V(\Ga_t)=\sup_{\Ga\in\mathcal T} \mathcal E_V(\Ga).
\end{equation}
  \item  The support $J_t$ of the equilibrium measure $\mu_t:=\mu_{\Ga_t}$ has the S-property in the external field $\re V(z;t)$ and the measure \( \mu_t \) is uniquely determined by \eqref{em4}. Therefore, \( \mu_t \) and its support \( J_t \) are the same for every $\Ga_t$ satisfying \eqref{em3}.
  \item The function
\begin{equation}
\label{em5}
Q(z;t)=\left(\frac{V'(z;t)}{2}- \int \frac{\mathrm d\mu_t(s)}{z-s}\right)^2,\quad z\in \C\setminus J_t,
\end{equation}
is a polynomial of degree 4.
\item The support $J_t$ consists of some short critical trajectories of the quadratic differential $-Q(z;t)\mathrm dz^2$ and the equation
\begin{equation}
\label{em6}
\mathrm d\mu_t(z)=-\frac{1}{\pi \ic}\,Q_+^{1/2}(z;t)\mathrm dz, \quad z\in J_t,
\end{equation}
holds on each such critical trajectory, where $Q^{1/2}(z;t)=\frac12z^2+\mathcal{O}(z)$ as $z\to\infty$ (in what follows, \( Q^{1/2}(z;t) \) will always stand for such a branch).
\end{enumerate}  
\end{theorem}

Much information on the structure of the critical graphs of a quadratic differential can be found in the excellent monographs \cite{Jenkins,Pommerenke,Strebel}. Since $\deg\, Q=4$, $J_t$ consists of one or two arcs, corresponding (respectively) to the cases where $Q(z;t)$ has two simple zeros and one double zero, and the case where it has four simple zeros. Away from \( J_t \), one has freedom in choosing \( \Ga_t \). In particular, let
\begin{equation}
\label{em0}
\mathcal U(z;t) := \mathrm{Re}\left(2\int_e^z Q^{1/2}(s;t)\dd s \right) = \ell_{\Ga_t} - \mathrm{Re}(V(z;t)) - 2U^{\mu_t}(z),
\end{equation}
where \( e\in J_t \) is any and the second equality follows from \eqref{em5} (since the constant \( \ell_{\Gamma_t} \) in \eqref{em2} is the same for both connected components of \( J_t \) and the integrand is purely imaginary on \( J_t \), the choice of \( e \) is indeed not important). Clearly, \( \mathcal U(z;t) \) is a subharmonic function (harmonic away from \( J_t \)) which is equal to zero \( J_t \) by \eqref{em2}. The trajectories of \( -Q(z;t)\dd z^2 \) emanating out of the endpoints of \( J_t \) belong to the set \( \{z:\mathcal U(z;t)=0 \} \) and it follows from the variational condition \eqref{em2} that \( \Ga_t\setminus J_t\subset \{z:\mathcal U(z;t)<0\} \). However, within the region \( \{z:\mathcal U(z;t)<0 \} \) the set \( \Ga_t\setminus J_t \) can be varied freely. The geometry of the set \( \{z:\mathcal U(z;t)<0 \} \) is described further below in Theorems~\ref{geometry1} and~\ref{geometry2}.

\section{Structure of $\Ga_t$}\label{s:Structure}

The partition of the phase space into one- and two-cut regions as well as the structure of $\Ga_t$ and its dependence on $t$ has been heuristically described in \cite{AMAM1,AMAM2}. Even before these works, the phase diagram in Figure \ref{fig:C-curves}  was heuristically described in \cite{Dav1}. Mathematically rigorous description was provided in \cite{BlDeaY17}, but only in the one-cut region. Let us quickly recall the important notions from \cite{BlDeaY17}.

Denote by $\mathcal C$ the critical graph of an auxiliary quadratic differential
\begin{equation}
\label{aux-d}
-(1+1/s)^3 \mathrm ds^2,
\end{equation}
see Figure~\ref{fig:loops}(a). It was shown in \cite[Section~5]{BlDeaY17} that $\mathcal C$ consists of 5 critical trajectories emanating from $-1$ at the angles $2\pi k/5$, $k\in\{0,1,2,3,4\}$, one of them being $(-1,0)$,
\begin{figure}[ht!]
\centering
\subfigure[]{\includegraphics[scale=.5]{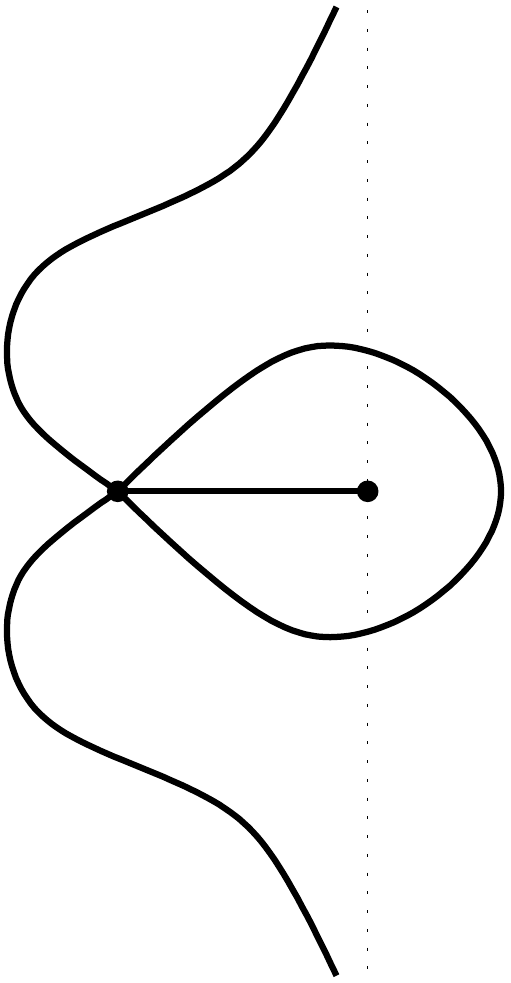}}
\begin{picture}(0,0)
\put(-23,62){$0$}
\put(-70,60){$-1$}
\end{picture}
\quad\quad
\subfigure[]{\includegraphics[scale=.8]{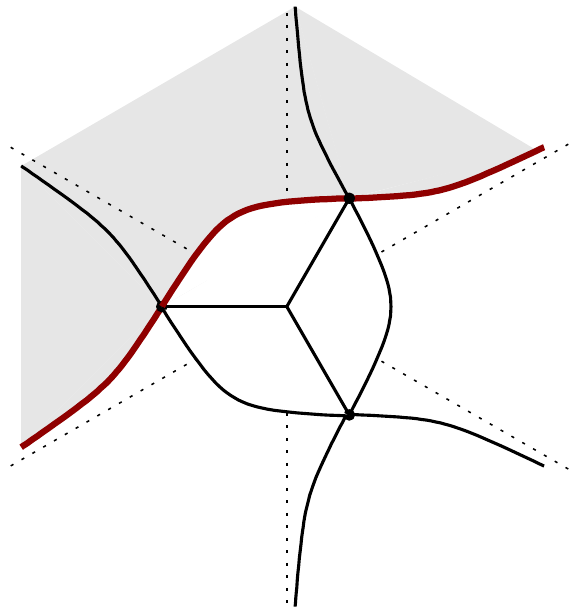}}
\begin{picture}(0,0)
\put(-142,65){$-\sqrt[3]{1/2}$}
\put(-35,90){$\Delta^a_\mathsf{birth}$}
\put(-90,80){$\Delta_\mathsf{split}$}
\put(-120,40){$\Delta^b_\mathsf{birth}$}
\put(-110,105){$\Omega_\mathsf{one-cut}$}
\end{picture}
\caption{\small Schematic representation of (a)  the critical graph $\mathcal C$; (b) the set $\Delta$ (solid lines) and the domain $\Omega_\mathsf{one-cut}$ (shaded region).}
\label{fig:loops}
\end{figure}
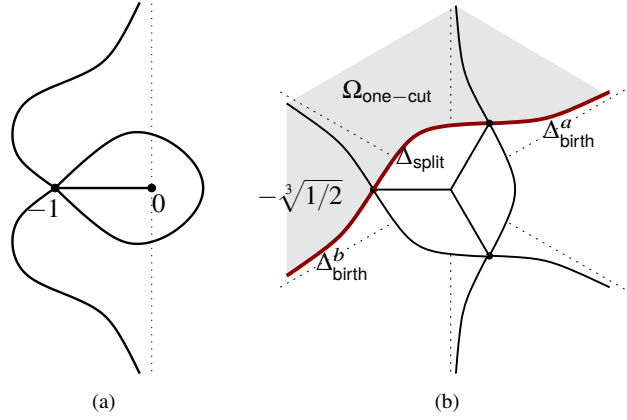
 other two forming a loop crossing the real line approximately at $0.635$, and the last two approaching infinity along the imaginary axis without changing the half-plane (upper or lower). Given $\mathcal C$, define
\[
\Delta:=\big\{x:~2x^3\in\mathcal C\big\}.
\]
Further, put $\Omega_\mathsf{one-cut}$ to be the shaded region on Figure~\ref{fig:loops}(b) and set
\[
\partial \Omega_\mathsf{one-cut} = \Delta^b_\mathsf{birth}\cup \big\{-2^{-1/3}\big\} \cup \Delta_\mathsf{split} \cup \big\{e^{\pi\mathrm i/3}2^{-1/3}\big\} \cup \Delta^a_\mathsf{birth},
\]
where $\Delta_\mathsf{split}$ connects $-2^{-1/3}$ and $e^{\pi\mathrm i/3}2^{-1/3}$, $\Delta^b_\mathsf{birth}$ extends to infinity in the direction of the angle $7\pi/6$ while $\Delta^a_\mathsf{birth}$ extends to infinity in the direction of the angle $\pi/6$. Let
\[
t(x) := (x^3-1)/x
\]
and set
\begin{equation}
\label{ts5}
\left\{
\begin{array}{l}
t_\mathsf{cr} := 3\cdot2^{-2/3} = t\big(-2^{-1/3}\big), \medskip \\ 
O_\mathsf{one-cut}:=t(\Omega_\mathsf{one-cut}), \medskip \\
C_\mathsf{split} := t\big(\Delta_\mathsf{split}\big), \quad C^b_\mathsf{birth} := t\big(\Delta^b_\mathsf{birth}\big), \quad C^a_\mathsf{birth} := t\big(\Delta^a_\mathsf{birth}\big), \medskip \\
S:=(t_\mathsf{cr},\infty), \quad e^{2\pi\mathrm i/3}S := \big\{z:~e^{-2\pi\mathrm i/3}z\in S\big\},
\end{array}
\right.
\end{equation}
see Figure~\ref{fig:C-curves}. 
\begin{figure}[ht!]
\centering
\includegraphics[scale=1.2]{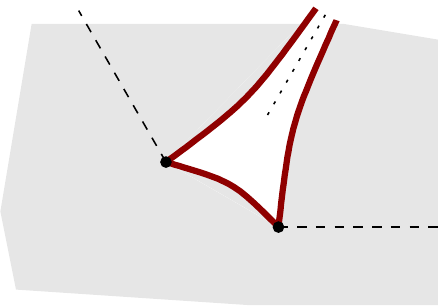}
\begin{picture}(0,0)
\put(-64,18){$t_\mathsf{cr}$}
\put(-140,47){$e^{2\pi\mathrm i/3}t_\mathsf{cr}$}
\put(-53,55){$C_\mathsf{birth}^b$}
\put(-95,77){$C_\mathsf{birth}^a$}
\put(-92,33){$C_\mathsf{split}$}
\put(-20,18){$S$}
\put(-153,80){$e^{2\pi\mathrm i/3}S$}
\put(-135,15){$O_\mathsf{one-cut}$}
\end{picture}
\caption{\small Domain $O_\mathsf{one-cut}$ (shaded region); $\partial O_\mathsf{one-cut}$ consisting of the open bounded arc $C_\mathsf{split}$, two open semi-unbounded arcs $C_\mathsf{birth}^a$ and $C_\mathsf{birth}^b$, and two points $t_\mathsf{cr}$ and $e^{2\pi\mathrm i/3}t_\mathsf{cr}$; the semi-unbounded open horizontal rays $S$ and $e^{2\pi\mathrm i/3}S$ (dashed lines).}
\label{fig:C-curves}
\end{figure}
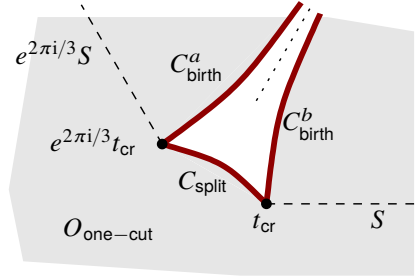
The function $t(x)$ is holomorphic in $\Omega_\mathsf{one-cut}$ with non-vanishing derivative there. It maps $\Omega_\mathsf{one-cut}$ onto $O_\mathsf{one-cut}$ in a one-to-one fashion. Hence, the inverse map $x(t)$ exists and is holomorphic. 
\begin{notation}
\label{Gamma-arcs}
Below, we adopt the following convention: $\Ga(z_1,z_2)$ (resp. $\Ga[z_1,z_2]$) stands for the trajectory or orthogonal trajectory (resp. the closure of) of the differential \( -Q(z;t)\dd z^2 \) connecting $z_1$ and $z_2$, oriented from $z_1$ to $z_2$, and $\Ga\big(z,e^{\mathrm i\theta}\infty\big)$ (resp. $\Ga\big(e^{\mathrm i\theta}\infty,z\big)$) stands for the orthogonal trajectory ending at $z$, approaching infinity at the angle $\theta$, and oriented away from $z$ (resp. oriented towards $z$).\footnote{This notation is unambiguous as the corresponding trajectories are unique for polynomial differentials as follows from Teichm\"uller's lemma.}
\end{notation}
\begin{figure}[ht!]
\subfigure[]{\includegraphics[scale=.2]{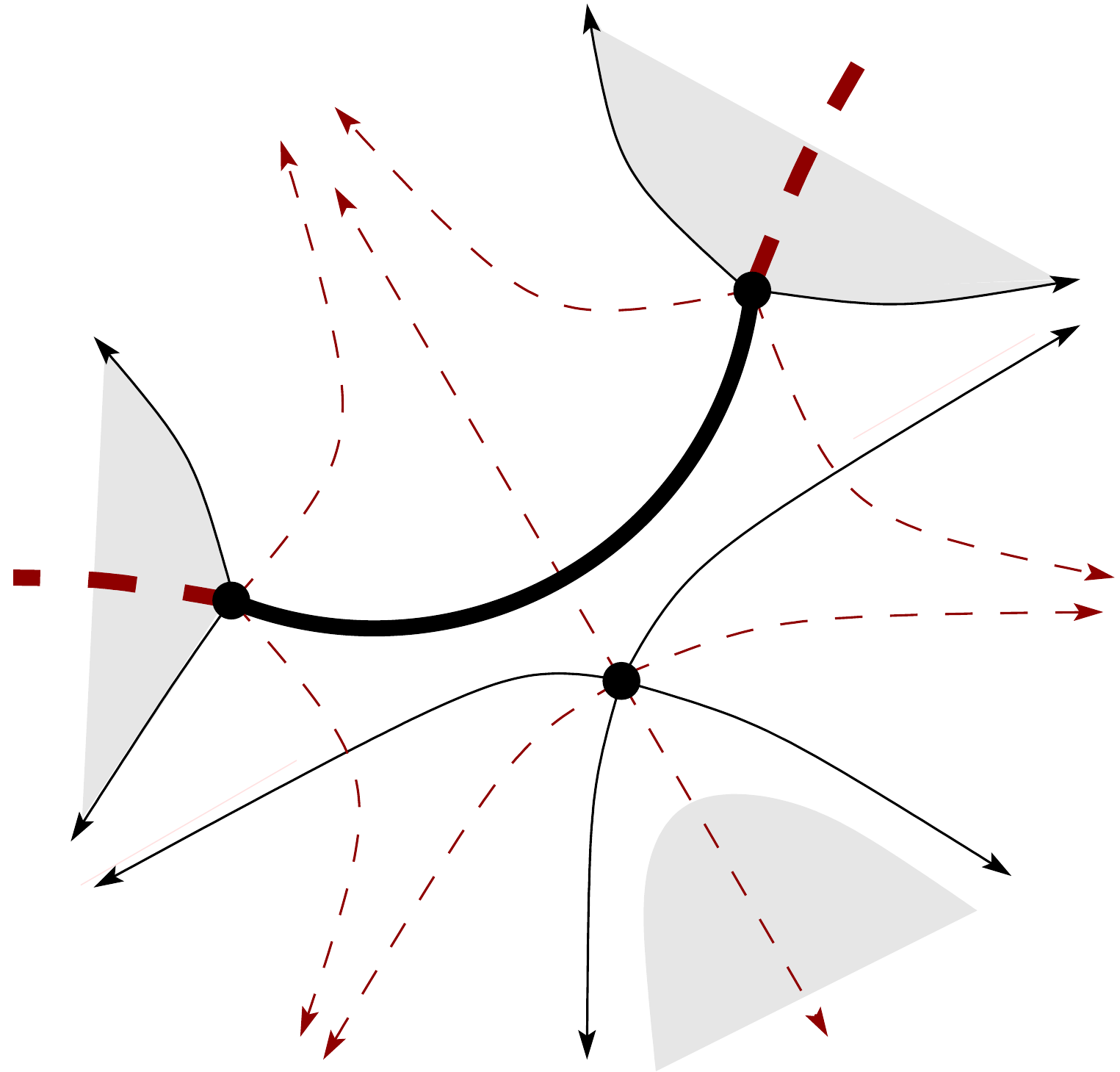}}
\subfigure[]{\includegraphics[scale=.2]{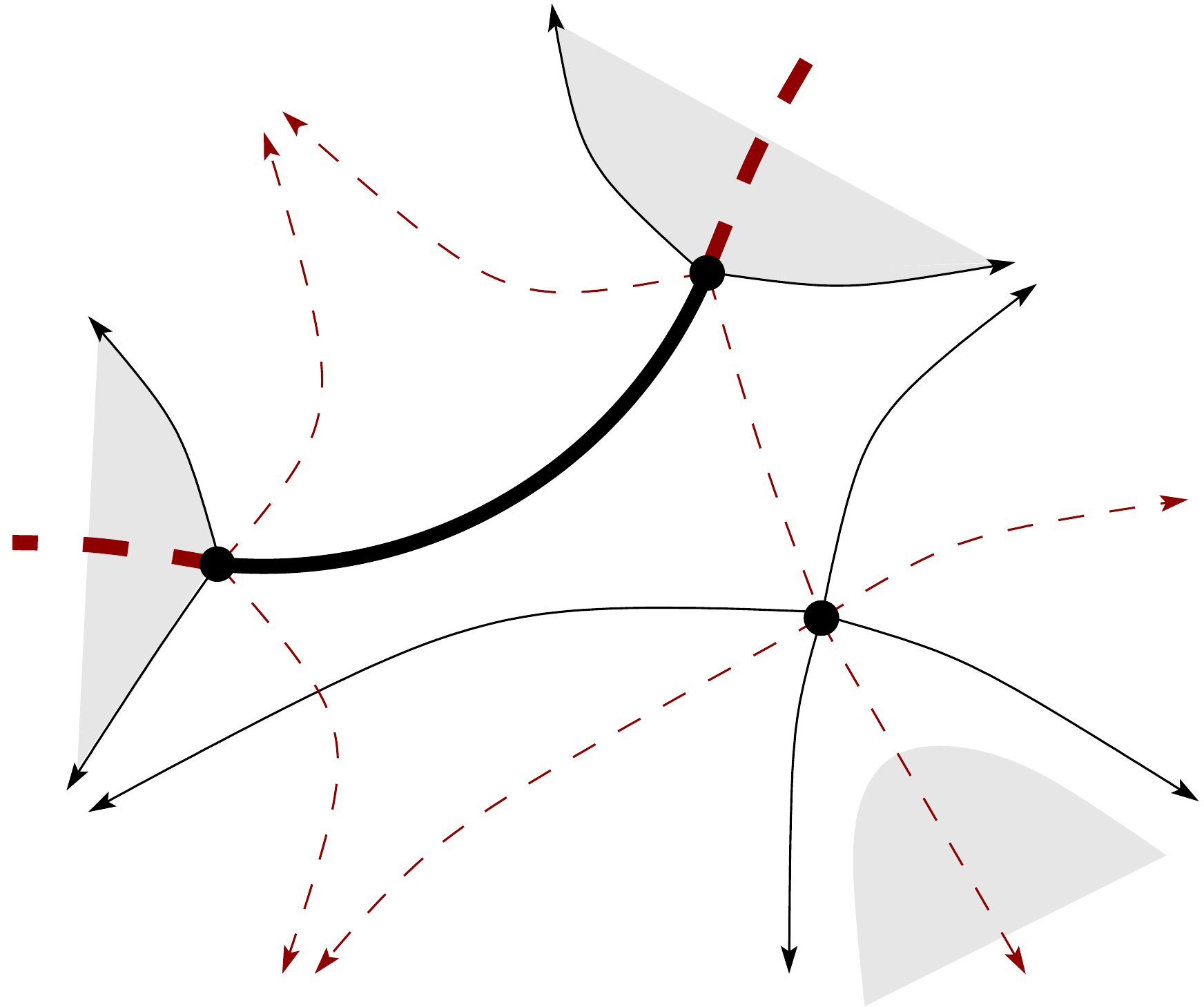}}
\subfigure[]{\includegraphics[scale=.19]{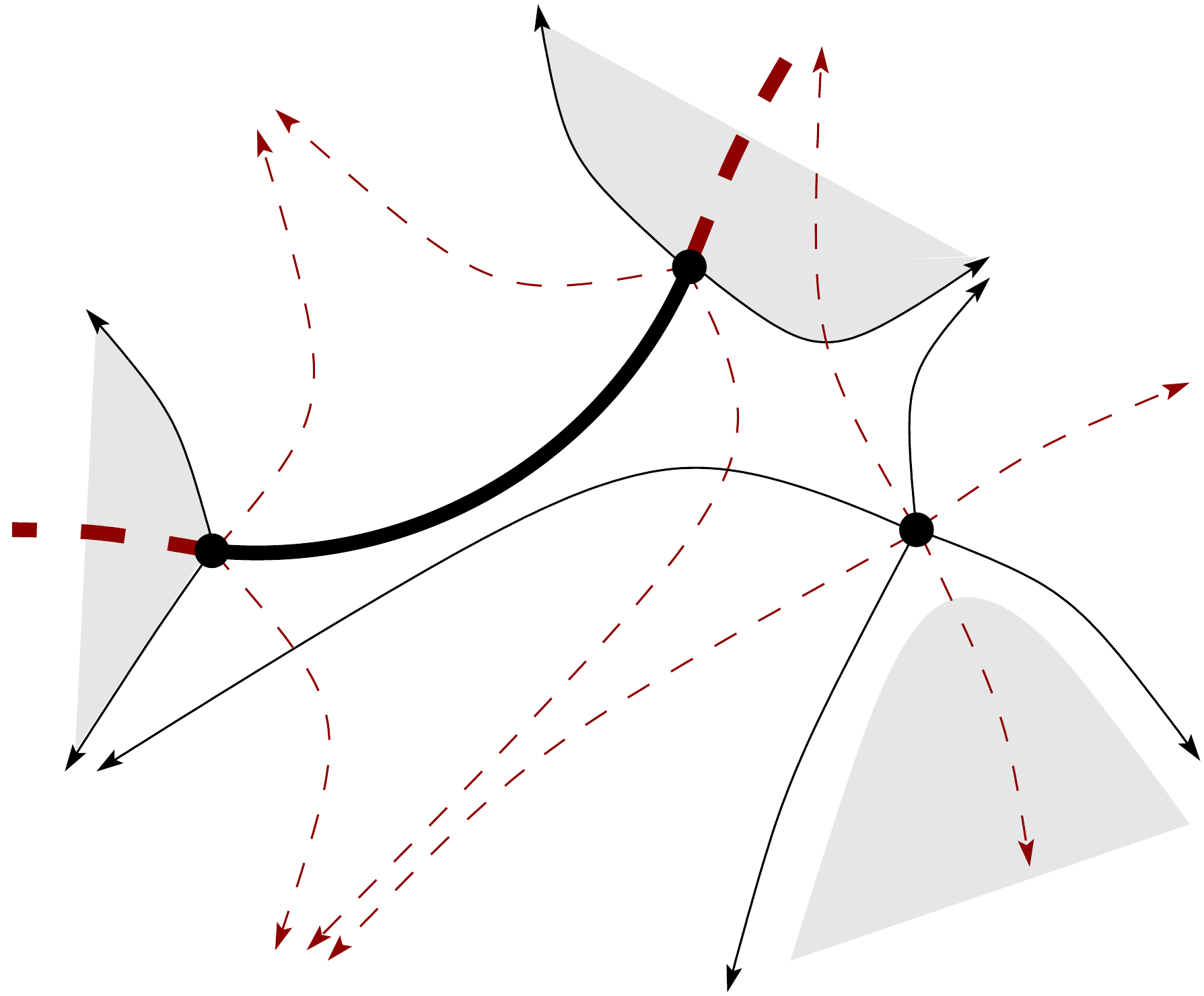}}
\subfigure[]{\includegraphics[scale=.18]{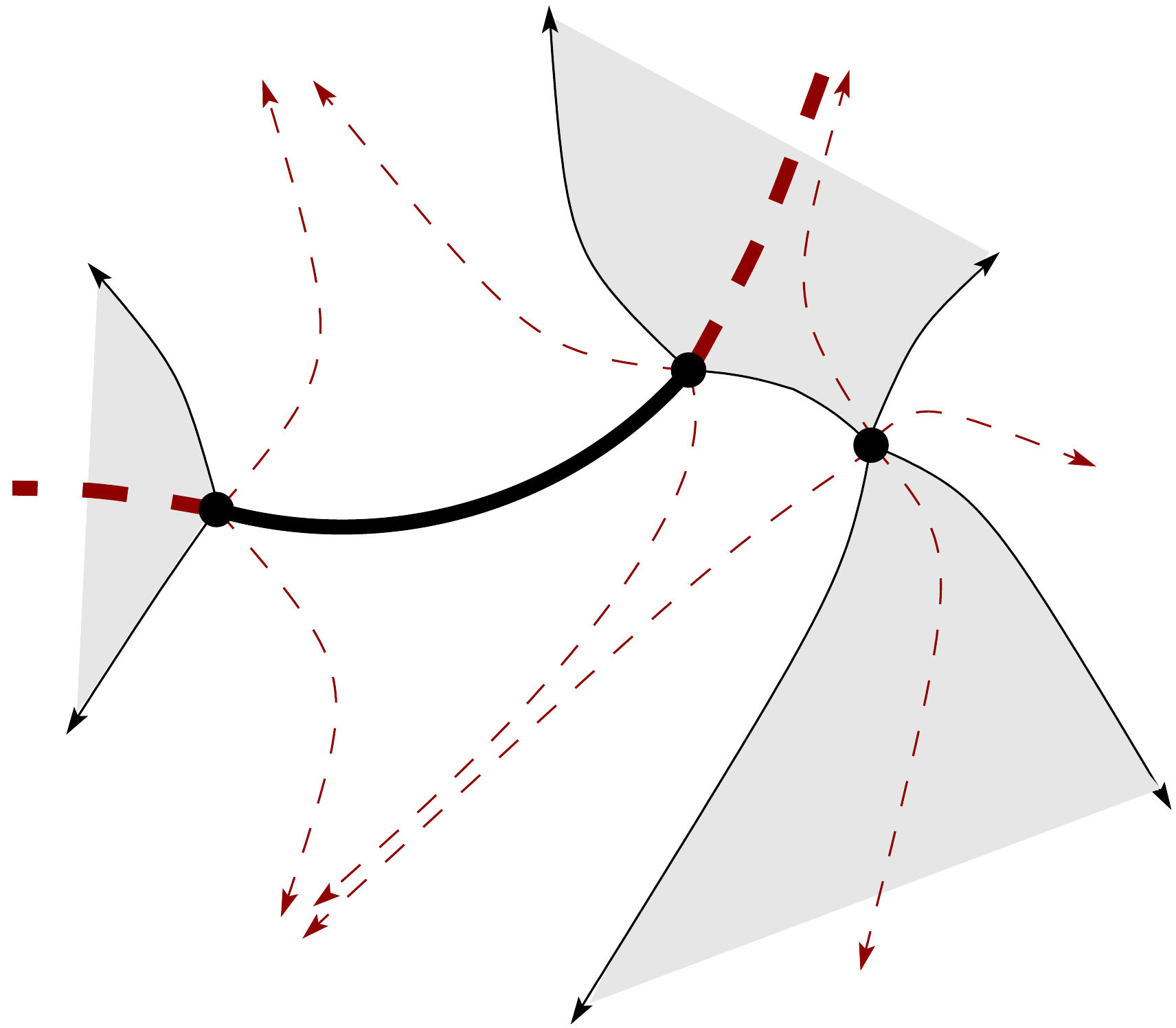}}\newline
\subfigure[]{\includegraphics[scale=.21]{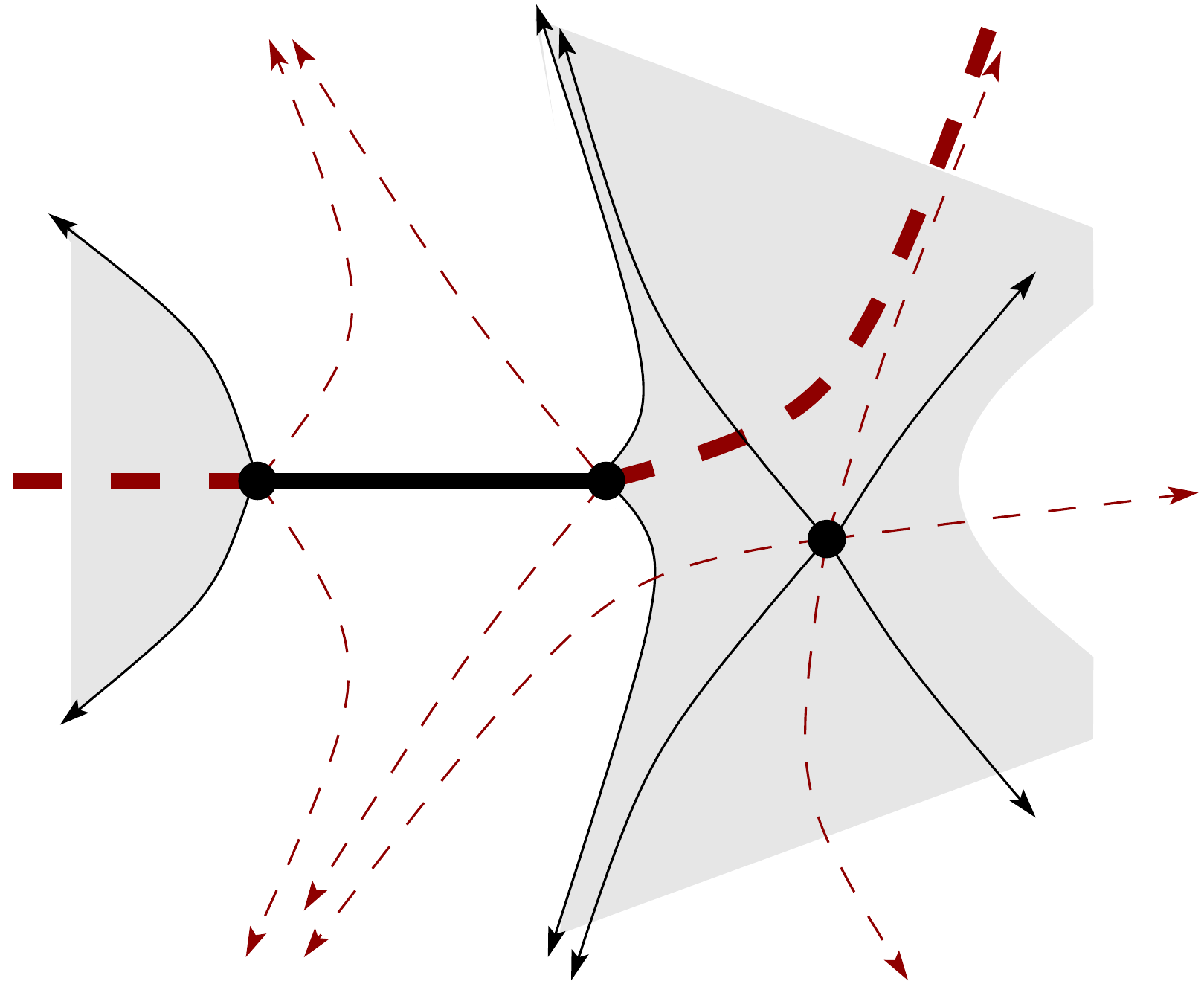}}
\subfigure[]{\includegraphics[scale=.21]{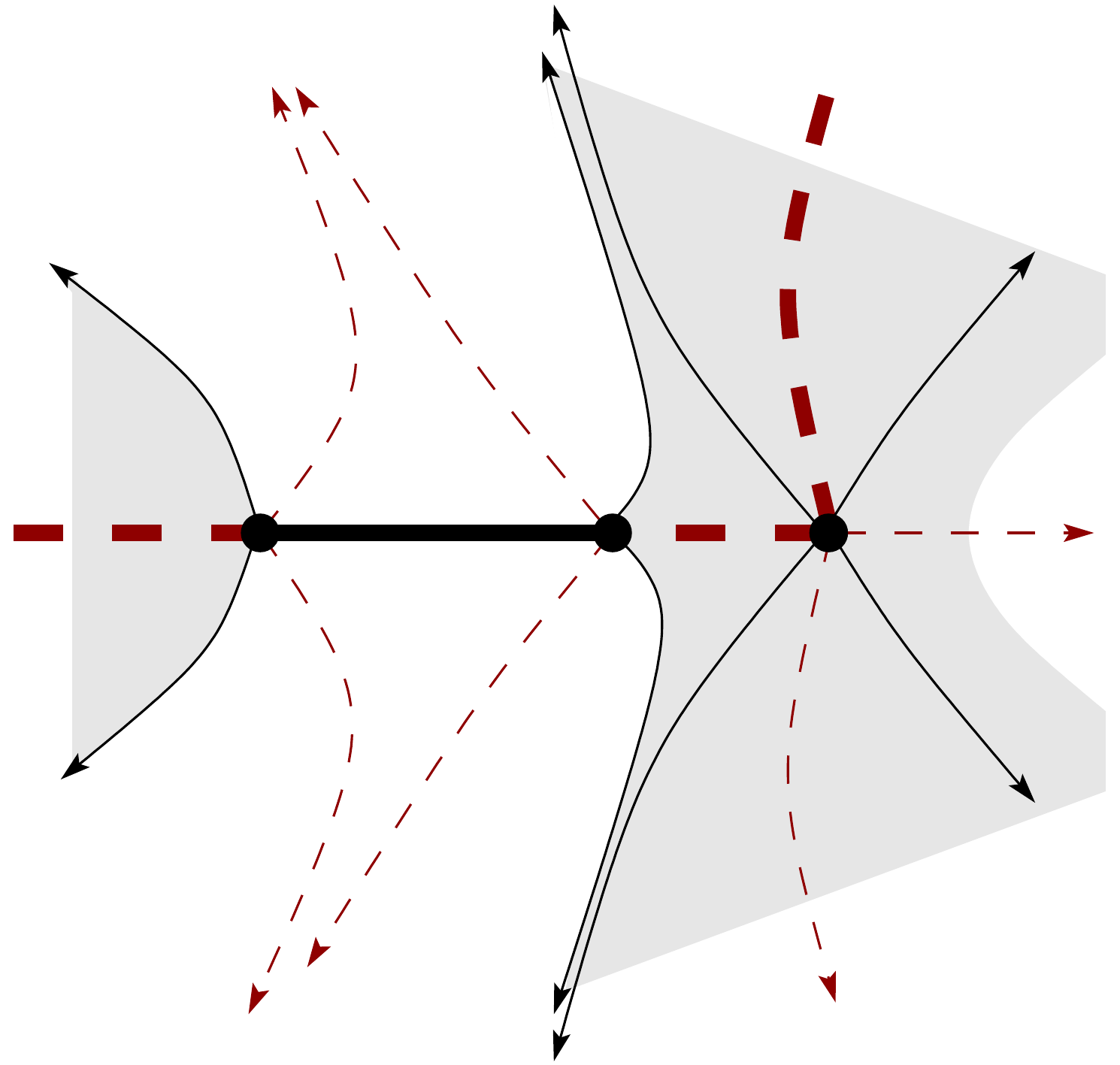}}
\subfigure[]{\includegraphics[scale=.21]{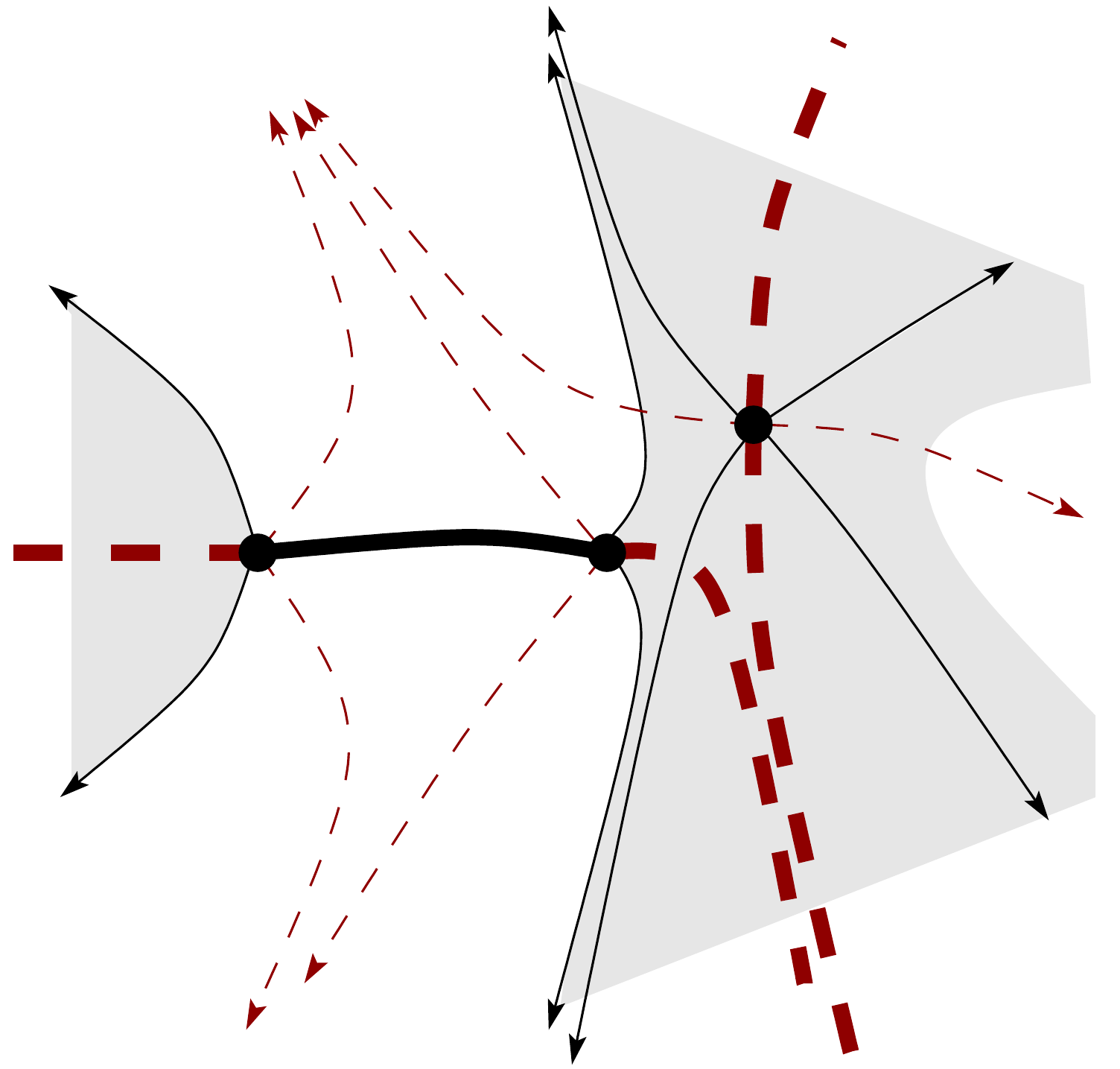}}
\caption{\small Schematic representation of the critical (solid) and critical orthogonal (dashed) graphs of $-Q(z;t)\mathrm dz^2$ when $t\in O_\mathsf{one-cut}$. The bold curves represent the preferred S-curve $\Ga_t$. Shaded region is the set $\{\mathcal U(z;t)<0\}$.}
 \label{s-curves1}
 \end{figure} 

The following theorem has been proven in \cite[Theorem~3.2]{BlDeaY17} and it describes the geometry of \( \Ga_t \) when $t\in\overline{O}_\mathsf{one-cut}$.

\begin{theorem}
\label{geometry1}
Let $\mu_t$ and $Q(z;t)$ be as in Theorem~\ref{fundamental}, and let $J_t=\mathsf{supp}(\mu_t)$. When $t\in\overline{O}_\mathsf{one-cut}$, the polynomial $Q(z;t)$ is of the form
\begin{equation}
\label{ts1}
Q(z;t)=\frac14(z-a(t))(z-b(t))(z-c(t))^2
\end{equation}
with $a(t)$, $b(t)$, and $c(t)$  given by
\begin{equation}
\label{ts6}
\left\{
\begin{array}{lll}
a(t) &:=& x(t)-\mathrm i\sqrt2/\sqrt x(t), \medskip \\
b(t) &:=& x(t)+\mathrm i\sqrt2/\sqrt x(t), \medskip \\
c(t) &:=& -x(t),
\end{array}
\right.
\end{equation}
where $\sqrt x(t)$  is the branch holomorphic in $O_\mathsf{one-cut}$ satisfying $\sqrt x(0)=e^{\pi\mathrm i/3}$. The set $J_t$ consists of a single arc and
\begin{itemize}
\item[(I)] if $t\in O_\mathsf{one-cut}$, then $J_t=\Ga[a,b]$ and an S-curve $\Ga_t\in\mathcal T$ can be chosen as
\begin{itemize}
\smallskip
\item[(a)] $\Ga\big(e^{\pi\mathrm i}\infty,a\big)\cup J_t \cup \Ga\big(b,e^{\pi\mathrm i/3}\infty\big)$ when $t$ belongs to the connected component bounded by $S\cup C_\mathsf{split}\cup e^{2\pi\mathrm i/3}S$, see Figure~\ref{s-curves1}(a--e);
\smallskip
\item[(b)] $\Ga\big(e^{\pi\mathrm i}\infty,a\big) \cup J_t\cup \Ga(b,c) \cup \Ga\big(c,e^{\pi\mathrm i/3}\infty\big)$ when $t\in S$, see Figure~\ref{s-curves1}(f);
\smallskip
\item[(c)] $\Ga\big(e^{\pi\mathrm i}\infty,c\big) \cup \Ga(c,a) \cup J_t \cup  \Ga\big(b,e^{\pi\mathrm i/3}\infty\big)$ when $t\in e^{2\pi\mathrm i/3}S$;
\smallskip
\item[(d)] $\Ga\big(e^{\pi\mathrm i}\infty,a\big) \cup J_t\cup \Ga\big(b,e^{-\pi\mathrm i/3}\infty\big) \cup \Ga\big(e^{-\pi\mathrm i/3}\infty,c\big) \cup \Ga\big(c,e^{\pi\mathrm i/3}\infty\big)$ when $t$ belongs to the connected component bounded by $S\cup C_\mathsf{birth}^b$, see Figure~\ref{s-curves1}(g);
\smallskip
\item[(e)] $\Ga\big(e^{\pi\mathrm i}\infty,c\big) \cup \Ga\big(c,e^{-\pi\mathrm i/3}\big) \cup \Ga\big(e^{-\pi\mathrm i/3}\infty,a\big) \cup J_t \cup \Ga\big(b,e^{\pi\mathrm i/3}\infty\big)$ when $t$ belongs to the connected component bounded by $e^{2\pi\mathrm i/3}S\cup C_\mathsf{birth}^a$.
\end{itemize}
\item[(II)] if $t=t_\mathsf{cr}$ (resp. $t=e^{2\pi i/3}t_\mathsf{cr})$, then $J_t=\Ga[a,b]$, $c$ coincides with $b$ (resp. $a$), and an S-curve $\Ga_t\in\mathcal T$ can be chosen as in Case I(a), see Figure~\ref{s-curves2}(a).
\item[(III)] if $t\in C_\mathsf{split}$, then $J_t=\Ga[a,c]\cup\Ga[c,b]$ and an S-curve $\Ga_t\in\mathcal T$ can be chosen as in Case I(a), see Figure~\ref{s-curves2}(b).
\item[(IV)] if  $t\in C_\mathsf{birth}^b$ (resp. $t\in C_\mathsf{birth}^a$), then $J_t=\Ga[a,b]$ and an S-curve $\Ga_t\in\mathcal T$ can be chosen as in Case I(d) (resp. Case I(e)), see Figure~\ref{s-curves2}(c).
\end{itemize}
\end{theorem}

\begin{figure}[ht!]
\subfigure[]{\includegraphics[scale=.23]{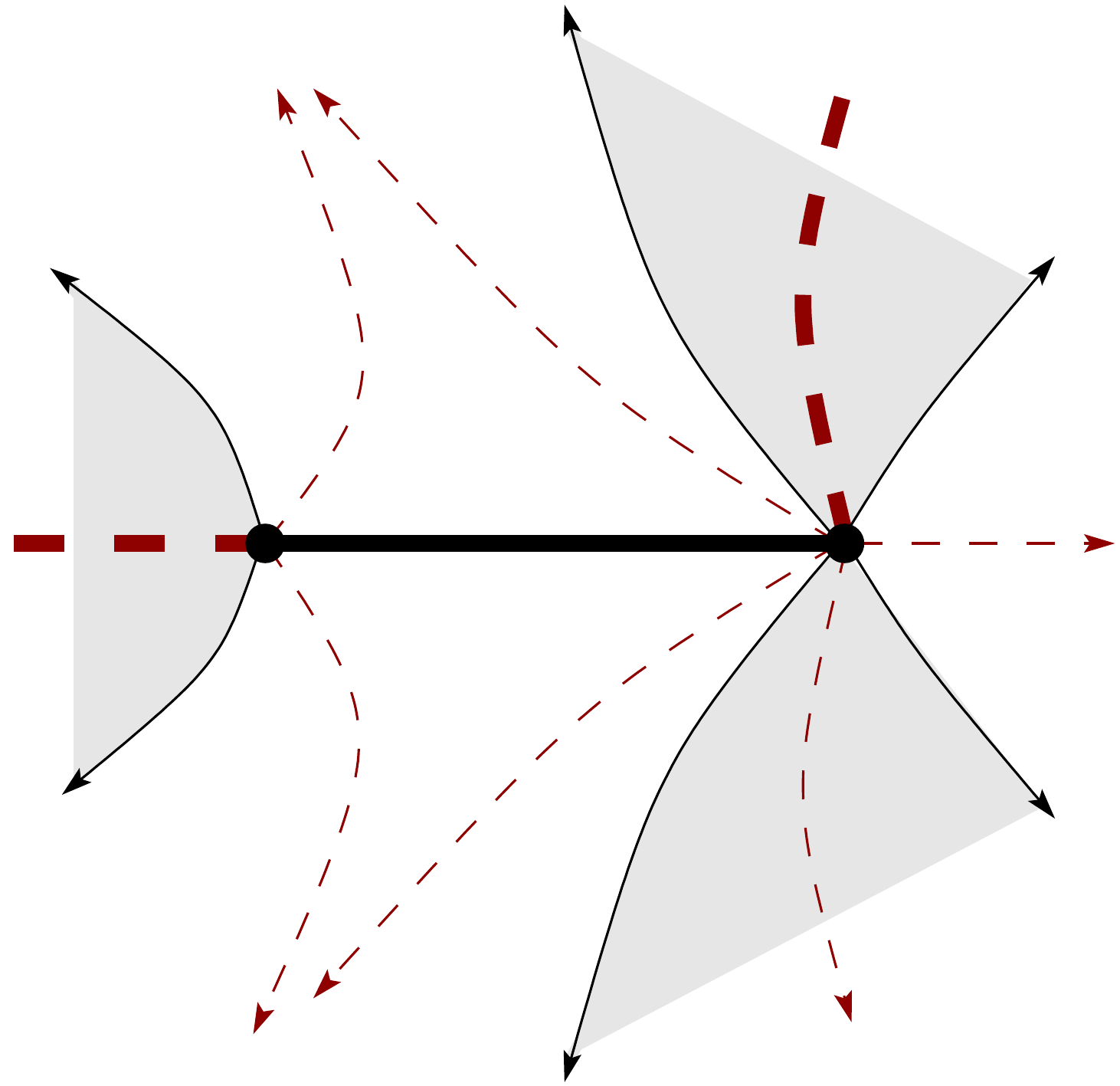}}
\subfigure[]{\includegraphics[scale=.23]{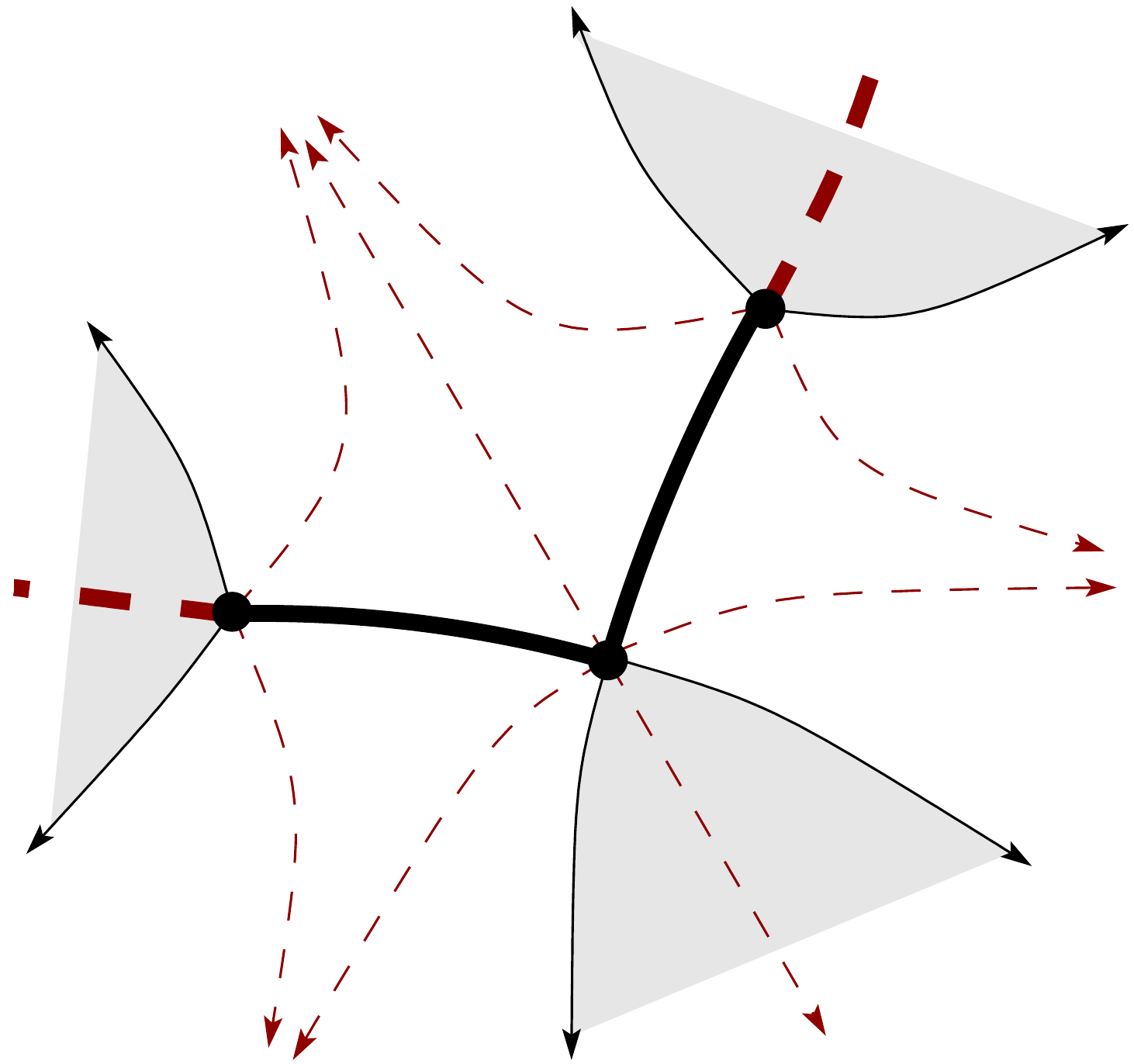}}
\subfigure[]{\includegraphics[scale=.2]{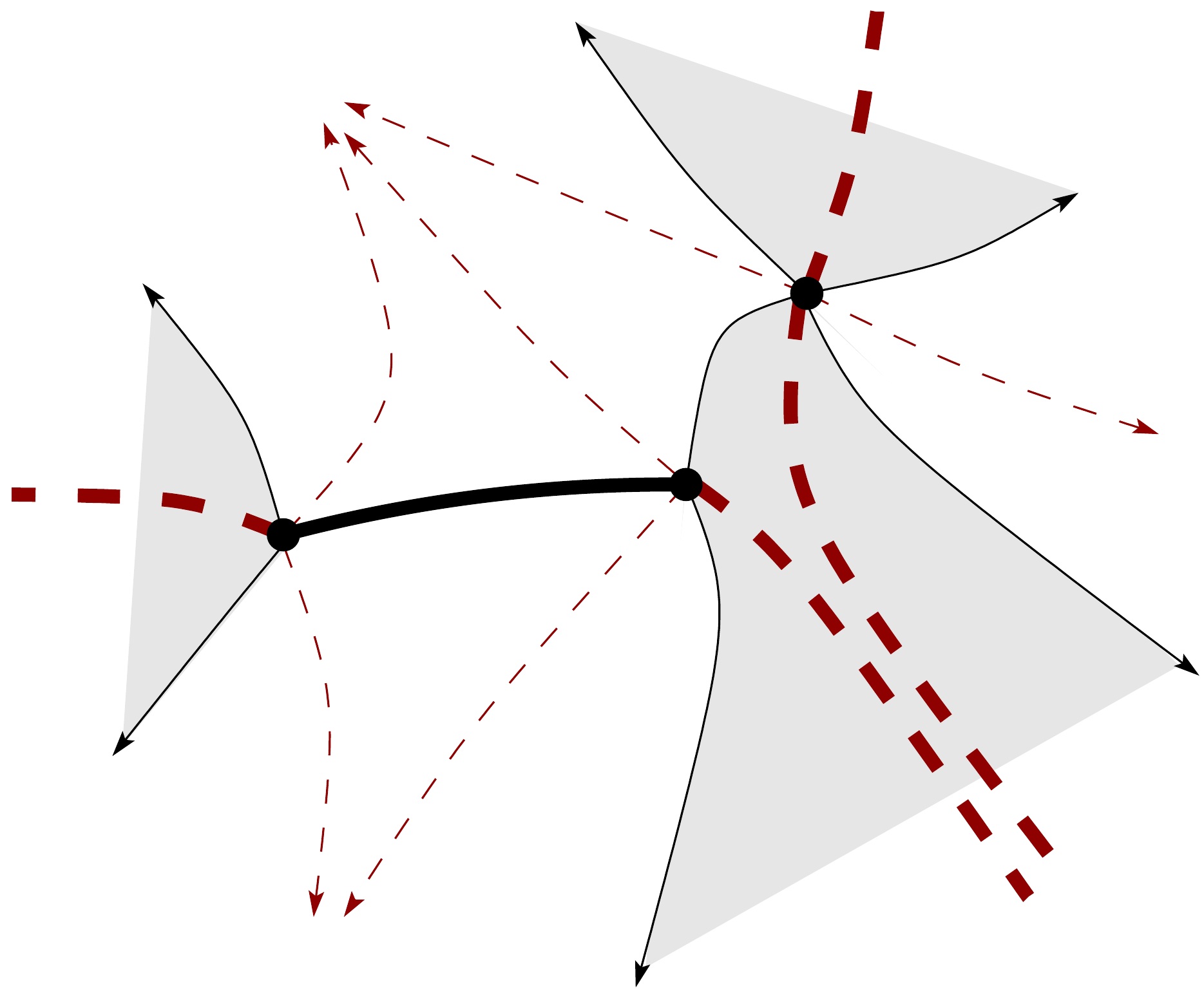}}
\caption{\small This is a continuation of Figure~\ref{s-curves1} for the case $t\in\partial O_\mathsf{one-cut}$.}
 \label{s-curves2}
 \end{figure}

Now, let \( O_\mathsf{two-cut} := \C\setminus \overline O_\mathsf{one-cut} \). Then the following theorem holds.

\begin{theorem}
\label{geometry2}
Let $\mu_t$ and $Q(z;t)$ be as in Theorem~\ref{fundamental}, $J_t=\mathsf{supp}(\mu_t)$. When $t\in O_\mathsf{two-cut}$, the polynomial $Q(z;t)$ is of the form
\begin{equation}
\label{ts8}
Q(z;t)=\frac14(z-a_1(t))(z-b_1(t))(z-a_2(t))(z-b_2(t))
\end{equation}
with $a_1(t)$, $b_1(t)$, \( a_2(t) \), and $b_2(t)$ all distinct. The real and imaginary parts of \( a_i(t),b_i(t) \) are real analytic functions of \( \re(t) \) and \( \im(t) \) when \( t \in O_\mathsf{two-cut} \); however, at no point of \( O_\mathsf{two-cut} \) any of the functions  \( a_i(t),b_i(t) \) is analytic. The S-curve \( \Gamma_t \) can be chosen as
\[
\Ga\big(e^{\pi\ic}\infty,a_1(t)\big) \cup J_{t,1} \cup \Ga\big(b_1(t),e^{-\pi\ic/3}\infty\big) \cup \Ga\big(e^{-\pi\ic/3}\infty,a_2(t)\big) \cup J_{t,2} \cup \Ga\big(b_2(t),e^{\pi\ic/3}\infty\big), 
\]
where $J_t=J_{t,1}\cup J_{t,2}$ and \( J_{t,i} := \Gamma\big[a_i(t),b_i(t)\big] \), \( i\in\{1,2\} \), see Figure~\ref{fig:tc} (this also explains how we choose the labeling of the zeros of \( Q(z;t) \) in the considered case). Moreover, it holds that
\begin{equation}
\label{ts8a}
a_1(t),b_1(t)\to a(t^*), \quad b_1(t),a_2(t) \to c(t^*), \qandq a_2(t),b_2(t) \to b(t^*)
\end{equation}
as \( t\to t^* \) with \( t^*\in C_\mathsf{birth}^a \cup \big\{e^{2\pi i/3}t_\mathsf{cr}\} \), \( t^*\in C_\mathsf{split} \), and \( t^*\in C_\mathsf{birth}^b \cup \big\{t_\mathsf{cr}\} \), respectively.
\end{theorem}

We prove Theorem~\ref{geometry2} in Section~\ref{s:pr-geom}.

\begin{figure}[ht!]
\includegraphics[scale=.25]{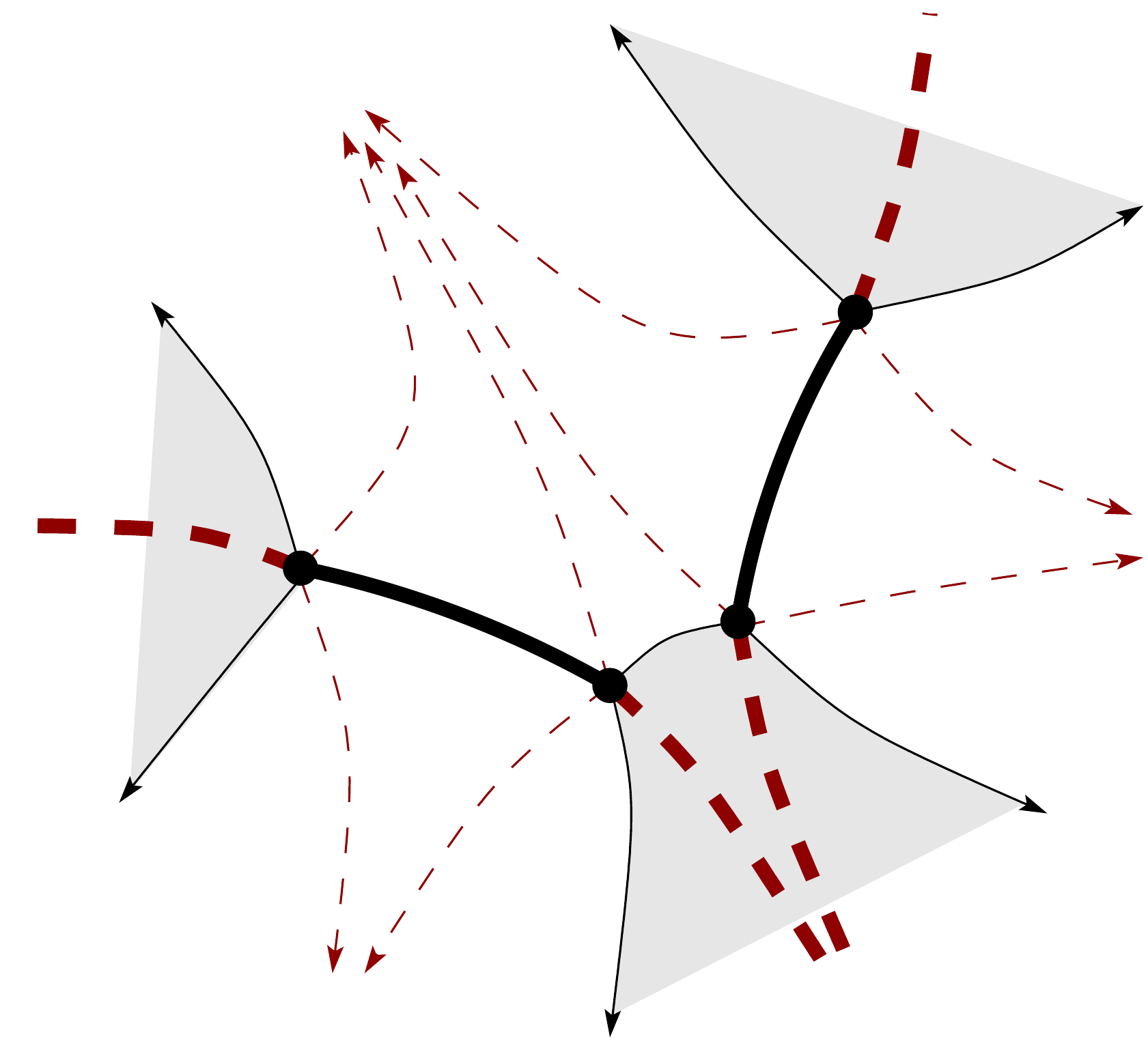}
\begin{picture}(0,0)
\put(-98,37){$a_1$}
\put(-68,25){$b_1$}
\put(-40,40){$a_2$}
\put(-28,72){$b_2$}
\end{picture}
\caption{\small The schematic representation of the critical and critical orthogonal graphs of \( -Q(z;t)\dd z^2 \) when $t\in  O_\mathsf{two-cut}$.  The bold curves represent the preferred S-curve $\Ga_t$. Shaded region is the set $\{\mathcal U(z;t)<0\}$.}
 \label{fig:tc}
 \end{figure}

\section{Main Results}
\label{s:main}

In this section we assume that \( t\in O_\mathsf{two-cut} \) and that \( Q(z;t) \) and \( J_t= J_{t,1} \cup J_{t,2} \) are as in Theorem~\ref{geometry2}. We also put
\[
E_t :=\big\{a_1(t),b_1(t),a_2(t),b_2(t) \big\}, \quad J_{t,i}^\circ := J_{t,i}\setminus E_t, \quad \text{and} \quad J_t^\circ := J_{t,1}^\circ \cup J_{t,2}^\circ.
\] 
\begin{notation}
\label{i-t}
When it comes to the definition of the S-contour \( \Ga_t \), it will be more practical for us to change the choice of \( \Ga_t \) made in Theorem~\ref{geometry2}. Rather than including unbounded trajectories \( \Ga\big(b_1(t),e^{-\pi\ic/3}\infty\big) \) and \( \Ga\big(e^{-\pi\ic/3}\infty,a_2(t)\big) \), we shall replace them by a smooth Jordan arc, say \( I_t \), connecting \( b_1(t) \) and \( a_2(t) \), such that \( I_t^\circ := I_t\setminus E_t \) lies entirely within the region \( \{\mathcal U(z;t)<0\} \) while there exist \( s_1(t)\in\Ga\big(b_1(t),e^{-\pi\ic/3}\infty\big) \) and \( s_2(t)\in\Ga\big(e^{-\pi\ic/3}\infty,a_2(t)\big) \) for which \( \Ga\big(b_1(t),s_1(t)\big)\cup\Ga\big(s_2(t),a_2(t)\big)\subset I_t^\circ \), see Figure~\ref{lens} further below.
\end{notation}
To describe the asymptotics of the orthogonal polynomials themselves, we need to construct the Szeg\H{o} function of \( e^{V(z;t)} \).

\begin{proposition}
\label{prop:Szego}
Let constants \( \varsigma(t) \) and \( d_i(t) \) be given by
\begin{equation}
\label{Ct}
\varsigma(t) := \frac{2t}{3d_0(t)} \quad \text{and} \quad d_i(t) := \int_{I_t} \dfrac{s^i \dd s }{Q^{1/2}(s; t)} ,
\end{equation}
where, as usual, we use the branch \( Q^{1/2}(z;t) = \frac12z^2 + \mathcal O(z) \) as \( z\to\infty \). Then the Szeg\H{o} function
\begin{equation}
\label{D}
\mathcal D(z;t) := \exp\left\{\frac12 V(z;t) + \frac13 \left(z + \int_{I_t}\frac{3\varsigma(t)}{s-z}\frac{\dd s}{Q^{1/2}(s;t)}\right) Q^{1/2}(z;t)\right\}
\end{equation}
is holomorphic and non-vanishing in \( \overline\C\setminus \Ga[a_1(t),b_2(t)] \) with continuous traces on \( J_t^\circ \cup I_t^\circ \) that satisfy
\begin{equation}
\label{Djumps}
\left\{
\begin{array}{ll}
\mathcal D_+(s;t)\mathcal D_-(s;t) = e^{V(s;t)}, & s\in J_t^\circ, \medskip \\
\mathcal D_+(s;t) = \mathcal D_-(s;t) e^{2\pi\ic\varsigma(t)}, & s\in I_t^\circ.
\end{array}
\right.
\end{equation}
Denote by $\mu_1(t):=\int s\dd\mu_t(s)$ the first moment of the equilibrium measure $\mu_t$. Then 
\begin{equation}
\label{D-exp}
\mathcal{D}(z;t) = \exp \left\{ \frac{1}{3} - \frac{(\varsigma d_1)(t)}{2} \right\} \left ( 1 + \dfrac{t^2 + \mu_1(t) - 3(\varsigma d_2)(t)/2}{3z} + \mathcal{O}\left(\frac{1}{z^2} \right) \right)
\end{equation}
as \( z\to \infty \). In particular, it holds that $\mathcal D(\infty;t) = \exp \left\{ \frac{1-t(d_1/d_0)(t)}{3}\right\}$.
\end{proposition}

We shall also denote by \( D(z;t):=\mathcal D(z;t)/\mathcal D(\infty;t) \) the normalized Szeg\H{o} function. We prove Proposition~\ref{prop:Szego} in Section~\ref{ss:Szego}. 

To describe the geometric growth of the orthogonal polynomials, let us define
\begin{equation}
\label{mQ}
\mathcal Q(z;t) := \int_{b_2(t)}^zQ^{1/2}(s;t)\dd s, \quad z\in\C\setminus\Ga\big(e^{\pi\ic}\infty,b_2(t)\big],
\end{equation}
(where we deviate slightly from Convention~\ref{Gamma-arcs} and denote by \( \Ga[s_1,s_2] \) a subarc of \( \Ga_t \) connecting \( s_1 \) to \( s_2 \); for example, \( \Ga(e^{\pi\ic}\infty,a_2] \) includes \( I_t \) rather than the short trajectories connecting \( b_1 \) to \( a_2 \)). Observe that \( \mathcal U(z;t) = 2\mathrm{Re}(\mathcal Q(z;t)) \) as defined in \eqref{em0}. This function has the following properties.

\begin{proposition}
\label{prop:g}
Let constants \( \tau(t),\omega(t) \) be given by
\begin{equation}
\label{tau-omega}
\tau(t) := -\frac1{\pi\ic}\int_{I_t}Q^{1/2}(s;t)\dd s \quad \text{and} \quad \omega(t) := -\frac1{\pi\ic}\int_{J_{t,1}}Q_+^{1/2}(s;t)\dd s.
\end{equation}
These constants are necessarily real (in fact, \( \omega(t)=\mu_t(J_{t,1}) \), see \eqref{em6}). The function \( e^{\mathcal Q(z;t)} \) is holomorphic in \( \C\setminus \Ga[a_1(t),b_2(t)] \) and there exists a constant \( \ell_*(t) \) such that
\begin{equation}
\label{oc5}
\exp\left\{\frac{V(z;t)-\ell_*(t)}2 + \mathcal Q(z;t)\right\} = z+ \mathcal{O}(1) \quad \text{as} \quad z\to\infty.
\end{equation}
Moreover, \( \mathcal Q(z;t) \) possesses continuous traces on  \(J_t^\circ \cup I_t^\circ \) that are purely imaginary on \( J_t^\circ \) and satisfy
\begin{equation}
\label{gjumps}
\left\{
\begin{array}{ll}
e^{\mathcal Q_+(s;t)+\mathcal Q_-(s;t)} = 1, & s\in J_{t,2}^\circ, \medskip \\
e^{\mathcal Q_+(s;t)+\mathcal Q_-(s;t)} = e^{2\pi\ic\tau(t)}, & s\in J_{t,1}^\circ, \medskip \\
e^{\mathcal Q_+(s;t)} = e^{\mathcal Q_-(s;t)-2\pi\ic\omega(t)}, & s\in I_t^\circ.
\end{array}
\right.
\end{equation}
\end{proposition}

We prove Proposition~\ref{prop:g} in Section~\ref{ss:g}. Observe that it follows from Theorem~\ref{geometry2} that \( |e^{\mathcal Q(z;t)}| \) is less than \( 1 \) when \( \mathcal U(z;t)<0 \) (the shaded areas of Figure~\ref{fig:tc}), is equal to \( 1 \) on the critical trajectories (black curves), and otherwise is greater than \( 1 \). The constant \( \ell_*(t) \) is such that $\re(\ell_*(t))=\ell_{\Ga_t}$, see \eqref{em2}. Moreover, the following identity holds.
\begin{proposition}
\label{prop:l}
Set
\begin{equation}
\label{BQ}
\mathsf B(t) := - \left(\int_{J_{t,1}}\frac{\dd s}{Q_+^{1/2}(s;t)}\right)/\left(\int_{I_t}\frac{\dd s}{Q^{1/2}(s;t)}\right).
\end{equation}
Let \( \theta(u;\mathsf B(t)) \) be the Riemann theta function associated with \( \mathsf B(t) \), i.e.,
\begin{equation}
\label{Rtheta}
\theta(u;\mathsf B(t)) = \sum_{k\in\Z}\exp\left\{\pi\ic \mathsf B(t)k^2+2\pi\ic uk\right\}, \quad u\in\C.
\end{equation}
Then it holds that
\begin{equation}
\label{el}
e^{\ell_*(t)} = \left(4\mathcal D(\infty;t)\frac{e^{\pi\ic\tau(t)(\varsigma(t)+\omega(t)+\mathsf B(t)\tau(t))}}{b_1(t)-a_1(t)+b_2(t)-a_2(t)}\frac{\theta(\varsigma(t)+\omega(t)+\mathsf B(t)\tau(t))}{\theta(0)}\right)^2.
\end{equation}
\end{proposition}
We prove this proposition in Section~\ref{ss:propl}. We also explain right after the statement of Proposition~\ref{prop:Nt} further below that \( b_1(t)-a_1(t)+b_2(t)-a_2(t) \neq 0\).

Another auxiliary function we need is given by
\begin{equation}
\label{alpha}
A(z;t) := \frac12\left(\left(\frac{z-b_2(t)}{z-a_2(t)}\frac{z-b_1(t)}{z-a_1(t)}\right)^{1/4}+\left(\frac{z-b_2(t)}{z-a_2(t)}\frac{z-b_1(t)}{z-a_1(t)}\right)^{-1/4}\right)
\end{equation}
for \( z\in \overline\C\setminus J_t \), where the branches are chosen so that the summands are holomorphic in \( \overline\C\setminus J_t \) and have value \( 1 \) at infinity. As explained in Section~\ref{ss:JIP}, this function can be analytically continued through each side of \( J_t^\circ \) and is non-vanishing in the domain of the definition.

Given a sequence \( \{N_n\}_{n\in \N} \), we define further below in \eqref{ab5} functions \( \Theta_n(z;t) \), which are ratios of Riemann theta functions \eqref{Rtheta} with various arguments. To shorten the presentation of the main results, we only discuss main features of the functions \( \Theta_n(z;t) \) and defer the detailed construction and description of further properties to Section~\ref{sec:An}.

\begin{proposition}
\label{prop:A}
Functions \( \Theta_n(z;t) \) are holomorphic in \( \overline\C\setminus\Ga[a_1(t),b_2(t)] \) with at most one zero there. They have continuous traces on \( I_t^\circ \) that satisfy
\[
\Theta_{n+}(s;t) = \Theta_{n-}(s;t) e^{2\pi\ic(n\omega(t)+(n-N_n)\varsigma(t))}.
\]
Assume that there exists a constant \( N_* \) such that \( |n-N_n|\leq N_* \) for all \( n\in \N \). Then for any \( \delta>0 \) there exists a constant \( C(t,\delta,N_*) \) such that
\[
|A(z;t)\Theta_n(z;t)| \leq C(t,\delta,N_*), \quad z\in \overline\C\setminus \bigcup_{e\in E_t} \big\{|z-e|<\delta\big\},
\]
that is, including the traces on \( \Ga[a_1(t),b_2(t)] \). Given \( \varepsilon>0 \), let \( \N(t,\varepsilon) \) be a subsequence of indices \( n \) such that \( \Theta_n(z;t) \) is non-vanishing in \( \{|z|\geq 1/\varepsilon\} \). Then there exists a constant \( c(t,\varepsilon)>0 \) such that
\[
|\Theta_n(\infty;t)| \geq c(t,\varepsilon), \quad n\in \N(t,\varepsilon).
\]
\end{proposition}

As in the case of the Szeg\H{o} functions \( \mathcal D(z;t) \), it would convenient for us to renormalize \( \Theta_n(z;t) \) at infinity. Thus, we set \( \vartheta_n(z;t) := \Theta_n(z;t)/\Theta_n(\infty;t) \). Observe that \( \vartheta_n(z;t)D^{N_n-n}(z;t)e^{n\mathcal Q(z;t)} \) are functions holomorphic in \( \C\setminus J_t \).

Proposition~\ref{prop:A} has substance only if the sets \( \mathbb N(t,\varepsilon) \) have infinite cardinality. Recall \( \mathsf B(t) \) from \eqref{BQ} above. It follows from the general theory of Riemann surfaces, see Section~\ref{ss:RS}, that \( \im(\mathsf B(t))>0 \). In particular, any \( s\in\C \) can be uniquely written as \( x+\mathsf B(t)y \) for some \( x,y\in\R \).

\begin{proposition}
\label{prop:N0}
Write \( \varsigma(t) =: x(t) + \mathsf B(t)y(t) \), \(x(t),y(t)\in\R\). It holds that
\begin{itemize}
\item[(i)] if \( x(t),y(t)\in\Z \), then the functions \( \Theta_n(z;t) \) do not depend on the choice of  \( N_n \) and at least one of the integers \( n,n+1 \) belongs to \( \N(t,\varepsilon) \)\footnote{In fact, this claim can be improved depending on the arithmetic properties of the numbers \( \omega(t) \) and \( \tau(t) \).};
\item[(ii)] otherwise, for any natural number \( N_*>0 \) there exists a choice of \( N_n \) such that \( |n-N_n| \leq N_* \) and \( \N(t,\varepsilon) = \N \) for all \( \varepsilon \) small enough.
\end{itemize}
\end{proposition}

On the other hand if the sequence \( \{N_n\}_{n\in\N}\) is fixed, we can claim the following.

\begin{proposition}
\label{prop:N}
Let \( \{N_n\}_{n\in\N}\) be a sequence such that \( |n-N_n|\leq N_* \) for some  \( N_*\geq0 \). The sequence \( \N(t,\varepsilon) \) is infinite for all \( \varepsilon \) small enough unless there exist integers \( d>0,k,i_1,i_2,m_1,m_2 \) such that
\begin{equation}
\label{degen-cond}
\varsigma(t) = (i_1+\mathsf B(t) i_2)/d, \quad \omega(t) d = (k-1)i_1 + m_1d,  \qandq \tau(t) d  = (k-1)i_2+ m_2d,
\end{equation}
where at least one of the fractions \( i_1/d \), \( i_2/d \) is irreducible, \( d \) is even, \( i_1 \) and \( i_2 \) are odd, and the sequence \( \{ N_n \} \) is such that all but finitely many numbers \( nk-N_n \) are divisible by \( d/2 \) but not by \( d \). Moreover, the following special cases take place:
\begin{itemize}
\item[(i)] if there exists an infinite subsequence \( \{n_l\} \) such that \( N_{n_l+1}-N_{n_l}\in\{0,1\} \) (in particular, this happens when \( N_n=n \), in which case the differences are always equal to \( 1 \)), then at least one of the integers \( n_l,n_l+1 \) belongs to \( \N(t,\varepsilon) \).
\item[(ii)] if one of the triples \( \omega(t),x(t),1 \) or \( \tau(t),y(t),1 \) is rationally independent, then at least one of the integers \( n,n+1 \) belongs to \( \N(t,\varepsilon) \). 
\end{itemize}
\end{proposition}

We prove Propositions~\ref{prop:A},~\ref{prop:N0} and~\ref{prop:N} in Sections~\ref{ss:A},~\ref{ss:N0}, and~\ref{ss:N}, respectively. 

\begin{theorem}
\label{global2}
Let \( t\in O_\mathsf{two-cut} \) and \( \{N_n\}_{n=1}^\infty \) be a sequence such that \( |n-N_n|\leq N_* \) for some \( N_* \) fixed. Let \( P_n(z;t,N) \) be defined by \eqref{cm2}--\eqref{cm3} and
\[
\psi_n(z;t) := P_n(z;t,N_n)e^{-n(V(z;t)-\ell_*(t))/2}.
\]
Given \( \varepsilon>0 \), let \( \N(t,\varepsilon) \) be as in Proposition~\ref{prop:A}. Then for all $n\in\N(t,\varepsilon)$ large enough it holds that\footnote{By writing $f(n) = \mathcal{O}_{\varepsilon}(n^{-1})$ we mean that there exists a constant $C_\varepsilon$ depending on $\varepsilon$ for which $|f(n)| \leq C_\varepsilon n^{-1}$.}
\begin{equation}
\label{tc17}
\psi_n(z;t) = \left( \left(A\vartheta_nD^{N_n-n}\right)(z;t)  +  \mathcal{O}_\varepsilon\big(n^{-1}\big)\right) e^{n\mathcal Q(z;t)}
\end{equation}
locally uniformly in $\C\setminus J_t$; moreover,
\begin{equation}
\label{tc18}
\psi_n(s;t) = \left(A\vartheta_nD^{N_n-n}\right)_+(s;t)e^{n\mathcal Q_+(s;t)} + \left(A\vartheta_nD^{N_n-n}\right)_-(s;t)e^{n\mathcal Q_-(s;t)} +\mathcal{O}_\varepsilon\big(n^{-1}\big)
\end{equation}
locally uniformly on $ J_t^\circ$. 
\end{theorem}

Recall that each \( \vartheta_n(z;t) \) might have a single zero in \( \C\setminus J_t \). If these zeros accumulate at a point \( z_* \) along a subsequence of \( \N(t,\varepsilon) \), then the polynomials \( P_n(z;t,N_n) \) will have a single zero approaching \( z_* \) along this subsequence by \eqref{tc17} and Rouche's theorem. With this exception, it also follows from \eqref{tc17} that \( P_n(z;t,N_n) \) are eventually zero free on compact subsets \( \C\setminus J_t \). The main part of the proof of Theorem~\ref{global2} is carried out in Section~\ref{s:aa} with auxiliary details relegated to Sections~\ref{s:g} and~\ref{sec:An}.

The functions \( \vartheta_n(z;t) \) were defined as pull-backs from one of the sheets (a copy of \( \overline\C\setminus J_t \)) of certain ratios of Riemann theta functions defined on the (two-sheeted) Riemann surface of \( Q^{1/2}(z;t) \). To describe asymptotics of the recurrence coefficients we shall need the pull-backs from the other sheet as well, which, for the purpose of  the next theorem, we denote by \( \vartheta_n^*(z;t) \). To complicate matters further, our analysis requires another related family of ratios of theta functions, which, in the next theorem, we will denote by \( \widetilde\vartheta_n(z;t) \), see \eqref{all-thetas} for the rigorous definition of these functions. The proof of Proposition~\ref{prop:A} in fact shows that \( \vartheta_n^*(\infty;t) = \mathcal O_\varepsilon(1) \) and \( \widetilde\vartheta_n(\infty;t) = \mathcal O_\varepsilon(1) \) for \( n\in \N(t,\varepsilon) \).

\begin{theorem}
\label{thm:beta}
 Let \( S_k(t) := \sum_{i=1}^2 \big(b_i^k(t)-a_i^k(t)\big) \) for \(k\in\{1,2\}\). In the setting of Theorem \ref{global2}, it holds for all $n \in \N(t, \varepsilon)$ large enough that 
\begin{equation}
\label{hn-asymp}
\begin{cases}
h_n(t, N_n) & = \frac{\pi}{2} e^{-n\ell_*(t)} S_1(t)  \mathcal{D}^{2(n - N_n)}(\infty; t)  \vartheta_n^*(\infty;t) + \mathcal{O}_{\varepsilon}\big( n^{-1}\big), \medskip \\
\gamma_n^2(t, N_n) & = \frac1{16} S_1^2(t) \vartheta_n^*(\infty;t)\widetilde\vartheta_n(\infty;t) + \mathcal{O}_{\varepsilon}\big( n^{-1}\big),
\end{cases}
\end{equation}
where $\mathcal{D}(\infty; t)$ is as in Proposition \ref{prop:Szego}, $\ell_*(t)$ as in \eqref{oc5}, and $\mathcal{O}_\varepsilon(\cdot)$ is interpreted in the same way as in Theorem \ref{global2}. Furthermore, it holds that
\begin{equation}
\label{betan-asymp}
\beta_n(t, N_n) = \frac{\vartheta_n^*(\infty;t) \big( \frac{S_2(t)}{2S_1(t)} + \frac{\dd}{\dd z} \big( \log \vartheta_n(1/z;t) + \log \vartheta_n^*(1/z;t) \big)\bigl|_{z = 0} \big) + \mathcal O_\varepsilon\big( n^{-1}\big)}{\vartheta_n^*(\infty;t) + \mathcal O_\varepsilon\big( n^{-1}\big)},
\end{equation}
where we agree that \( \vartheta_n^*(\infty;t)\frac{\dd}{\dd z} \big(\log \vartheta^*_n(1/z;t) \big)\big|_{z=0} = \frac{\dd}{\dd z} \vartheta^*_n(1/z;t) \big|_{z = 0} \) even when \(  \vartheta^*_n(\infty;t)=0 \).
\end{theorem}
Of course, if there exists a subsequence of \( \N(t,\varepsilon) \) along which the values \( \vartheta^*_n(\infty;t) \) are separated away from zero, then formula \eqref{betan-asymp} can be further simplified. However, such a subsequence might not exist (in particular, one can deduce from \eqref{residue} and \eqref{7.3.2} further below that it does not exist when \( N_n=n \) and \( \omega(t)=\tau(t)=1/2 \)). Notice that formulae \eqref{hn-asymp} are not immediately reflective of the equality in \eqref{cm5b}. However, as we point out further below in \eqref{rc8}, the leading term of the asymptotics of \( \gamma_n^2(t,N_n) \) can be rewritten to make this connection more transparent. Theorem \ref{thm:beta} is proven in Section \ref{s:ae}, but its proof heavily relies on the material of Section~\ref{s:aa}.

\section{S-curves}
\label{s:pr-geom}

In this section we prove Theorem~\ref{geometry2}. We do it in several steps. In Section~\ref{ss:51} we gather results about quadratic differentials that will be important to us throughout the proof. In Section~\ref{ss:52} we show the validity of formula \eqref{ts8}; that is, we prove that we are indeed in the two-cut case when \( t\in O_\mathsf{two-cut} \). In Section~\ref{ss:53} we show that the critical and critical orthogonal graphs of
\[
\varpi_t(z):=-Q(z;t)\dd z^2
\]
do look like as depicted on Figure~\ref{fig:tc}. In Section~\ref{ss:54} we describe the dependence of the zeros of \( Q(z;t) \) on \( t \) by showing  that the variational condition \eqref{em2} and the S-property \eqref{em5} yield that the zeros satisfy a certain system of real equations with non-zero Jacobian, see \eqref{jac7} and \eqref{jac9}, and that this system is, in fact, uniquely solved by them. Finally, in Section~\ref{ss:55} we establish the limits in~\eqref{ts8a}.

\subsection{On Quadratic Differentials}
\label{ss:51}

 To start, let us also recall the following important result, known as Teichm\"uller's lemma, see \cite[Theorem 14.1]{Strebel}. Let \( P \) be a geodesic polygon of a quadratic differential, that is, a Jordan curve in $\overline{\C}$ that consists of a finite number of trajectories and orthogonal trajectories of this differential. Then it holds that 
\begin{equation}
\label{tl}
\sum_{z \in P} \left(1 - \theta(z) \dfrac{2+\mathrm{ord}(z)}{2 \pi}\right) = 2 + \sum_{z \in \mathrm{int}(P)} \mathrm{ord}(z),
\end{equation}
where $\mathrm{ord}(z)$ is the order of $z$ with respect to the considered differential and $\theta(z) \in [0, 2\pi]$, $z \in P$, is the interior angle of $P$ at $z$. Both sums in \eqref{tl} are finite since only critical points of the differential have a non-zero contribution. 

Let us briefly recall the main properties of the differential \( \varpi_t(z) \). The only critical points of \( \varpi_t(z) \) are the zeros of \( Q(z;t) \) and the point at infinity. Regular points have order \(0\), the order of a zero of \( Q(z;t) \) is equal to its multiplicity, and infinity is a critical point of order $-8$. Through each regular point passes exactly one trajectory and one orthogonal trajectory of $\varpi_t(z)$, which are orthogonal to each other at the point.  Two distinct (orthogonal) trajectories meet only at critical points \cite[Theorem 5.5]{Strebel}. As $Q(z;t)$ is a polynomial, no finite union of (orthogonal) trajectories can form a closed Jordan curve, while a trajectory and an orthogonal trajectory can intersect at most once \cite[Lemma 8.3]{Pommerenke}.  Furthermore, (orthogonal) trajectories of $\varpi_t(z)$ cannot be recurrent (dense in two-dimensional regions) \cite[Theorem 3.6]{Jenkins}. From each critical point of order \( m>0 \) there emanate \( m+2 \) critical trajectories whose consecutive tangent lines at the critical point  form an angle \( 2\pi/(m+2) \). Furthermore, since infinity is a pole of order 8, the critical trajectories can approach infinity only in six distinguished directions, namely, asymptotically to the lines \( L_{-\pi/6} \), \( L_{\pi/6} \), and \( L_{\pi/2} \), where \( L_\theta = \{z:~z=re^{\ic\theta},~r\in(-\infty,\infty)\} \). In fact, there exists a neighborhood of infinity such that any trajectory entering it necessarily tends to infinity \cite[Theorem 7.4]{Strebel}. This discussion also applies to orthogonal trajectories. In particular, they can approach infinity asymptotically to the lines \( L_0 \), \( L_{\pi/3} \), and \( L_{2\pi/3} \) only.

Denote by \( \mathcal G \) the critical graph of \( \varpi_t(z) \), that is, the totality of all the critical trajectories of \( \varpi_t(z) \). Then, see \cite[Theorem 3.5]{Jenkins}, the complement of \( \mathcal G \) can be written as a disjoint union of either half-plane or strip domains. Recall that a \emph{half-plane (or end) domain} is swept by trajectories unbounded in both directions that approach infinity along consecutive critical directions. Its boundary is connected and consists of a union of two unbounded critical trajectories and a finite number (possibly zero) of short trajectories of $\varpi_t(z)$. The map $z \mapsto \int^z \sqrt{-\varpi_t}$ maps end domains conformally onto half planes $\{z \in \C \ | \ \re (z) >c \}$ for some $c \in \R$ that depends on the domain, and extends continuously to the boundary. Similarly, a \emph{strip domain} is again swept by trajectories unbounded in both directions, but its boundary consists of two disjoint $\varpi_t(z)$-paths, each of which is comprised of two unbounded critical trajectories and a finite number (possibly zero) of short trajectories. The map $z \mapsto \int^z \sqrt{-\varpi_t}$ maps strip domains conformally onto vertical strips $\{w \in \C \ | \ c_1 < \re (w) < c_2 \}$ for some $c_1,c_2 \in \R$ depending on the domain, and extends continuously to their boundaries. The number $c_2 - c_1$ is known as the \emph{width of a strip domain} and can be calculated in terms of $\varpi_t(z)$ as 
\begin{equation}
\label{strip-width}
\left| \re \int_p^q \sqrt{-\varpi_t} \right|
\end{equation}
where $p, q$ belong to different components of the boundary of the domain.

\subsection{Proof of \eqref{ts8}}
\label{ss:52}

Assume to the contrary that for a given $t\in O_{\mathsf{two-cut}}$ we are in one-cut case. That is, there exists a choice of \( a \), \( b \), and \( c \) such that the polynomial \( Q(z;t) \) from \eqref{em5} has the form \eqref{ts1}. It follows from  \eqref{em5} in conjunction with \eqref{cm2} that
\begin{equation}
\label{ts2}
Q(z;t)=\left(\frac{-z^2+t}{2}-\frac{1}{z}+\mathcal O\left(z^{-2}\right)\right)^2=\frac{(z^2-t)^2}{4}+z+C
\end{equation}
for some constant \( C \). Then, by equating the coefficients in \eqref{ts1} and \eqref{ts2}, we obtain a system of equations
\begin{equation}
\label{ts3}
\left\{
\begin{aligned}
&a+b+2c=0,\\
&ab+c^2+2(a+b)c=-2t,\\
& 2abc+(a+b)c^2=-4.
\end{aligned}
\right.
\end{equation}
Setting $x:= (a + b)/2$ and eliminating the product $ab$ from the second and third relations in \eqref{ts3} yields 
\begin{equation}
\label{ts4}
x^3 - tx - 1 = 0,
\end{equation}
which is exactly the equation appearing before \eqref{ts5}. Given any solution of \eqref{ts4}, say \( x(t) \), then \( a(t) \), \( b(t) \), and \( c(t) \) are necessarily expressed via \eqref{ts6}. Theorem~\ref{fundamental} and the variational condition \eqref{em2} imply that there must exist a contour \( \Ga_t \in\mathcal T\) (this class of contours was defined right after \eqref{cm2b}) such that
\begin{equation}
\label{needed}
\mathcal U(z;t) \leq 0 \quad \text{ for all } \quad z\in\Ga_t,
\end{equation}
see \eqref{em0}. In what follows, we shall show that no such contour exists in \( \mathcal T \) for any of the three  possible choices of $x(t)$ solving \eqref{ts4} when $t \in O_{\mathsf{two-cut}}$ and \( Q(z;t) \) is given by \eqref{ts1} and \eqref{ts6}. 

In accordance with the above strategy, observe that the solutions of \eqref{ts4} can be written as 
\begin{equation}
\label{ts7}
x_k(t) = u_k(t) + \dfrac{t}{3u_k(t)}, \quad u_k(t) := \left(\dfrac{1}{2} - \sqrt{\dfrac{1}{4} - \dfrac{t^3}{27}}\right)^{1/3} e^{2k\pi \ic/3},
\end{equation}
\( k \in\{ 0, 1, 2 \} \), with all branches being principal. It can be readily verified that \( x_1(t) \) is analytic in \( \C\setminus(\overline{e^{2\pi\ic/3}S},\overline S) \) (here, \(\overline\cdot\) means topological closure), see \eqref{ts5} for the definition of the ray \( S \), and
\begin{equation}
\label{xrelations}
\left\{
\begin{array}{ll}
x_0(t) & = e^{4\pi \ic /3} x_1\big(t e^{-2\pi \ic /3}\big), \\
x_1(t) &= e^{4\pi \ic /3}\overline{x_1\big(\overline{t} e^{2\pi \ic/3}\big)},\\
x_2(t) &= \overline{x_1\big(\overline{t}\big)}.
\end{array}
\right.
\end{equation}
Furthermore, noting that the function $x(t)$, defined after \eqref{ts5}, maps $\Omega_{\mathsf{one-cut}}$ onto $O_{\mathsf{one-cut}}$, it can be easily checked that $x(t)$ is evaluated as shown on Figure~\ref{fig:x-values}, where the dashed lines are the chosen branch cuts of $x_1(t)$. In particular, $x(t)$ can be analytically continued across \( C_\mathsf{split} \), \( C_\mathsf{birth}^a \), and \( C_\mathsf{birth}^b \), see \eqref{ts5} and Figure~\ref{fig:x-values}. In what follows, we consider what happens in the case of each of these continuations. 

\begin{figure}[ht!]
	\centering
	\includegraphics[scale=1.2]{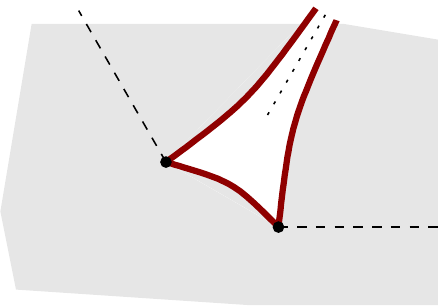}
	\begin{picture}(0,0)
	\put (-120,18){$x_1(t)$}
	\put (-100, 75){$x_0(t)$}
	\put (-35, 50){$x_2(t)$}
	\put(-40,85){$C_\mathsf{birth}^b$}
	\put(-80,90){$C_\mathsf{birth}^a$}
	\put(-92,33){$C_\mathsf{split}$}
	\end{picture}
	\caption{\small Determination of $x(t)$}
	\label{fig:x-values}
\end{figure}

Continue \( x(t) \) into \( O_{\mathsf{two-cut}} \) by either \( x_2(t) \) or \( x_0(t) \), that is, analytically across either  \( C_\mathsf{birth}^b \) or  \( C_\mathsf{birth}^a \), see Figure~\ref{fig:x-values}. The first and the last symmetries in \eqref{xrelations} then yield that \( \varpi_t(z) \) is either equal to
\[
\overline{\varpi_{\overline{t}}(\overline{z})} \quad \text{or} \quad \varpi_{t e^{-2\pi \ic/3}} \big(z e^{-4\pi \ic/3}\big).
\]
Since the set \( O_{\mathsf{two-cut}} \) is symmetric with respect to the line \( L_{\pi/3} \), its rotation by \( -2\pi/3 \) is equal to its reflection across the real axis. Thus, the critical graph \( \varpi_t(z) \) when \( t\in O_\mathsf{two-cut} \) and \( x(t) \) is continued by either \( x_2(t) \) or \( x_0(t) \) is equal to the reflection across the real axis or the rotation by \( 4\pi/3 \) of the critical graph of \( \varpi_{t_*}(z) \) for some \( t_* \) such that \( \overline t_* \in  O_\mathsf{two-cut} \). These graphs were studied in \cite[Theorems~3.2 and~3.4]{BlDeaY17} and determined to have the structure as depicted on Figure~\ref{s-curves1}(a--c) (or the reflection of these three panels across the line \( L_{2\pi/3} \)). A direct examination shows that none of these critical graphs form a curve in \(  \mathcal T \) for which \eqref{needed} holds (such a curve must belong to the closure of the gray regions on Figure~\ref{s-curves1}).

Suppose now that we continue $x(t)$ by \( x_1(t) \), that is, analytically across $C_{\mathsf{split}}$. For such a choice of \( x(t) \), the critical graph of \( \varpi_t(z) \) was studied in \cite{HKL} when \( t\in L_{\pi/3} \) (in which case the critical graph is symmetric with respect to \( L_{2\pi/3} \)). In particular, it was shown that \( x(t)\in L_{2\pi/3} \), no union of critical trajectories join \( a \) and \( b \), and no critical trajectory of \( \varpi_t(z) \) crosses the line \( L_{2\pi/3} \) when \( t\in L_{\pi/3}\cap O_\mathsf{two-cut} \), see \cite[Lemma~3.2]{HKL}. Since critical trajectories cannot intersect, can approach infinity asymptotically to the lines \( L_{\pi/6} \), \( L_{\pi/2}\), and \( L_{-\pi/6} \) only and must obey Teichm\"uller's lemma \eqref{tl}, the critical graph of $\varpi_t(z)$ must be as on Figure~\ref{fig:cc-sym-line}.
\begin{figure}[ht!]
\includegraphics[scale=0.25]{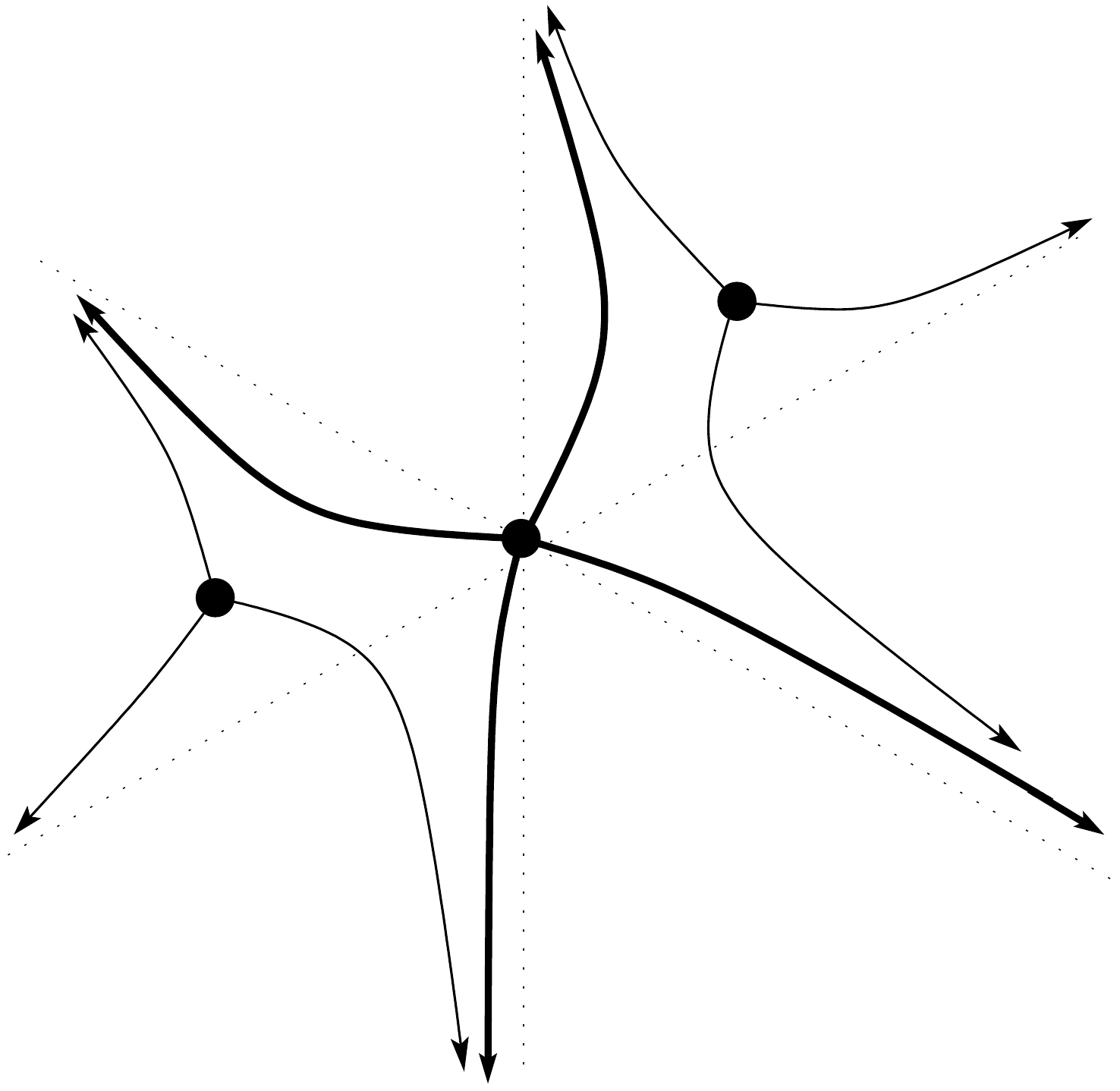}
\caption{\small The critical graph of $\varpi_t(z)$ when $x(t)$ is analytically continued across $C_{\mathsf{split}}$.}
\label{fig:cc-sym-line}
\end{figure}
Clearly, this critical graph does not yield a curve in \(  \mathcal T \) for which \eqref{needed} holds. Thus, to complete the proof, we need to argue that the structure of the critical trajectories of $\varpi_t(z)$ remains the same for all $t \in O_\mathsf{two-cut}$ in the considered case. Observe that it is enough to show that all the trajectories out of \( b \) approach infinity.

Recall that the trajectories emanating out of \( b \) are part of the level set \( \{\mathcal U(z;t) =0\} \), see \eqref{em0}. When \( t\in L_{\pi/3}\cap O_\mathsf{two-cut} \), these trajectories approach infinity at the angles \( -\pi/6 \), \( \pi/6 \), and \( \pi/2 \), see Figure~\ref{fig:cc-sym-line}. Since the values of \( \mathcal U(z;t) \) analytically depend on \( t \),  the same must be true in a neighborhood of each such \( t \). This will remain so until one of the trajectories hits a critical point different from the one at infinity. As can be seen from Figure~\ref{fig:cc-sym-line}, this critical point must necessarily be \( c=-x \), see \eqref{ts6}. That is, as long as \( \mathcal U(-x;t)\neq 0 \), the trajectories out of \( b \) will asymptotically behave as on  Figure~\ref{fig:cc-sym-line}. 
\begin{figure}[ht!]
\subfigure[]{\includegraphics[scale=.23]{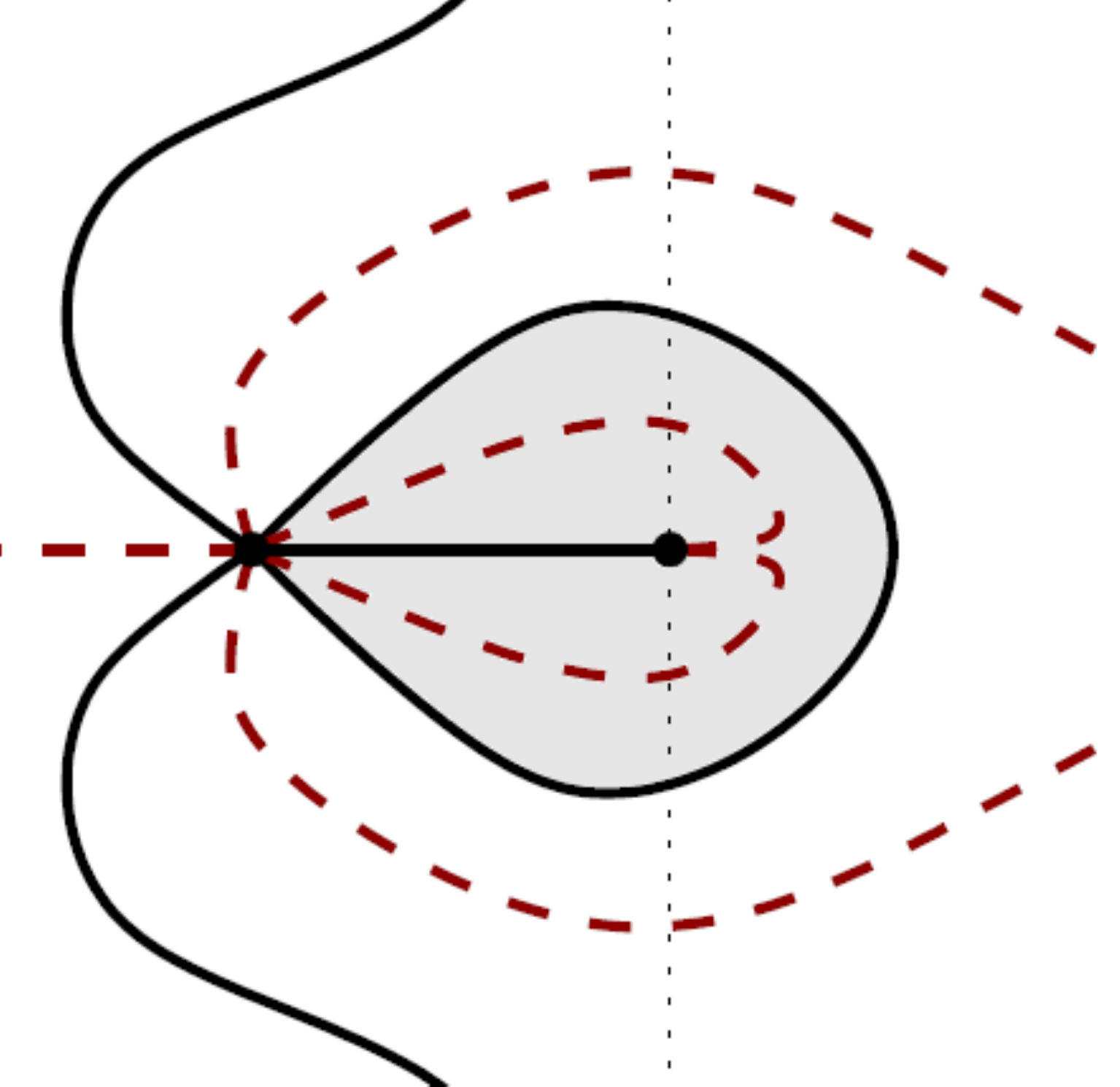}} \quad\quad
\subfigure[]{\includegraphics[scale=.2]{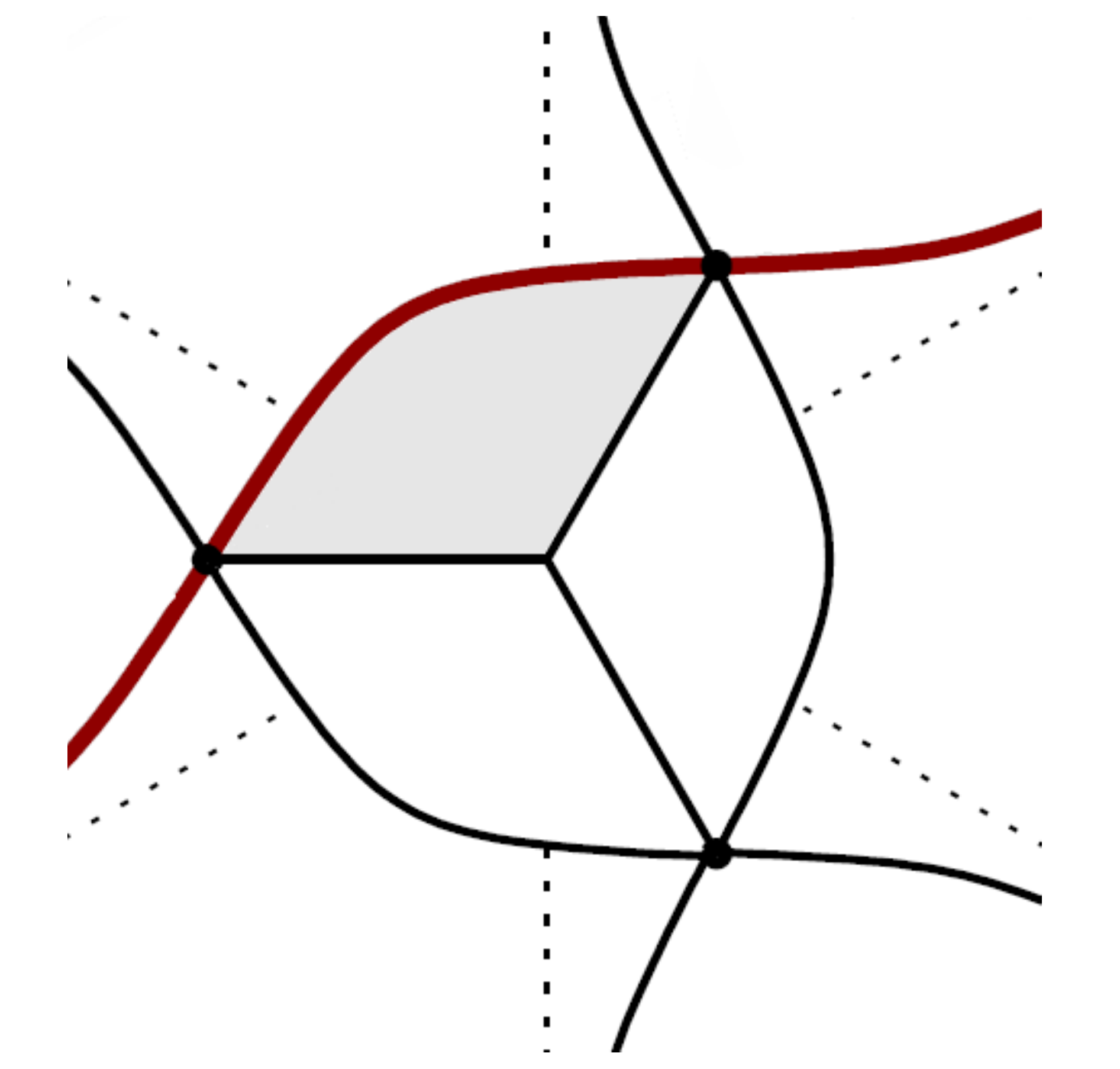}}
\caption{\small Shaded regions represent the domains within which \( 2x_1^3(t) \) (panel a) and \( x_1(t) \) (panel b) change when \( t\in O_\mathsf{two-cut} \).}
 \label{fig:x1-cont}
 \end{figure} 
 When \( x(t) \) is continued into \( O_\mathsf{two-cut} \) by \( x_1(t) \), its values lie within the gray region on Figure~\ref{fig:x1-cont}(b), see also Figure~\ref{fig:loops}(b). Respectively, the values \( 2x^3(t) \) lie within the gray region on Figure~\ref{fig:x1-cont}(a), see also Figure~\ref{fig:loops}(a). It was verified in \cite[Section~5.3]{BlDeaY17} that
\[
\mathcal U(-x;t) = \re\left(\frac23 \int_{-1}^{2x^3} \left( 1 + \dfrac{1}{s}\right)^{3/2} \dd s\right),
\]
where the path of integration lies within the shaded domain on Figure~\ref{fig:x1-cont}(a). Hence, \( \mathcal U(-x;t) = 0 \) if and only if $2x^3$ belongs to a trajectory of $-( 1 + 1/s)^3\dd s^2$ emanating from $-1$. These trajectories are drawn on Figures~\ref{fig:loops} and ~\ref{fig:x1-cont}(a) (black lines).  Thus,  \( \mathcal U(-x;t)\neq 0 \) in the considered case as claimed.

\subsection{Critical Graph of \( \varpi_t(z) \)}
\label{ss:53}

Let, as usual, \( Q(z;t) \) be the polynomial guaranteed by Theorem~\ref{fundamental}. According to what precedes, it has the form \eqref{ts8} when \( t\in O_\mathsf{two-cut} \). Recall the properties of the differential \( \varpi_t(z) = - Q(z;t) \dd z^2 \) described at the beginning of Section~\ref{ss:51}. In particular, it has four critical points of order \( 1 \), which, for a moment, we label as \( z_1(t),z_2(t),z_3(t),z_4(t) \) (these are the zeros of \( Q(z;t) \)), a critical point of order \( -8 \) at infinity, and no other critical points. It follows from Theorem~\ref{fundamental}(4) and \eqref{em2} with \eqref{em0} that  
\begin{equation}
\label{widths}
\re\left( \int_{\Ga_t[z_i(t),z_j(t)]} Q_+^{1/2}(z;t)\dd z \right)= 0,
\end{equation}
where \( \Ga_t[z_i(t),z_j(t)] \) is the subarc of \( \Ga_t \) with endpoints \( z_i(t),z_j(t) \) and \( Q_+^{1/2}(z;t) \) is the trace of \( Q^{1/2}(z;t) \) on the positive side of \( \Ga_t \). Equations \eqref{widths} imply existence of three short critical trajectories of \( \varpi_t(z) \). Indeed, if all three critical trajectories out of a zero \( z_i(t) \) approach infinity, then \( z_i(t) \) must belong to a boundary of at least one strip domain. Let \( z_j(t) \) be a different zero of \( Q(z;t) \) belonging to the other component of the boundary of this strip domain. Then it follows from \eqref{strip-width} and \eqref{widths} that the width of this strip domain is  \( 0 \), which is impossible. Thus, each zero  of \( Q(z;t) \) must be coincident with at least one short trajectory. Therefore, either there is a zero, say \( z_4(t) \), connected by short trajectories to the remaining three zeros or there are at least two short trajectories connecting two pairs of zeros. In the latter case, label these zeros by \( a_1(t),b_1(t) \) and \( a_2(t),b_2(t) \). If the other two trajectories out of both \( a_1(t) \) and \( b_1(t) \) approach infinity, one of these zeros again must belong to the boundary of a strip domain with either \( a_2(t) \) or \( b_2(t) \) belonging to the other component of the boundary. As before, \eqref{widths} yields that the width of this strip domain is \( 0 \), which, again, is impossible. Thus, in this case there also exists a third short critical trajectory. Then we choose a labeling of the zeros so that \( b_1(t) \) and \( a_2(t) \) are connected by this trajectory. 

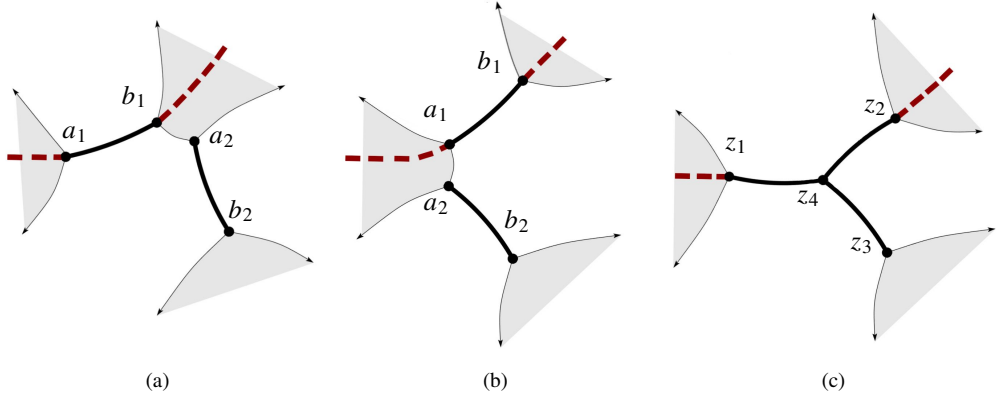
\begin{figure}[ht!]
	\subfigure[]{\includegraphics[scale=.18]{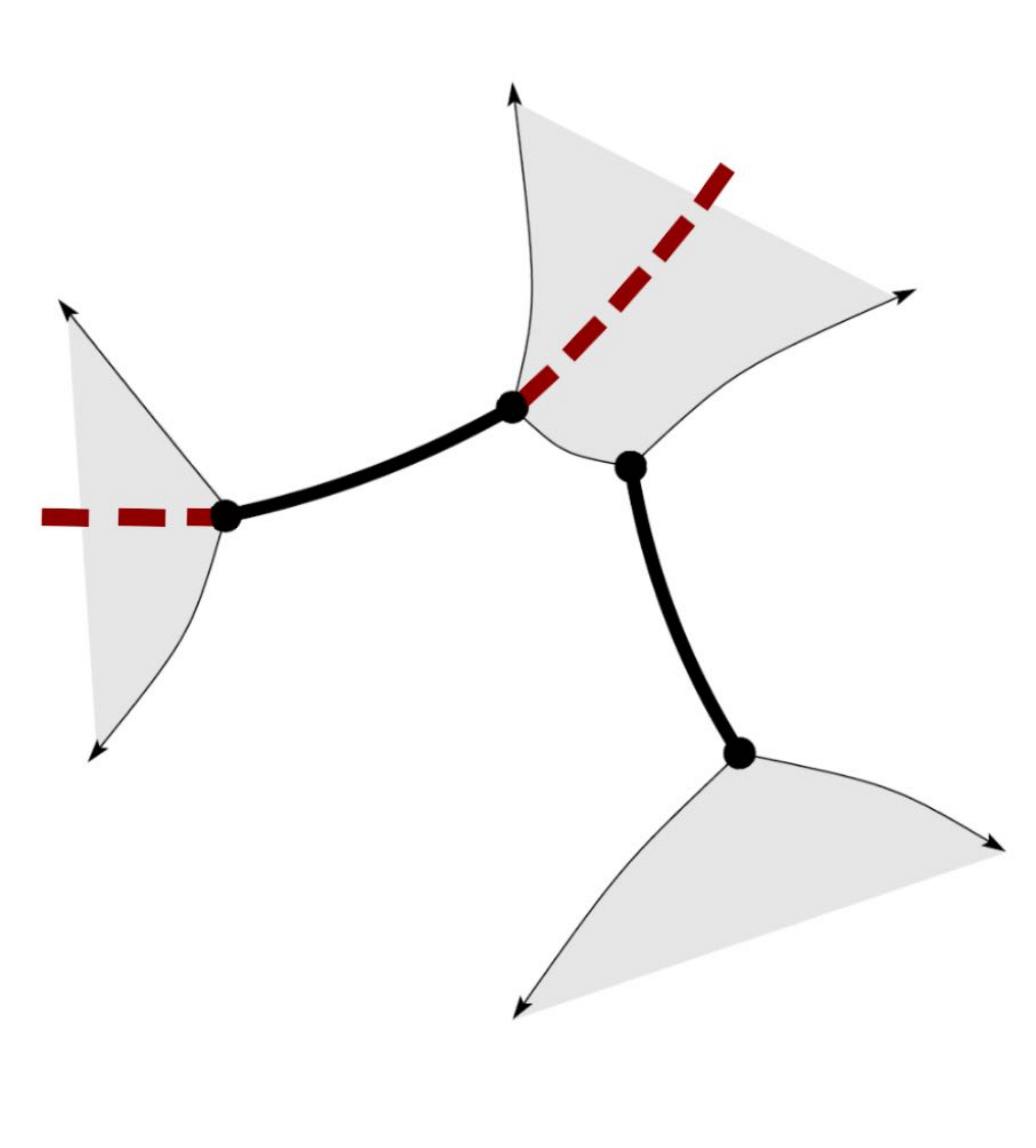} \begin{picture}(0,0)
		\put (-100,82){$a_1$}
		\put (-78, 95){$b_1$}
		\put (-45, 80){$a_2$}
		\put (-38, 52){$b_2$}
		\end{picture}}~~~
	\subfigure[]{\includegraphics[scale=.18]{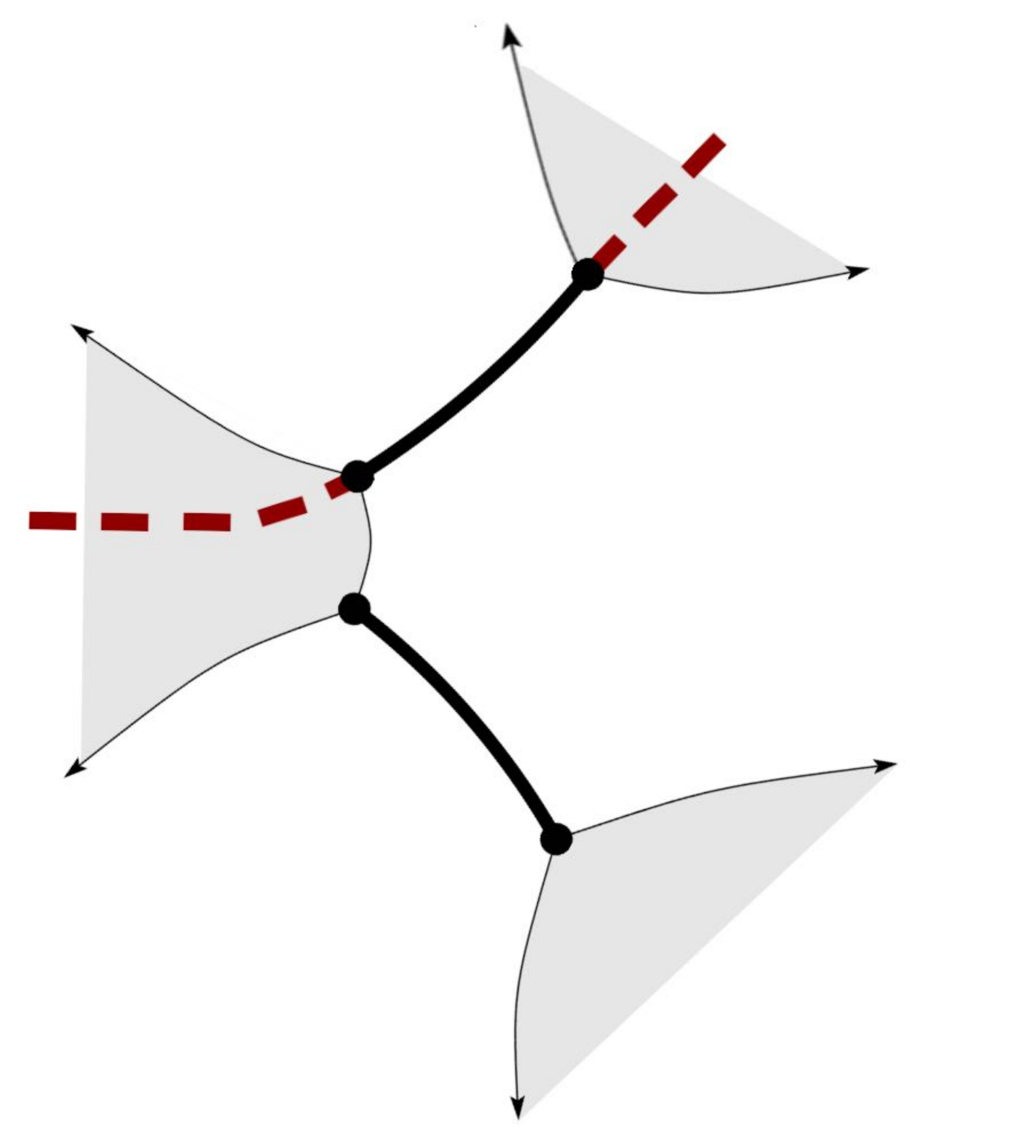}\begin{picture}(0,0)
		\put (-90,90){$a_1$}
		\put (-70, 108){$b_1$}
		\put (-90, 55){$a_2$}
		\put (-60, 48){$b_2$}
		\end{picture}}~~~
	\subfigure[]{\includegraphics[scale=.18]{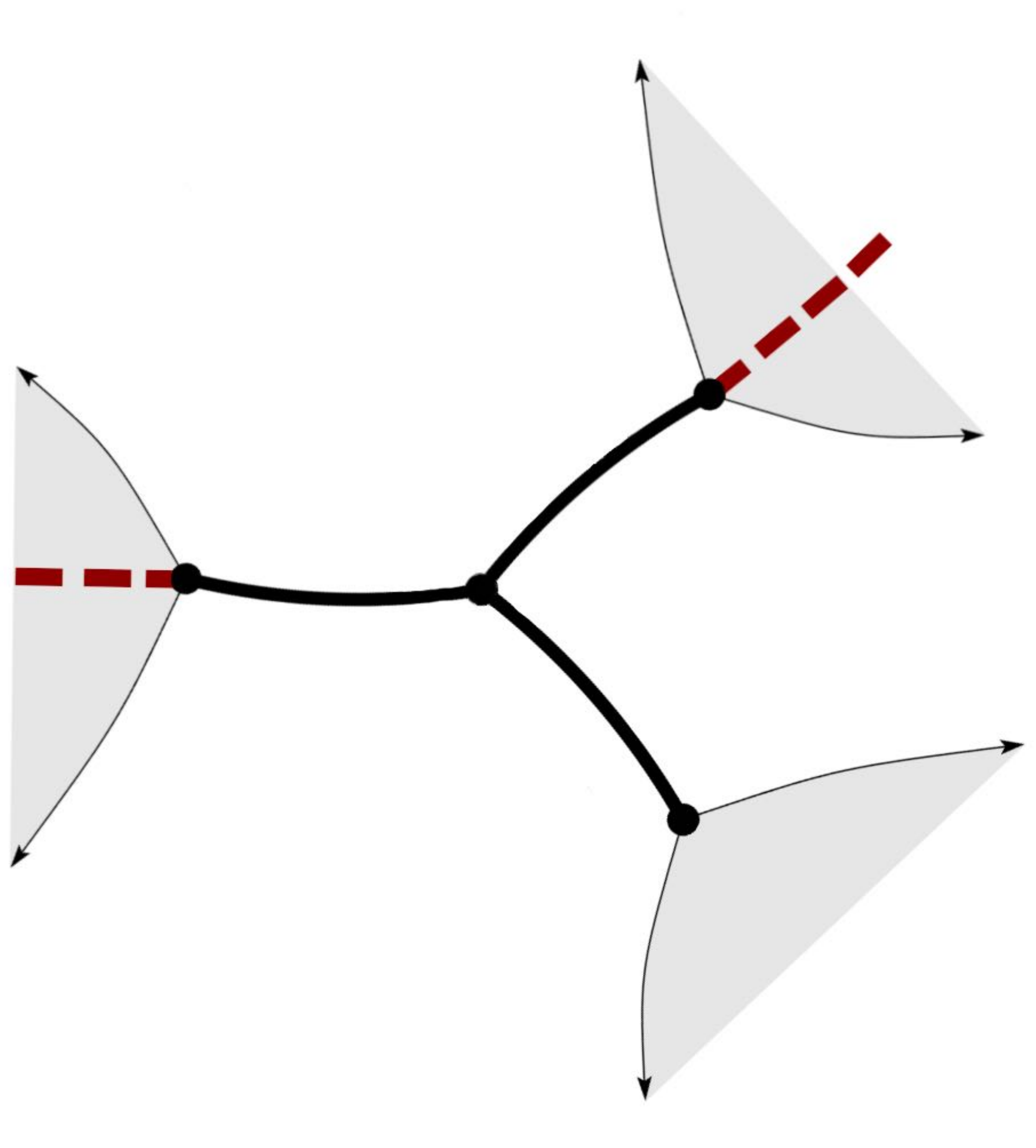} \begin{picture}(0,0)
		\put (-105,78){$z_1$}
		\put (-53, 92){$z_2$}
		\put (-58, 40){$z_3$}
		\put (-78, 58){$z_4$}
		\end{picture}}
	\caption{\small Geometries of the critical graph of \( \varpi_t(z) \). Shaded regions represent the open set \( \{\mathcal U(z;t)<0\} \), the white regions represent the open \( \{\mathcal U(z;t)>0\} \), and the red dashed arcs form \( \Ga_t\setminus J_t \).}
	\label{fig:53}
\end{figure}  

Since short critical trajectories cannot form closed curves, there cannot be any more of them. That is, the remaining critical trajectories are unbounded. Consider the two unbounded critical trajectories out of \( z_1(t) \) in the case where short ones form a threefold, see Figure \ref{fig:53}(c). Since critical trajectories cannot intersect and the remaining zeros are connected to \( z_1(t) \) by short critical trajectories, the unbounded critical trajectories out of \( z_1(t) \) delimit a half-plane domain and, in particular, must approach infinity along consecutive critical directions (those are given by the angles \( (2k+1)\pi/6\), \( k\in\{0,\ldots,5\}\), see Section~\ref{ss:51}). Clearly, the same is true for the unbounded critical trajectories out of \( z_2(t) \) and \( z_3(t) \) as well as for the unbounded critical trajectories out of \( a_1(t) \), \( b_2(t) \), and  the union of the unbounded critical trajectory out of \( b_1(t) \), the short critical trajectory connecting \( b_1(t) \) to \( a_2(t) \), and the unbounded critical trajectory out of \( a_2(t) \) in the case where short critical trajectories form a Jordan arc, see Figure~\ref{fig:53}.

Now, let \( \mathcal U(z;t) \) be given by \eqref{em0}. Clearly, \( \mathcal U(z;t) \) is a subharmonic function which is equal to zero on \( J_t \), see \eqref{em2}. Since \( \mathcal U(z;t) \) must have the same sign on both sides of each subarc of \( J_t \) by \eqref{em4}, it follows from the maximum principle for subharmonic functions that it is positive there. Further, since trajectories of \( \varpi_t(z) \) cannot form a closed Jordan curve, all the connected components of the open set  \( \{\mathcal U(z;t)<0\} \) must necessarily extend to infinity. Since \( \re(V(z;t)) \) is the dominant term of \( \mathcal U(z;t) \) around infinity, see \eqref{em0},  for any \( \delta>0 \) there exists \(R>0\) sufficiently large so that
\begin{equation}
\label{sectorsS}
\left\{
\begin{array}{ll}
\big(S_{\pi/3,\delta}\cup S_{\pi,\delta}\cup S_{-\pi/3,\delta}\big)\cap \{|z|>R\} &\subset \{\mathcal U(z;t)<0\},  \medskip \\
\big(S_{0,\delta}\cup S_{2\pi/3,\delta}\cup S_{-2\pi/3,\delta}\big) \cap \{|z|>R\} &\subset \{\mathcal U(z;t)>0\},
\end{array}
\right.
\end{equation}
where \( S_{\theta,\delta} := \{|\arg(z)-\theta|<\pi/6-\delta\}\). Altogether, the critical graph of \( \varpi_t(z) \) must look like either on Figure~\ref{fig:tc} or on Figure~\ref{fig:53}.

It remains to show that \( \varpi_t(z) \) cannot have the critical graph as on any of the panels of Figure~\ref{fig:53}. To this end, recall that the contour \( \Ga_t \) must contain \( J_t \) and two unbounded arcs extending to infinity in the directions \( \pi/3 \) and \( \pi \) (red dashed unbounded arcs on Figure~\ref{fig:53}).  Let \( \Ga_* \) be obtained from \( \Ga_t \) by dropping the short trajectory that is a part of \( J_t \) and whose removal keeps \( \Ga_* \) connected (this can be done for any of the panels on Figure~\ref{fig:53}). Observe that \( \Ga_* \) also belongs to \( \mathcal T \). Let \( \mu_* \) be the weighted equilibrium distribution on \( \Ga_* \) as defined in Definition~\ref{def:eq}. Since \( \Ga_* \subset \Ga_t \), it holds that \( \mu_*\in\mathcal M(\Ga_t) \). Moreover, since \( \mu_*\neq\mu_t \), \( \mathcal E_V(\Ga_*) = E_V(\mu_*)>E_V(\mu_t) = \mathcal E_V(\Ga_t) \), see \eqref{em1}. However, the last inequality clearly contradicts~\eqref{em3}. 

We have shown that the critical graph of \( \varpi_t(z) \) must look like on Figure~\ref{fig:tc}. As the critical orthogonal and critical trajectories cannot intersect, the structure of the critical orthogonal graph is uniquely determined by structure of the critical graph. Now, we can completely fix the labeling of the zeros of \( Q(z;t) \) by given the label \( a_1(t) \) to one that is incident with the orthogonal critical trajectory extending to infinity asymptotically to the ray \( \arg(z)=\pi \).

\subsection{Dependence on \( t \)}
\label{ss:54}

We start with some general considerations. Let \( f(z) \) and \( g(z) \) be analytic functions of \( z=x+\ic y \). Consider a determinant of the form
\[
D = \left|\begin{matrix} \partial_x \re(f) & \partial_y\re(f) & * \\ \partial_x \im(f) & \partial_y\im(f) & * \\ \partial_x \re(g) & \partial_y\re(g) & * \end{matrix} \right|,
\]
where the entries of the third column are not important for the forthcoming computation. Due to Cauchy-Riemann relations it holds that \( f^\prime = \partial_x \re(f) + \ic\partial_x\im(f) = \partial_y\im(f) -\ic \partial_y\re(f) \). Therefore,
\[
D = \left|\begin{matrix} \re(f^\prime) & -\im(f^\prime) & * \\ \im(f^\prime) & \re(f^\prime) & * \\ \re(g^\prime) & -\im(g^\prime) & * \end{matrix} \right| = \left|\begin{matrix} f^\prime & \ic f^\prime & * \\ \im(f^\prime) & \re(f^\prime) & * \\ \re(g^\prime) & -\im(g^\prime) & * \end{matrix} \right| = \frac\ic2\left|\begin{matrix} f^\prime & \ic f^\prime & * \\ \overline {f^\prime} & -\ic\overline{f^\prime} & * \\ \re(g^\prime) & -\im(g^\prime) & * \end{matrix} \right|
\]
by adding the second row times \( \ic \) to the first one and then multiplying the second row by \( -2\ic \) and adding the first row to it. It further holds that
\[
D = \frac\ic2\left|\begin{matrix} 2f^\prime & \ic f^\prime & * \\ 0 & -\ic\overline{f^\prime} & * \\ g^\prime & -\im(g^\prime) & * \end{matrix} \right| = \frac\ic2\left|\begin{matrix} 2f^\prime & 0 & * \\ 0 & -\ic\overline{f^\prime} & * \\ g^\prime & -\ic\overline{g^\prime}/2 & * \end{matrix} \right| = \frac12\left|\begin{matrix} f^\prime & 0 & * \\ 0 & \overline{f^\prime} & * \\ g^\prime & \overline{g^\prime} & * \end{matrix} \right|,
\]
where we added the second column times \( -\ic \) to the first one, then added the first column times \( -\ic/2 \) to the second one, and then factored 2 from the first column, \( -\ic \) from the second one, and \( 1/2 \) from the third row.

Now, let \( f_j(z_1,z_2,z_3,z_4) \), \( j\in\{1,2,3,4,5\} \), be analytic functions in each variable \( z_i = x_i + \ic y_i \). We would like to compute the Jacobian of the following system of real-valued functions of \( x_1,y_1,\ldots,x_4,y_4 \):
\begin{equation}
\label{jac1}
\re(f_1), \; \im(f_1), \; \re(f_2), \; \im(f_2), \; \re(f_3), \; \im(f_3), \; \re(f_4), \; \re(f_5).
\end{equation}
That is, we are interested in
\[
{\rm Jac} = \left|\begin{matrix} 
\re(f_{11}) & -\im(f_{11}) & \re(f_{12}) & -\im(f_{12}) & \re(f_{13}) & -\im(f_{13}) & \re(f_{14}) & -\im(f_{14}) \\ 
\im(f_{11}) & \re(f_{11}) & \im(f_{12}) & \re(f_{12}) & \im(f_{13}) & \re(f_{13}) & \im(f_{14}) & \re(f_{14}) \\
\re(f_{21}) & -\im(f_{21}) & \re(f_{22}) & -\im(f_{22}) & \re(f_{23}) & -\im(f_{23}) & \re(f_{24}) & -\im(f_{24}) \\ 
\im(f_{21}) & \re(f_{21}) & \im(f_{22}) & \re(f_{22}) & \im(f_{23}) & \re(f_{23}) & \im(f_{24}) & \re(f_{24}) \\
\re(f_{31}) & -\im(f_{31}) & \re(f_{32}) & -\im(f_{32}) & \re(f_{33}) & -\im(f_{33}) & \re(f_{34}) & -\im(f_{34}) \\ 
\im(f_{31}) & \re(f_{31}) & \im(f_{32}) & \re(f_{32}) & \im(f_{33}) & \re(f_{33}) & \im(f_{34}) & \re(f_{34}) \\ 
\re(f_{41}) & -\im(f_{41}) & \re(f_{42}) & -\im(f_{42}) & \re(f_{43}) & -\im(f_{43}) & \re(f_{44}) & -\im(f_{44}) \\
\re(f_{51}) & -\im(f_{51}) & \re(f_{52}) & -\im(f_{52}) & \re(f_{53}) & -\im(f_{53}) & \re(f_{54}) & -\im(f_{54}) \end{matrix}\right|,
\]
where \( f_{ji} := \partial_{z_i}f_j \). By performing the same row and column operations as for the determinant \( D \) above, we get that
\[
{\rm Jac} =  \frac1{2\ic}\left|\begin{matrix} 
f_{11} & 0 & f_{12} & 0 & f_{13} & 0 & f_{14} & 0 \\ 
0 & \overline f_{11} & 0 & \overline f_{12} & 0 & \overline f_{13} & 0 & \overline f_{14} \\
f_{21} & 0 & f_{22} & 0 & f_{23} & 0 & f_{24} & 0 \\ 
0 & \overline f_{21} & 0 & \overline f_{22} & 0 & \overline f_{23} & 0 & \overline f_{24} \\
f_{31} & 0 & f_{32} & 0 & f_{33} & 0 & f_{34} & 0 \\ 
0 & \overline f_{31} & 0 & \overline f_{32} & 0 & \overline f_{33} & 0 & \overline f_{34} \\
f_{41} & \overline f_{41} & f_{42} & \overline f_{42} & f_{43} & \overline f_{43} & f_{44} & \overline f_{44} \\
f_{51} & \overline f_{51} & f_{52} & \overline f_{52} & f_{53} & \overline f_{53} & f_{54} & \overline f_{54} \end{matrix}\right|
\]
with the constant in front of the determinant coming from \( (\ic/2)^3 \times 2^4 \times (-\ic)^4\times (1/2)^2 \), where the first factor is due to multiplications of the second,  fourth and sixth rows by \( -2\ic \), the second is due to factoring 2 from the odd columns, the third one comes from factoring \( -\ic \) from the even columns, and the fourth factor is due to factoring \( 1/2 \) from the last two rows. Assume further that
\begin{equation}
\label{jac2}
\left\{
\begin{array}{ll}
f_1(z_1,z_2,z_3,z_4) & = z_1+z_2+z_3+z_4, \\
f_2(z_1,z_2,z_3,z_4) & = z_1z_2+z_1z_3+z_1z_4+z_2z_3+z_2z_4 + z_3z_4, \\
f_3(z_1,z_2,z_3,z_4) & = z_2z_3z_4+z_1z_3z_4+z_1z_2z_4+z_1z_2z_3.
\end{array}
\right.
\end{equation}
Then, by using the above explicit expressions and subtracting the first (resp. second) column from the third, fifth, and seventh (resp. fourth, sixth, and eighth), we get that
\begin{multline*}
{\rm Jac} =\\  \frac1{2\ic}\left|\begin{matrix} 
1 & 0 & 0 & 0 & 0 & 0 & 0 & 0 \\ 
0 & 1 & 0 & 0 & 0 & 0 & 0 & 0 \\
* & * & z_1-z_2 & 0 & z_1-z_3 & 0 & z_1-z_4 & 0 \\ 
* & * & 0 & \overline z_1-\overline z_2 & 0 & \overline z_1-\overline z_3 & 0 & \overline z_1-\overline z_4 \\
* & * & \displaystyle \frac{z_1-z_2}{(z_3+z_4)^{-1}} & 0 & \displaystyle \frac{z_1-z_3}{(z_2+z_4)^{-1}} & 0 & \displaystyle \frac{z_1-z_4}{(z_2+z_3)^{-1}} & 0 \\ 
* &  * & 0 & \displaystyle \frac{\overline z_1-\overline z_2}{(\overline z_3+\overline z_4)^{-1}} & 0 & \displaystyle \frac{\overline z_1-\overline z_3}{(\overline z_2+\overline z_4)^{-1}} & 0 & \displaystyle \frac{\overline z_1-\overline z_4}{(\overline z_2+\overline z_3)^{-1}} \smallskip \\
* & * & (z_1-z_2)g_{42} & (\overline z_1-\overline z_2)\overline g_{42} & (z_1-z_3)g_{43} & (\overline z_1-\overline z_3)\overline g_{43} & (z_1-z_4)g_{44} & (\overline z_1-\overline z_4)\overline g_{44} \smallskip \\
* & * & (z_1-z_2)g_{52} & (\overline z_1-\overline z_2)\overline g_{52} & (z_1-z_3)g_{53} & (\overline z_1-\overline z_3)\overline g_{53} & (z_1-z_4)g_{54} & (\overline z_1-\overline z_4)\overline g_{54} \end{matrix}\right|,
\end{multline*}
where \( g_{ji}(z_1,z_2,z_3,z_4) := (z_1-z_i)^{-1}(f_{ji}-f_{j1})(z_1,z_2,z_3,z_4) \), \( j\in\{4,5\} \) and \( i\in\{2,3,4\} \). Hence,
\[
{\rm Jac}= \frac{|z_1-z_2|^2|z_1-z_3|^2|z_1-z_4|^2}{2\ic}\left|\begin{matrix}
 1 & 0 & 1 & 0 & 1 & 0 \\
 0 & 1 & 0 & 1 & 0 & 1 \\
 z_3+z_4 & 0 & z_2+z_4 & 0 & z_2+z_3 & 0 \\
 0 & \overline z_3+\overline z_4 &  0 & \overline z_2+\overline z_4 & 0 & \overline z_2+\overline z_3 \\
 g_{42} & \overline g_{42} & g_{43} & \overline g_{43} & g_{44} & \overline g_{44} \\
  g_{52} & \overline g_{52} & g_{53} & \overline g_{53} & g_{54} & \overline g_{54}
 \end{matrix}\right|.
\]
Absolutely analogous computation now implies that
\[
{\rm Jac} = \frac{|z_1-z_2|^2|z_1-z_3|^2|z_1-z_4|^2|z_2-z_3|^2|z_2-z_4|^2}{2\ic}\left|\begin{matrix}
1 & 0 & 1 & 0 \\
0 & 1 & 0 & 1 \\
h_{43} & \overline h_{43} & h_{44} & \overline h_{44} \\
h_{53} & \overline h_{53} & h_{54} & \overline h_{54}
 \end{matrix}\right|,
\]
where \( h_{ji}(z_1,z_2,z_3,z_4) := (z_2-z_i)^{-1}(g_{ji}-g_{j2})(z_1,z_2,z_3,z_4) \), \( j\in\{4,5\} \) and \( i\in\{3,4\} \). The above expression immediately yields that
\begin{equation}
\label{jac3}
{\rm Jac} = \frac{\prod_{i<j}|z_i-z_j|^2}{2\ic}\left|\begin{matrix} k_4 & \overline k_4 \\ k_5 & \overline k_5 \end{matrix}\right| = \prod_{i<j}|z_i-z_j|^2\im\big(k_4\overline k_5\big),
\end{equation}
where \( k_j(z_1,z_2,z_3,z_4) := (z_3-z_4)^{-1}(g_{j4}-g_{j3})(z_1,z_2,z_3,z_4) \), \( j\in\{4,5\} \). Finally, let
\[
w(z):=\sqrt{(z-z_1)(z-z_2)(z-z_3)(z-z_4)}
\]
be a branch such that \( w(z) = z^2 + \mathcal O(z) \) as \( z\to\infty \) with branch cuts \( \gamma_{12} \) and \( \gamma_{34} \) that are bounded, disjoint, and smooth, and where \( \gamma_{ij} \) connects \( z_i \) to \( z_j \). Further, select a smooth arc \( \gamma_{32} \) disjoint (except for the endpoints) from the previous two. Set
\begin{equation}
\label{jac4}
f_4(z_1,z_2,z_3,z_4) := 4\int_{\gamma_{12}}w(z)\dd z \qandq f_5(z_1,z_2,z_3,z_4) := 4\int_{\gamma_{32}}w(z)\dd z,
\end{equation}
where we integrate \( w(z) \) on the positive side of \( \gamma_{12} \). Let \( O\subset\{z_i\neq z_j,~i\neq j,~i,j\in\{1,2,3,4\} \}  \) be a domain such that there exist arcs \( \gamma_{ij}(z_1,z_2,z_3,z_4) \) with the above properties for each \( (z_1,z_2,z_3,z_4)\in O \), which, in addition, possess parameterizations that depend continuously on each variable \( z_1,z_2,z_3,z_4 \). Then the functions \( f_j(z_1,z_2,z_3,z_4) \), \( j\in\{4,5\} \), are analytic in each variable \( z_i \) for \( (z_1,z_2,z_3,z_4)\in O \). Furthermore, it can be readily computed that
\[
g_{ji} = 2\int\frac{w(z)\dd z}{(z-z_1)(z-z_i)}, \quad h_{ji} = -2\int \frac{w(z)\dd z}{(z-z_1)(z-z_2)(z-z_i)}, \qandq k_j = 2\int\frac{\dd z}{w(z)},
\]
where the integrals are taken over \( \gamma_{12} \) when \( j=4 \) and \( \gamma_{32} \) when \( j=5 \). Trivially, \eqref{jac3} can be rewritten as
\begin{equation}
\label{jac5}
{\rm Jac} = 4\prod_{i<j}|z_i-z_j|^2\im\left(\int_{\gamma_{12}}\frac{\dd z}{w(z)}\overline{\int_{\gamma_{32}}\frac{\dd z}{w(z)}}\right).
\end{equation}

Now, consider the Riemann surface \( \RS := \big\{ \z:=(z,w):~w^2=(z-z_1)(z-z_2)(z-z_3)(z-z_4)\big\} \). Denote by \( \pi:\RS\to\overline\C \) the natural projection \( \pi(\z) =z \) and write \( w(\z) \) for a rational function on \( \RS \) such that \( \z = (z,w(\z)) \). Let \( \boldsymbol\beta := \pi^{-1}(\gamma_{12}) \) and \( \boldsymbol\alpha := \pi^{-1}(\gamma_{32}) \). Orient these cycles so that
\[
2\int_{\gamma_{12}}\frac{\dd z}{w(z)} = \oint_{\boldsymbol \beta}\frac{\dd z}{w(\z)} \qandq 2\int_{\gamma_{32}}\frac{\dd z}{w(z)} = \oint_{\boldsymbol \alpha}\frac{\dd z}{w(\z)}.
\]
Observe that the cycles \( \boldsymbol\alpha,\boldsymbol\beta \) form the right pair at the point of their intersection and that \( \RS\setminus\{\boldsymbol\alpha\cup\boldsymbol\beta\} \) is simply connected. Hence, the cycles \( \boldsymbol\alpha,\boldsymbol\beta \) form a homology basis on \( \RS \). Since the genus of \( \RS \) is \( 1 \), it has a unique (up to multiplication by a constant) holomorphic differential. It is quite easy to check that this differential is  \( \dd z/w(\z) \). Hence, we get from \eqref{jac5} that
\begin{equation}
\label{jac6}
{\rm Jac} = \prod_{i<j}|z_i-z_j|^2\im\left(\oint_{\boldsymbol \beta}\frac{\dd z}{w(\z)}\overline{\oint_{\boldsymbol \alpha}\frac{\dd z}{w(\z)}}\right)>0
\end{equation}
when \( (z_1,z_2,z_3,z_4) \in O \), where the last inequality was shown by Riemann.

Now, let \( Q(z;t)=\frac14(z-a_1(t))(z-b_1(t))(z-a_2(t))(z-b_2(t)) \) be the polynomial from Theorem~\ref{fundamental}. It can be easily deduced from \eqref{em5} that
\begin{equation}
\label{jac7}
\left\{
\begin{array}{llr}
f_1(a_1,b_1,a_2,b_2) & = & 0, \smallskip \\
f_2(a_1,b_1,a_2,b_2) & = & -2t, \smallskip \\
f_3(a_1,b_1,a_2,b_2) & = & -4,
\end{array}
\right.
\end{equation}
where the functions \( f_i(z_1,z_2,z_3,z_4) \), \( i\in\{1,2,3\} \), are given by \eqref{jac2}. 

Fix \( t^*\in O_\mathsf{two-cut} \) and let \( \delta^*>0 \) be small enough so that all four disks \( \{|z-z_i^*|\leq \delta^*\} \) are disjoint, where \( z_1^*=a_1(t^*) \), \( z_2^*=b_1(t^*) \), \( z_3^*=a_2(t^*) \), and \( z_4^*=b_2(t^*) \). Let \( O := \{(z_1,z_2,z_3,z_4):|z_i-z_i^*|< \delta^*\} \), \( s_i^* \) be the point of intersection of \( \{|z-z_i^*|=\delta^*\} \) and \( J_{t^*} \), \( i\in\{1,2,3,4 \} \), and \( u_i^* \) be the point of intersection of \( \{|z-z_i^*|=\delta^*\} \) and \( \Ga(b_1(t^*),a_2(t^*)) \), \( i\in\{2,3\} \), where, as usual, \( \Ga(a,b) \) is the subarc of the trajectories of \( \varpi_t(z) \) connecting \( a \) and \( b \). Then we can choose \( \gamma_{12} = [z_1,s_1^*] \cup \Ga(s_1^*,s_2^*)\cup [s_2^*,z_2] \), \( \gamma_{32} = [z_3,u_3^*] \cup \Ga(u_3^*,u_2^*)\cup [u_2^*,z_2] \), and \( \gamma_{34} = [z_3,s_3^*] \cup \Ga(s_3^*,s_4^*)\cup [s_4^*,z_4] \), where \( [a,b]\) is the line segment connecting \( a \) and \( b \) in \( \C \), see Figure \ref{fig:54}. Clearly, the arcs \( \gamma_{ij} \) continuously depend on \( (z_1,z_2,z_3,z_4)\in O \). Now, the relations \eqref{widths} can be rewritten as
\begin{equation}
\label{jac9}
\re\left(f_4(a_1,b_1,a_2,b_2\big) \right)=0 \qandq \re\left(f_5\big(a_1,b_1,a_2,b_2\big) \right)=0
\end{equation}
with \( f_j(z_1,z_2,z_3,z_4) \), \( j\in\{4,5\} \), given by \eqref{jac4}, where we set \( w(z):=Q^{1/2}(z;t) \). 
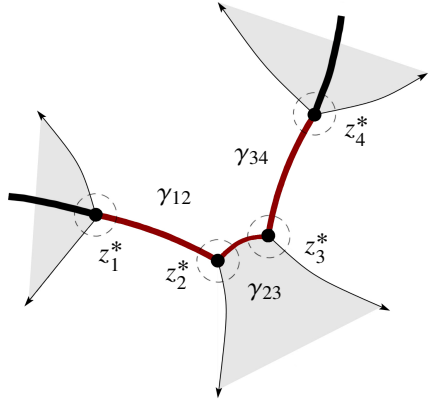
\begin{figure}[ht!]
\includegraphics[scale=.4]{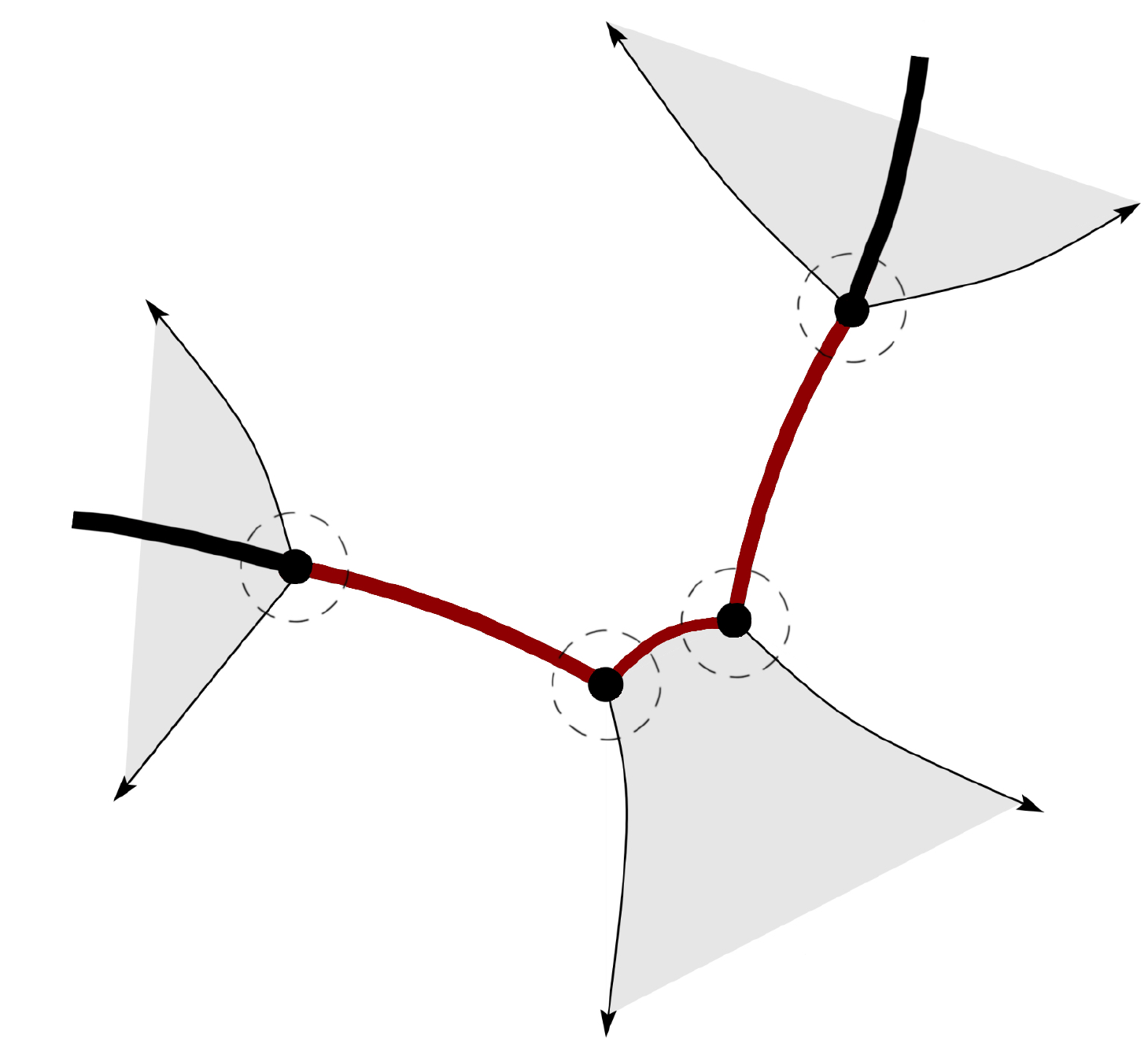} \begin{picture}(0,0)
			\put (-127,52){$z^*_1$}
			\put (-102, 45){$z^*_2$}
			\put (-50, 55){$z_3^*$}
			\put (-35, 99){$z_4^*$}
			\put (-105, 75){$\gamma_{12}$}
			\put (-71, 39){$\gamma_{23}$}
			\put (-76, 90){$\gamma_{34}$}
	\end{picture}
	\caption{\small A schematic illustration of $\gamma_{ij}$ and $z_i^*$.}
	\label{fig:54}
\end{figure} 
It follows from \eqref{jac6} and the implicit function theorem that there exists a neighborhood of \( t^* \) in which system \eqref{jac7} and \eqref{jac9} is uniquely solvable and the solution, say \( (a_1^*(t),b_1^*(t),a_2^*(t),b_2^*(t))\), is such that the real and imaginary parts of \( a_i^*(t),b_i^*(t) \) are real analytic functions of \( \re(t) \) and \( \im(t) \) for \( t \) in this neighborhood. These local solutions are unique only locally around the point \( (a_1(t^*),b_1(t^*),a_2(t^*),b_2(t^*))\) and we still need to argue that they do coincide with the zeros \( a_i(t), b_i(t)\) of \( Q(z;t) \) (of course, it holds that \( a_i^*(t^*)=a_i(t^*) \) and \( b_i^*(t^*)=b_i(t^*) \)).

In what follows, we always assume that \( t \) belongs to a disk centered at \( t^* \) of small enough radius so that the functions \( a_i^*(t),b_i^*(t) \) are defined and continuous in this disk. Let
\begin{equation}
\label{Q*}
Q^*(z;t) := \frac14(z-a_1^*(t))(z-b_1^*(t))(z-a_2^*(t))(z-b_2^*(t)) \quad \text{and} \quad \varpi_t^*(z) := -Q^*(z;t)\dd z^2.
\end{equation}
Further, let \( \mathcal U^*(z;t) \) be defined as in \eqref{em0} with \( Q(z;t) \) replaced by \( Q^*(z;t) \). For the moment, choose the branch cut for \( Q^*(z;t)^{1/2} \) as in the paragraph between \eqref{jac7} and \eqref{jac9}. Each function \( \mathcal U^*(z;t) \) is harmonic off the chosen branch cut and can be continued harmonically across it by \( -\mathcal U^*(z;t) \). Moreover, it follows immediately from their definition that the functions \( \mathcal U^*(z;t) \) are uniformly bounded above and below on any compact set for all considered values of the parameter \( t \). Thus, they converge to \( \mathcal U(z;t^*) \) locally uniformly in \( \C\setminus\{a_1(t^*),b_1(t^*),a_2(t^*),b_2(t^*)\} \) as \( t\to t^* \). Since the critical graph of \( \varpi_t^*(z) \) is the zero-level set of \( \mathcal U^*(z;t) \), it converges to the critical graph of \( \varpi_{t^*}(z) \) in any disk \( \{|z|<R\} \). Due to relations \eqref{jac9}, the argument at the beginning of Section~\ref{ss:53} also shows that the critical graph of  \( \varpi_t^*(z) \) has three short critical trajectories, which, due to uniform convergence, necessarily connect \( a_1^*(t) \) to \( b_1^*(t) \), \( b_1^*(t) \) to \( a_2^*(t) \), and \( a_2^*(t) \) to \( b_2^*(t) \) (the disk around \( t^* \) can be decreased if necessary). Thus, arguing as in Section~\ref{ss:53} and using uniform convergence, we can show that Figure~\ref{fig:tc} also schematically represents the critical and critical orthogonal graphs of \( \varpi_t^*(z) \). Moreover, let us now take the branch cut for \( Q^*(z;t)^{1/2} \), say \( J_t^* \), along the short critical trajectories of \( \varpi_t^*(z) \) connecting \( a_i^*(t) \) to \( b_i^*(t) \). Then the shading on Figure~\ref{fig:tc} corresponds to regions where \( \mathcal U^*(z;t) \) is positive (white) and negative (gray).

Define \( \Ga_t^* \) to be the union of the critical orthogonal trajectory of \( \varpi_t^*(z) \) that connects infinity to \( a_1^*(t) \), its short critical trajectories, and the critical orthogonal trajectory that connects \( b_2^*(t) \) to infinity. Orient it so that the positive direction proceeds from \( a_1^*(t) \) to \( b_2^*(t) \). Let the measures \( \mu_t^* \) be given by \eqref{em6} with \( Q(z;t) \) replaced by \( Q^*(z;t) \) and \( J_t \) replaced by \( J_t^* \). Clearly, each \( \mu_t^* \) is a positive measure. Moreover, it has a unit mass by the Cauchy theorem and since \( Q^*(z;t)^{1/2} = (z^2 - t)/2 + 1/z + \mathcal O(1/z^2) \) due to \eqref{jac7}, see also \eqref{ts2}. Thus, it holds that
\[
F^*(z;t) := Q^*(z;t)^{1/2} + \frac{V^\prime(z;t)}2 - \int\frac{\dd\mu_t^*(s)}{z-s} = \mathcal O\left(z^{-2}\right) 
\]
as \( z\to\infty \) and \( F^*(z;t) \) is holomorphic in \( \overline\C\setminus J_t^* \). It follows from the well known behavior of Cauchy integrals of smooth densities, see \cite[Section~I.8]{Gakhov}, that the traces of \( F^*(z;t) \) on \( J_t^* \) are bounded. It further follows from the Sokhotski-Plemelj formulae, see \cite[Section~I.4]{Gakhov}, that
\begin{multline*}
F_+^*(s;t) - F_-^*(s;t) \\
= Q_+^*(s;t)^{1/2} - Q_-^*(s;t)^{1/2} - \left(\int\frac{Q_+^*(w;t)^{1/2}}{w-z}\frac{\dd w}{\pi\ic}\right)_+  + \left( \int\frac{Q_+^*(w;t)^{1/2}}{w-z}\frac{\dd w}{\pi\ic}\right)_- \\
 \\=  Q_+^*(s;t)^{1/2} - Q_-^*(s;t)^{1/2} - 2Q_+^*(s;t)^{1/2} \equiv 0
\end{multline*}
for \( s\in J_t^* \). Hence, \( F^*(z;t) \) is an entire function and therefore is identically zero. This observation, in particular, yields that
\[
\mathcal U^*(z;t) := \mathrm{Re}\left(2\int_{b_2^*(t)}^z Q^*(s;t)^{1/2}\dd s \right) = \ell_t^* - \mathrm{Re}(V(z;t)) - 2U^{\mu_t^*}(z)
\]
for some constant \( \ell_t^* \), see also \eqref{em0}. Since \( \mathcal U^*(z;t) \) can be harmonically continued across \( J_t^* \) by \( -\mathcal U^*(z;t) \), we get that \( \mu_t^* \) satisfies \eqref{em4}; that is \( J_t^* \) has the S-property in the field \( \mathrm{Re}(V(z;t)) \). Since \( \Ga_t^*\in\mathcal T \), it follows from the uniqueness part of Theorem~\ref{fundamental}(2) that \( \mu_t^* = \mu_t \). In particular, \( a_i^*(t)=a_i(t) \) and \( b_i^*(t) = b_i(t) \), \( i\in\{1,2\} \).

Since any compact subset of \( O_\mathsf{two-cut} \) can be covered by finitely many disks where the above considerations hold, the functions \( a_i(t),b_i(t) \) continuously depend on \( t\in O_\mathsf{two-cut} \) and, moreover, their real and imaginary parts are real analytic functions of \( \mathrm{Re}(t) \) and  \( \mathrm{Im}(t) \).

Let us now show that at no point in \( O_\mathsf{two-cut} \) are the functions \( a_i(t),b_i(t) \) analytic in~\( t \). 
Suppose, for the sake of contradiction, that one of these functions, say $b_2(t)$, is analytic in $t$ at some point $t_0\in  O_\mathsf{two-cut}$. Then all the other endpoint functions, $a_1(t)$, $b_1(t)$, and $a_2(t)$, are analytic in $t$ at $t_0$ as well. Indeed, from \eqref{jac2} and \eqref{jac7} we have that
\[
\left\{
\begin{aligned}
	&a_1+b_1+a_2=-b_2,\\
	&a_1b_1+a_1a_2+b_1a_2=-2t-b_2(a_1+b_1+a_2)=-2t+b_2^2,\\
	&a_1b_1a_2=-4-b_2(a_1b_1+a_1a_2+b_1a_2)
	=-4-b_2(-2t+b_2^2)=-4+2tb_2-b_2^3.
\end{aligned}
\right.
\]
Hence $a_1(t),b_1(t),a_2(t)$ are roots of the cubic equation
\[
z^3+b_2z^2+(-2t+b_2^2)z+(4-2tb_2+b_2^3)=0.
\]
Since $a_1(t),b_1(t),a_2(t)$ are pairwise different, they analytically depend on the coefficients of the cubic equation. That is, they are analytic functions of $t$ and $b_2$. This implies that $a_1(t)$, $b_1(t)$, and $a_2(t)$ are analytic in $t$ at $t_0$. The analyticity of $a_1(t)$, $b_1(t)$, $a_2(t)$, and $b_2(t)$ yields that the integral
\[
\frac1{\pi\ic}\int_{\Gamma_t[a_1(t),b_1(t)]} Q_+^{1/2}(s;t)\dd s
\]
is also an analytic function of \( t \) at $t_0$. However, the integral above is equal to \( -\omega(t) \), see \eqref{tau-omega}, which is a real number, see \eqref{em6}. Hence, \( \omega(t) \) must be constant in \( U \).  On the other hand, it follows from \eqref{ts8a}, the proof of which in the next subsection is independent of the current considerations, that \( \omega(t) \) is not constant in \( O_\mathsf{two-cut} \). We also have already shown that \( \omega(t) \) is real analytic in \( \re(t) \) and \( \im(t) \) and therefore cannot be locally constant. This contradiction proves that the endpoints  $a_1(t)$, $b_1(t)$, $a_2(t)$, and $b_2(t)$ are not analytic functions of $t$ at any point in $ O_\mathsf{two-cut}$.

\subsection{Proof of \eqref{ts8a}}
\label{ss:55}

It follows from \cite[Theorem 5.11]{BT} that if \( t \) remains in a bounded set, so do the zeros of \( Q(z;t) \). Fix $t^* \in \partial O_{\mathsf{two-cut}}$ and let \( \{t_m\} \) be a sequence such that \( t_m\to t^* \) as \( m\to\infty \). Restricting to a subsequence if necessary, we see that there exist \( a_i^*,b_i^* \) such that \( e(t_m)\to e^* \) as \( m\to\infty \), where \( e\in\{ a_1,b_1,a_2,b_2 \} \). Clearly, polynomials \( Q(z;t_m) \) converge uniformly on compact subsets of \( \C \) to \( Q^*(z) \), the polynomial with zeros \( a_i^*,b_i^* \) and leading coefficient \( 1/4 \). Let \( \varpi^*(z) := -Q^*(z)\dd z^2 \). Then repeating the argument after \eqref{Q*}, we get that the critical graphs of \( \varpi_{t_m}(z) \) converge to the critical graph of \( \varpi^*(z) \) in any disk \( |z|<R \). Furthermore, as stated in \cite[Theorem~3.3]{Jenkins}, there exists \( R_t>0 \) such that every  trajectory of \( \varpi_t(z) \) entering \( \{ |z|>R_t \} \) necessarily remains in \( \{ |z|>R_t \} \)  and tends to infinity. Moreover, examination of the proof of \cite[Theorem~3.3]{Jenkins} also shows that one can take \( R_t = 2^{-1/2}\max_i\big\{|a_i(t)|,|b_i(t)|\big\} \) and that if a trajectory enters \( S_{(2k-1)\pi/3,\delta} \cap \{|z|>R_t\} \), \( k\in\{0,1,2\} \), it stays in this sector, see \eqref{sectorsS}. Put \( R^* := \sup_m R_{t_m} \). As all the differentials \( \varpi_{t_m}(z) \) have structurally ``the same'' critical graph, see Figure~\ref{fig:tc}, and all the trajectories of \( \varpi_{t_m}(z) \) and \( \varpi^*(z) \) entering \( S_{(2k-1)\pi/3,\delta} \cap \{|z|>R^*\} \) must remain there, tend to infinity, and not intersect, the unbounded critical trajectories of \( \varpi^*(z) \) behave like on Figure~\ref{fig:tc}.

Observe that the first and the third equations of \eqref{jac7}, see also \eqref{jac2}, must remain true for \( a_1^*,b_1^*,a_2^*,b_2^* \) as well. Hence, we cannot  simultaneously have that \( a_1^*=b_1^* \) and \( a_2^* = b_2^* \).  As short trajectories cannot form loops and their tangent vectors cannot become parallel at each critical point, the above considerations yield that the critical and therefore critical orthogonal graphs of \( \varpi^*(t) \) look like either on Figure~\ref{fig:tc}, or on Figure~\ref{s-curves2}, or as the graphs obtained by reflection across the line \( L_{2\pi/3} \) of the graphs on Figure~\ref{s-curves2}. Repeating the arguments at the end of Section~\ref{ss:54}, we see that \( \varpi^*(t) \) gives rise to an S-contour in \( \mathcal T \). The limits in \eqref{ts8a} now follows from the uniqueness of such a S-contour and Theorem~\ref{geometry1} (in particular, the critical graph \( \varpi_*(t) \) cannot look like as on Figure~\ref{fig:tc}).

\section{Functions \( \mathcal D(z;t) \) and \( \mathcal Q(z;t) \)}
\label{s:g}

In this section we prove Propositions~\ref{prop:Szego} and~\ref{prop:g} as well as discuss other properties of $\mathcal Q(z;t)$. We consider the parameter $t\in O_\mathsf{two-cut}$ to be fixed and stop indicating the dependence on $t$ of the various quantities appearing below whenever this does not introduce ambiguity and is convenient.

\subsection{Proof of Proposition~\ref{prop:Szego}}
\label{ss:Szego}

It follows from \eqref{cm2}, \eqref{em5}, and the choice of the branch of \( Q^{1/2}(z) \) that
\begin{equation}
\label{rootQ}
Q^{1/2}(z) = \frac{z^2-t}2 + \frac1z + \frac{\mu_1}{z^2} + \mathcal O\left(\frac1{z^3}\right).
\end{equation}
It further follows from the choice of the constant \( \varsigma, d_1 \) in \eqref{Ct} that
\[
z + \int_{I_t}\frac{3\varsigma}{s-z}\frac{\dd s}{Q^{1/2}(s)} = z - \frac{2t}z - \frac{3 \varsigma d_1}{z^2}  - \frac{3\varsigma d_2}{z^3} + \mathcal O\left(\frac1{z^4}\right)
\]
as \( z\to\infty \). Thus, the analyticity properties of \( \mathcal D(z) \) as well as \eqref{D-exp} now follow from the fact that the product of the above functions behaves like 
\[
-\frac{3V(z)}2 + 1 - \frac{3\varsigma d_1}2 + \frac{t^2+\mu_1-3\varsigma d_2/2}z + \mathcal O\left(\frac1{z^2}\right)
\]
as \( z\to\infty \). Moreover, since \( Q_+^{1/2}(s) = - Q_-^{1/2}(s) \) for \( s\in J_t \), we get the first relation in \eqref{Djumps}. The second relation in \eqref{Djumps} follows from Plemelj-Sokhotski formula
\[
\left(\int_{I_t}\frac{\varsigma}{s-z}\frac{\dd s}{Q^{1/2}(s)}\right)_+ -  \left(\int_{I_t}\frac{\varsigma}{s-z}\frac{\dd s}{Q^{1/2}(s)}\right)_- =  \frac{2\pi\ic\varsigma}{Q^{1/2}(z)}, \quad z\in I_t^\circ.
\]

\subsection{Proof of Proposition~\ref{prop:g}} 
\label{ss:g}

Since the arc \( I_t \) is homologous to  the short critical trajectory of \( -Q(z)\dd z^2 \) connecting \( b_1 \) and \( a_2 \) and \( J_{t,1} = \Ga[a_1,b_1]  \) is such a trajectory, see Figure~\ref{fig:tc}  and Conventions~\ref{Gamma-arcs} and~\ref{i-t}, the constants \( \tau,\omega \) are indeed real. Let
\begin{equation}
\label{oc4}
g(z) := \int\log(z-s)\mathrm d\mu_t(s), \quad z\in\C\setminus \Ga\big(e^{\pi\ic}\infty,b_2\big],
\end{equation}
where we take the principal branch of $\log(\cdot-s)$ holomorphic outside of $\Ga\big(e^{\pi\ic}\infty,s\big]$ and $\mu_t$ is the equilibrium measure defined in \eqref{em6}. It follows directly from definition \eqref{oc4} that
\[
\partial_z g(z) = \int\frac{\mathrm d\mu_t(s)}{z-s},
\]
where, as usual, $\partial_z:=(\partial_x-\mathrm i\partial_y)/2$. Therefore, it can be deduced from \eqref{em2} and \eqref{em5} that
\begin{equation}
\label{g1}
g(z) = \frac{V(z)-\ell_*}2 + \int_{b_2}^zQ^{1/2}(s)\mathrm ds = \frac{V(z)-\ell_*}2 + \mathcal Q(z),
\end{equation}
where, as usual, we take the branch $Q^{1/2}(z)=\frac12z^2+\mathcal{O}(z)$, $\ell_*$ is a constant such that the equality holds at \( b_2 \) (notice that $\re(\ell_*)=\ell$, see \eqref{em2}), and \( \mathcal Q(z) \) is given by \eqref{mQ}. Property \eqref{oc5} clearly follows from \eqref{oc4} and \eqref{g1}. In the view of \eqref{g1}, let us define
\begin{equation}
\label{g2}
\phi_e(z) := 2\int_e^zQ^{1/2}(s)\dd s, \quad e\in\{a_1,b_1,a_2,b_2\}, 
\end{equation}
holomorphically in $\C\setminus\Gamma(e^{\pi\ic}\infty,b_2]$ when \( e=b_2 \), in \( \C\setminus\Ga[a_1,e^{\pi\ic/3}\infty) \) when \( e=a_1 \), and in  $\C\setminus\big(\Gamma(e^{\pi\ic}\infty,b_1)\cup\Gamma(a_2,e^{\pi\ic/3}\infty)\big)$ when \( e\in\{b_1,a_2\} \). Clearly, \( \phi_{b_2}(z) = 2\mathcal Q(z) \). One can readily check that
\begin{equation}
\label{g3}
\phi_{b_2}(z) = 
\left\{
\begin{array}{l}
\phi_{a_2}(z) \pm 2\pi\ic(1-\omega), \medskip \\
\phi_{b_1}(z) \pm 2\pi\ic(1-\omega) + 2\pi\ic\tau, \medskip \\
\phi_{a_1}(z) \pm 2\pi\ic + 2\pi\ic\tau,
\end{array}
\right. \quad z\in\C\setminus\Gamma,
\end{equation}
where the plus sign is used if $z$ lies to the left of $\Ga$ and the minus sign if $z$ lies to the right of it, and
\begin{equation}
\label{g4}
\phi_{b_2\pm}(s) = 
\left\{
\begin{array}{ll}
\displaystyle \pm2\pi\ic\mu_t\big(\Ga[s,b_2]\big), & s\in\Ga(a_2,b_2), \medskip \\
\displaystyle \pm 2\pi\ic\mu_t\big(\Ga[s,b_2]\big) +  2\pi\ic\tau, & s\in\Ga(a_1,b_1).
\end{array}
\right. 
\end{equation}
The jump relations in \eqref{gjumps} now easily follow from \eqref{g3} and \eqref{g4}. 

For future use let us record that \eqref{g1}, \eqref{g3}, and \eqref{g4} imply that
\begin{equation}
\label{g5}
g_+(s) - g_-(s) = \left\{
\begin{array}{rl}
0, & s\in\Gamma(b_2,e^{\pi\ic/3}\infty), \medskip \\
\pm\phi_{b_2\pm}(s), & s\in \Gamma(a_2,b_2), \medskip \\
2\pi\ic(1-\omega), & s\in\Gamma(b_1,a_2), \medskip \\
\pm\big(\phi_{b_2\pm}(s) - 2\pi\ic\tau\big), & s\in \Gamma(a_1,b_1), \medskip \\
2\pi\ic, & s\in\Gamma(e^{\pi\ic}\infty,a_1),
\end{array}
\right.
\end{equation}
and that
\begin{equation}
\label{g6}
g_+(s) + g_-(s) - V(s) + \ell_* = \left\{
\begin{array}{rl}
\phi_{b_2}(s), & s\in\Gamma(b_2,e^{\pi\ic/3}\infty), \medskip \\
0, & s\in\Gamma(a_2,b_2), \medskip \\
\phi_{a_2}(s), & s\in\Gamma(b_1,a_2),\medskip \\
2\pi\ic\tau, & s\in\Gamma(a_1,b_1), \medskip \\
\phi_{a_1}(s)+2\pi\ic\tau, & s\in\Gamma(e^{\pi\ic}\infty,a_1).
\end{array}
\right.
\end{equation}

\subsection{Local Analysis at \( e\in\{a_1,b_1,a_2,b_2\} \)}

Given \( e\in\{a_1,b_1,a_2,b_2\} \), let
\begin{equation}
\label{g7}
U_e := \big\{z:~|z-e|<\delta_e\rho(t)/3\big\},
\end{equation}
where \( \delta_e \in(0,1] \) to be adjusted later and we shall specify the function \( \rho(t) \) at the end of this subsection. Set
\begin{equation}
\label{g8}
J_e:=U_e\cap J_t \quad \text{and} \quad I_e:=U_e\cap(\Ga_t\setminus J_t),
\end{equation}
where the arcs $J_e$ and $I_e$ inherit their orientation from $\Ga_t$ and we assume that the value of \( \rho(t) \) is small enough so that these arcs are connected and that $I_e$ is a subarc of the orthogonal critical trajectory of $-Q(z)\dd z^2$ emanating from $e$ (see the remarks about \( \Ga_t \) at the very beginning of Section~\ref{s:main}). The latter fact and Theorem~\ref{geometry2} yield that
\begin{equation}
\label{g9} 
\phi_e(s)<0, \quad s\in I_e,
\end{equation}
see Figure~\ref{fig:tc}. In fact, the same reasoning shows that \eqref{g9} holds not only on \( I_e \), but on \( \Ga(e^{\pi\ic}\infty,a_1) \) when \( e=a_1 \), on \( \Ga(b_2,e^{\pi\ic/3}\infty) \) when \( e=b_2 \), and, for \( \re(\phi_e(z)) \) on \( \Ga(b_1,a_2) \) when \( e\in\{b_1,a_2\} \) (observe that these functions are also monotone on the respective arcs).   Furthermore, each function $\phi_e(z)$ is analytic in $U_e\setminus J_e$ and its traces on $J_e$ satisfy
\begin{equation}
\label{g10}
\phi_{e\pm}(s) = \pm2\pi\ic\nu_e\mu(J_{s,e})  = 2\pi e^{\pm 3\pi\ic\nu_e/2}\mu(J_{s,e}),
\end{equation}
where $J_{s,e}$ is the subarc of $J_e$ with endpoints $e$ and $s$,
\begin{equation}
\label{g11}
\nu_e :=
\left\{
\begin{array}{rl}
1, & e\in\{b_1,b_2\}, \medskip \\
-1, & e\in\{a_1,a_2\},
\end{array}
\right.
\end{equation}
and the second equality follows from \eqref{em6} and \eqref{g2}.  Since \( |\phi_e(z)|\sim |z-e|^{3/2} \) as \( z\to e \), it follows from \eqref{g9} and \eqref{g10} that we can define an analytic branch of $(-\phi_e)^{2/3}(z)$ in \( U_e \) that is positive on $I_e$ and satisfies \( (-\phi_e)^{2/3}(s) = -\big(2\pi \mu_t(J_{s,e})\big)^{2/3} \), \( s\in J_e \). Since $(-\phi_e)^{2/3}(z)$ has a simple zero at \( e \), it is conformal in $U_e$ for all radii small enough. Altogether, $(-\phi_e)^{2/3}(z)$ maps $e$ into the origin, is conformal in \( U_e \), and satisfies
\begin{equation}
\label{g12}
\left\{
\begin{array}{lcl}
(-\phi_e)^{2/3}(J_e) &\subset& (-\infty,0), \medskip \\
(-\phi_e)^{2/3}(I_e) &\subset& (0,\infty).
\end{array}
\right.
\end{equation}
Furthermore, if we define $(-\phi_e)^{1/6}(z)$ to be holomorphic in $U_e\setminus J_e$ and positive on $I_e$, then
\begin{equation}
\label{g13}
(-\phi_e)^{1/6}_+(s) = \nu_e \ic(-\phi_e)^{1/6}_-(s), \quad s\in J_e.
\end{equation}

To specify \( \rho(t) \), let \( \rho_e(t) \) be the radius of the largest disk around \( e \) for which \( J_e,I_e \) are connected and in which \( (-\phi_e)^{2/3}(z) \) is conformal. Observe that the disk around \( e \) of radius \( \rho_e(t) \) cannot contain other endpoints of \( J_t \) besides \( e \). We set \( \rho(t):=\min_e\{\rho_e(t)\}\). Then the disks \( U_e \) in \eqref{g7} are necessarily disjoint. Observe also that \( \rho(t) \) is non-zero for all \( t\in O_\mathsf{two-cut} \) and continuously depends on \( t \) due to continuous dependence on \( t \) of \( \phi_e(z) \), which in itself follows from Theorem~\ref{geometry2} and \eqref{g2}.

\section{Functions \(\Theta_n(z;t) \)}
\label{sec:An}

In this section we prove Proposition~\ref{prop:l}--\ref{prop:N} as well as discuss some related results. As in the previous section, we omit indicating the explicit dependence on \( t \) whenever convenient.

\begin{figure}[h!]
\begin{center}
\begin{tikzpicture}
    \coordinate (0) at (0,0);
    \coordinate (a) at (-3,-0.25);
    \coordinate (b) at (4,-0.25);
    \coordinate (c) at (-3,-2.5);
    \coordinate (d) at (4,-2.5);

   \filldraw (-0.65,0.75) circle (1.5pt) node[below=8pt,right=0pt] {$a_1$};
   \filldraw (0.9,0.25) circle (1.5pt) node[below=2pt,left=0pt] {$b_1$};
   \filldraw (1.9,0.25) circle (1.5pt) node[below=2pt,right=2pt] {$a_2$};
   \filldraw (3.5,0.9) circle (1.5pt) node[above=8pt,right=0pt] {$b_2$};
   \filldraw (-0.65,-1.5) circle (1.5pt) node[below=8pt,right=0pt] {$a_1$};
   \filldraw (0.9,-2) circle (1.5pt) node[above=8pt,right=-2pt] {$b_1$};
   \filldraw (1.9,-2) circle (1.5pt) node[above=2pt,left=0pt] {$a_2$};
   \filldraw (3.5,-1.35) circle (1.5pt) node[above=8pt,right=0pt] {$b_2$};
  
    \node at (5.25,1.25)   {$D^{(0)}$};
    \node at (5.25,-1)   {$D^{(1)}$};
    \node at (0.5,1.25)   {\textcolor{red}{$\beta$}};
    \node at (0.25,-2.25)   {\textcolor{blue}{$\alpha$}};
    \begin{scope}[very thick,decoration={
    markings,
    mark=at position 0.4 with {\arrow[scale=1.5]{>}}}
    ]
    \draw[thick] (a)--(b);
    \draw[thick] (c)--(d);
    \draw[thick] (a) -- (-1,1.25);
    \draw[thick] (b) -- (6,1.25);
    \draw[thick] (c) -- (-1,-1);
    \draw[thick] (d) -- (6,-1);
    \draw[thick,postaction={decorate}] (-0.65,0.75) .. controls (-0.25,0.8) and (0.25,0.65) ..node[near start, above=2pt] {} (0.9,0.25);
    \draw[thick,postaction={decorate}] (1.9,0.25) .. controls (2.25,0.65) and (2.5,0.75) ..node[near start, above=2pt] {} (3.5,0.9);
 \draw[thick,postaction={decorate}] (-0.65,-1.5) .. controls (-0.25,-1.45) and (0.25,-1.6) ..node[near start, above=2pt] {} (0.9,-2);
    \draw[thick,postaction={decorate}] (1.9,-2) .. controls (2.25,-1.6) and (2.5,-1.5) ..node[near start, above=2pt] {} (3.5,-1.35);
    \draw[thick,dashed] (-0.65,0.75)--(-0.65,-1.5);
    \draw[thick,dashed] (0.9,0.25)--(0.9,-2);
    \draw[thick,dashed] (1.9,0.25)--(1.9,-2);
    \draw[thick,dashed] (3.5,0.9)--(3.5,-1.35);
    \draw[red,thick,dashed,postaction={decorate}] (-0.8,1).. controls (-0.2,1.3) and (1.7,0.6)..node[near start, above=2pt] {} (0.9,0);
    \draw[red,thick,dashed,postaction={decorate}] (0.9,0).. controls (-1,0) and (-1.25,0.6)..node[near start, above=2pt] {} (-0.8,1);
    \draw[blue,thick,dashed,postaction={decorate}] (2.3,0.65).. controls (2,1.1) and (0.7,1.1) ..node[near start, above=2pt] {} (0.5,0.5);
    \draw[blue,thick,dashed,postaction={decorate}] (0.5,-1.8).. controls (0.3,-2.5) and (2.7,-2.5) ..node[near start, above=2pt] {} (2.3,-1.7) ;
\end{scope}
\end{tikzpicture}
\end{center}
\caption{\small Schematic plot of the Riemann surface $\RS$ and the cycles $\boldsymbol\alpha$ and $\boldsymbol\beta$.}
\label{fig_surface}
\end{figure}
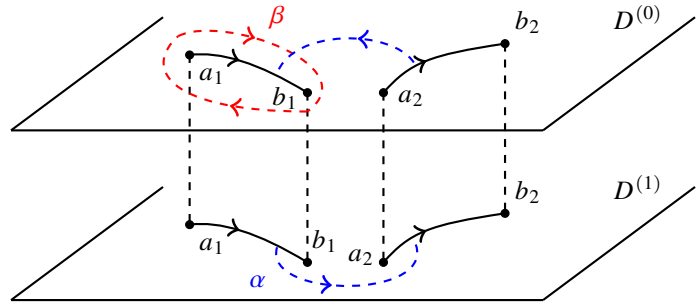

\subsection{Riemann Surface}
\label{ss:RS}

In this subsection  we discuss properties of the Riemann surface \( \RS \) that has already appeared between \eqref{jac5} and \eqref{jac6}. Once again, set
\begin{equation}
\label{RS}
\RS := \big\{ \z:=(z,w):~w^2=Q(z)\big\}.
\end{equation}
We denote by \( \pi:\RS\to\overline\C \) the natural projection \( \pi(\z) =z \) and by \( \cdot^* \) a holomorphic involution on \( \RS \) acting according to the rule \(\z^*=(z,-w) \). We use notation \( \z,\s,\boldsymbol a \) for points on \( \RS \) with natural projections \( z,s,a \). 

The function \( w(\z) \), defined by \( w^2(\z)=Q(z) \), is meromorphic on \( \RS \) with simple zeros at the ramification points \( \boldsymbol E_t =\{ \boldsymbol a_1,\boldsymbol b_1,\boldsymbol a_2,\boldsymbol b_2 \} = \pi^{-1}(E_t) \), double poles at the points on top of infinity, and is otherwise non-vanishing and finite.  Put 
\[
\boldsymbol\Delta:=\pi^{-1}(J_t) \quad \text{and} \quad \RS = D^{(0)}\cup\boldsymbol\Delta\cup D^{(1)},
\]
where the domains \( D^{(k)} \) project onto \( \overline\C\setminus J_t \) with labels chosen so that \( 2w(\z) = (-1)^kz^2 + \mathcal O(z) \) as \( \z \) approaches the point on top of infinity within \( D^{(k)} \). For \( z\in\overline\C\setminus J_t \) we let \( z^{(k)} \) stand for \( \z\in D^{(k)} \) with, as agreed, \( \pi(\z)=z \). We define a homology basis on \( \RS \) in the following way: we let 
\[ 
\boldsymbol\alpha := \pi^{-1}(I_t) \quad \text{and} \quad \boldsymbol\beta := \pi^{-1}(J_{t,1}),
\]
where \( \boldsymbol\alpha \) is oriented towards \( \boldsymbol b_1 \) within \( D^{(0)} \) and \( \boldsymbol\beta \) is oriented so that \( \boldsymbol\alpha,\boldsymbol\beta \) form the right pair at \( \boldsymbol b_1 \), see Figures~\ref{fig_surface} and~\ref{planar-model}. It also will be convenient to put \( \RS_{\boldsymbol\alpha,\boldsymbol\beta} := \RS\setminus\{\boldsymbol\alpha\cup\boldsymbol\beta\} \), which is simply connected.

\begin{figure}[!h]
\includegraphics[scale=0.3]{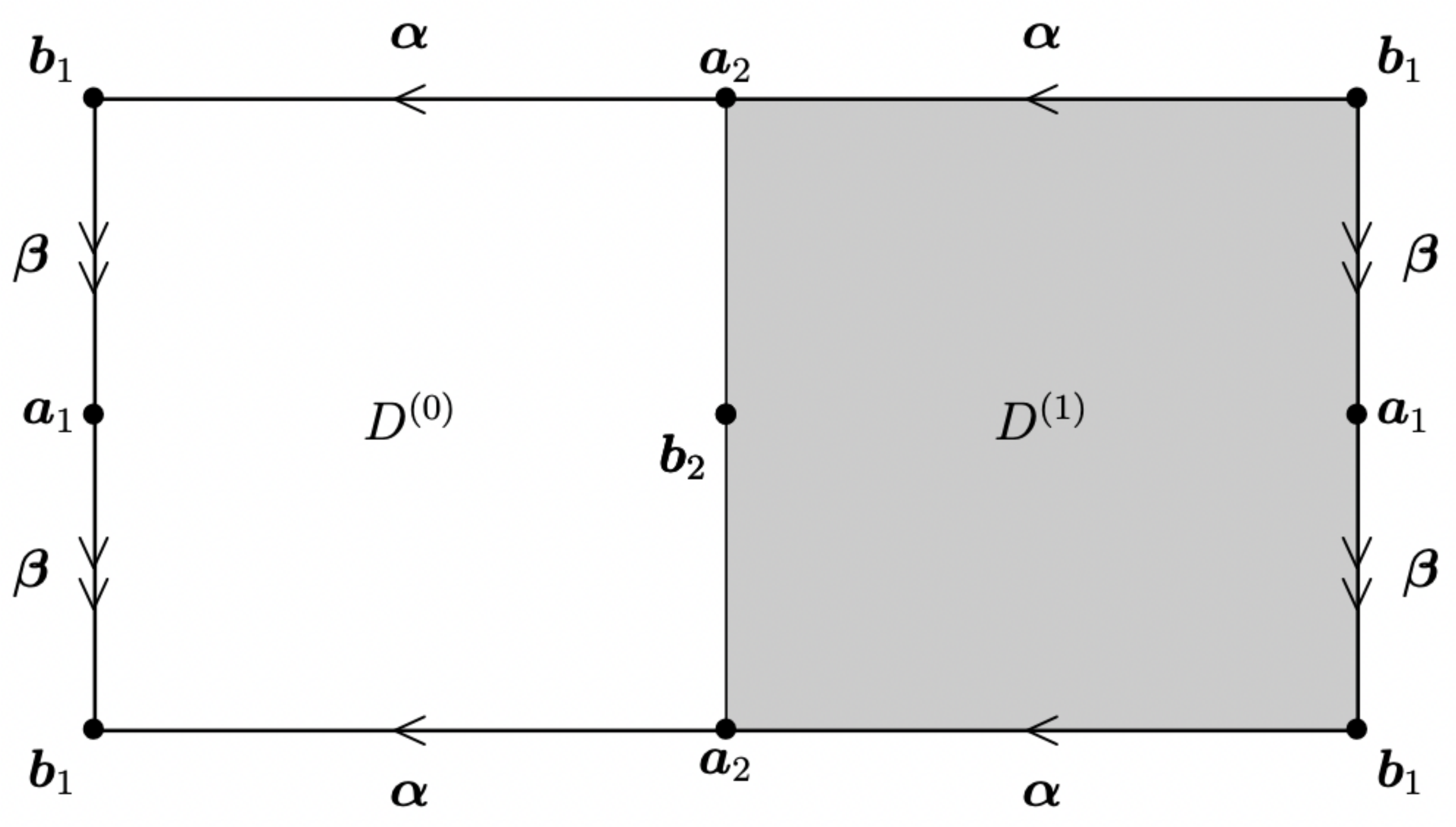}
 \caption{\small Planar model of the Riemann surface $\RS$ and the cycles $\boldsymbol \alpha$ and $\boldsymbol \beta$.}
\label{planar-model}
 \end{figure}
The surface $\RS$ has genus one. Thus, there exists a unique holomorphic differential on \( \RS \) normalized to have a unit period on \( \boldsymbol\alpha \), say \( \mathcal H \). In fact, it can be explicitly expressed as
\begin{equation}
\label{hol-diff}
\mathcal H(\z) = \left(\oint_{\boldsymbol\alpha} \frac{\dd s}{w(\s)} \right)^{-1} \frac{\dd z}{w(\z)}.
\end{equation}
With this notation, the choice of the homology basis, and the definition of \( w(\z) \), \eqref{BQ} becomes
\begin{equation}
\label{B}
\mathsf B = \oint_{\boldsymbol\beta}\mathcal H, \quad  \im(\mathsf B)>0, 
\end{equation} 
where the last inequality is a classical result of Riemann, see \cite{M07tata}.

Given the normalized holomorphic differential, we define
\begin{equation}
\label{ab1}
\mathfrak a(\z) := \int_{\boldsymbol b_2}^{\z}\mathcal H,
\end{equation}
where we restrict \( \z \) as well as the path of integration to $\RS_{\boldsymbol\alpha,\boldsymbol\beta}$. The function \( \mathfrak a(\z) \) is a holomorphic function in \( \RS_{\boldsymbol\alpha,\boldsymbol\beta} \) with continuous traces on $\boldsymbol\alpha,\boldsymbol\beta$ (away from the point of their intersection) that satisfy
\begin{equation}
\label{ab2}
\mathfrak a_+(\s) - \mathfrak a_-(\s) =  \left\{
\begin{array}{rl}
-\mathsf B, & \s\in\boldsymbol\alpha\setminus\{\boldsymbol b_1\}, \medskip \\
1, & \boldsymbol \s\in\boldsymbol\beta\setminus\{\boldsymbol b_1\},
\end{array}
\right.
\end{equation}
by the normalization of \( \mathcal H \) and the definition of \( \mathsf B \). Moreover, it continuously extends to \( \partial \RS_{\boldsymbol\alpha,\boldsymbol\beta} \), the topological boundary of \( \RS_{\boldsymbol\alpha,\boldsymbol\beta} \) (the rectangle on Figure~\ref{planar-model}). Notice that
\begin{equation}
\label{periods}
\tau = \frac1{2\pi\ic} \oint_{\boldsymbol\alpha} w(\s)\dd s \qandq \omega = -\frac1{2\pi\ic}\oint_{\boldsymbol\beta} w(\s)\dd s,
\end{equation}
where \( \tau,\omega \) are the constants from \eqref{tau-omega}. It readily follows from \eqref{ab2} and \eqref{periods} that
\[
\oint_{\partial \RS_{\boldsymbol\alpha,\boldsymbol\beta}} (w\mathfrak a)(\s)\dd s = \oint_{\boldsymbol\beta} w(\s)(\mathfrak a_+-\mathfrak a_-)(\s)\dd s  +  \oint_{\boldsymbol\alpha} w(\s)(\mathfrak a_+-\mathfrak a_-)(\s)\dd s = -2\pi\ic(\omega+\mathsf B\tau) ,
\]
where \( \partial \RS_{\boldsymbol\alpha,\boldsymbol\beta} \) is oriented so that \( \RS_{\boldsymbol\alpha,\boldsymbol\beta} \) remains on the left when \( \partial \RS_{\boldsymbol\alpha,\boldsymbol\beta} \) is traversed in the positive direction. On the other hand, the function \( (w\mathfrak a)(\z)\) is meromorphic in \( \RS_{\boldsymbol\alpha,\boldsymbol\beta} \) with only two singularities, both polar, at \( \infty^{(0)} \) and \( \infty^{(1)} \). Moreover, since \( w(\z^*)=-w(\z) \) and \( \mathfrak a(\z^*) = - \mathfrak a(\z) \), the residues at those poles coincide. Therefore, it holds that
\[
\omega + \mathsf B\tau = -\frac1{2\pi\ic} \oint_{\partial \RS_{\boldsymbol\alpha,\boldsymbol\beta}} (w\mathfrak a)(\s)\dd s  = -2\mathrm{res}_{z=\infty^{(0)}} (w\mathfrak a)(\z).
\]
Recall that \( w\big(z^{(0)}\big) = Q^{1/2}(z) \) and therefore it has expansion \eqref{rootQ} as \( z^{(0)}\to\infty^{(0)} \). This and \eqref{ab1} also allow us to deduce that
\[
\mathfrak a\big(z^{(0)}\big) = \mathfrak a\big(\infty^{(0)}\big) + \int_{\infty^{(0)}}^{z^{(0)}} \mathcal H = \mathfrak a\big(\infty^{(0)}\big) - \left(\oint_{\boldsymbol\alpha} \frac{\dd s}{w(\s)} \right)^{-1}\left(\frac2z+\frac{2t}{3z^3}+\mathcal O\left(\frac1{z^4}\right)\right)
\]
as \( z^{(0)}\to\infty^{(0)} \). Therefore, taking into account \eqref{Ct} with the definition of \( \boldsymbol\alpha \) and the symmetry \( \mathfrak a(\z^*) = - \mathfrak a(\z) \), we get that
\begin{equation}
\label{residue}
\varsigma + \omega + \mathsf B\tau = \varsigma + 2\mathfrak a\big(\infty^{(0)}\big) + \frac{4t}3\left(\oint_{\boldsymbol\alpha} \frac{\dd s}{w(\s)} \right)^{-1} = 2\int_{\boldsymbol b_2}^{\infty^{(0)}}\mathcal H = \int_{\infty^{(1)}}^{\infty^{(0)}}\mathcal H,
\end{equation}
where the path of integration lies entirely in \( \RS_{\boldsymbol\alpha,\boldsymbol\beta} \) and is involution-symmetric.

\subsection{Jacobi Inversion Problem}
\label{ss:JIP}

Denote by \( \mathsf{Jac}(\RS) := \C/ \{\Z+\mathsf B\Z\} \) the Jacobi variety of \( \RS \). We shall represent elements of \( \mathsf{Jac}(\RS) \) as equivalence classes \( [s] = \{s + j + m\mathsf B:j,m\in \Z\} \), where \( s\in\C \). One can readily deduce from \eqref{ab2} that \( \mathfrak a(\z) \), defined in \eqref{ab1}, is essentially a holomorphic bijection from \( \RS \) onto \( \mathsf{Jac}(\RS) \). Formally, we define this bijection, known as Abel's map, by
\begin{equation}
\label{Abel-map}
\z\in\RS \mapsto \left[\int_{\boldsymbol b_2}^\z\mathcal H\right]\in \mathsf{Jac}(\RS).
\end{equation}
Given \( s \in \C \), we denote by \( \z_{[s]} \) the unique point in \( \RS \) such that \( \left[\int_{\boldsymbol b_2}^{\z_{[s]}}\mathcal H\right] = [s] \). 

\begin{proposition}
\label{prop:Nt}
Let \( \tau,\omega \) be given by \eqref{tau-omega} and \( \varsigma \) by \eqref{Ct}. Further, let \( \{N_n\} \) be a sequence as in Theorem~\ref{global2}. Denote by \( \z_{n,k}=\z_{n,k}(t) \) the unique solution \( \z_{[s_{n,k}(t)]}  \) of the Jacobi inversion problem with
\begin{equation}
\label{tc4}
s_{n,k}(t) := \int_{\boldsymbol b_2}^{p^{(k)}}\mathcal H + (n-N_n)\varsigma + (\omega+\mathsf B\tau) n, \quad p = p(t) := \frac{b_1b_2-a_1a_2}{(b_2-a_2)+(b_1-a_1)},
\end{equation}
\( k\in\{0,1\} \). Then for any subsequence \( \N_* \) the point \( \infty^{(0)} \) is a topological limit point of \( \{\z_{n,1}\}_{n\in \N_*} \) if and only if \( \infty^{(1)} \) is a topological limit point of \( \{\z_{n,0}\}_{n\in \N_*} \). Moreover, it holds that \( [s_{n,0}(t;N_n)]  = [s_{n-1,1}(t;N_n)]\); that is, \( z_{n,0}(N_n) = z_{n-1,1}(N_n) \).
\end{proposition}

The behavior of the points $\z_{n,k}$ with respect to $n$ can be extremely chaotic. Assuming that \( n-N_n \) is constant, it is known that if the numbers $\omega$ and $\tau$ are rational, then there exist only finitely many distinct points $\z_{n,k}$; when $\omega$ and $\tau$ are irrational, all the points $\z_{n,k}$ are distinct, lie on a Jordan curve if $1$, $\omega$ and $\tau$ are rationally dependent, and are dense on the whole surface $\RS$ otherwise \cite{Suet03}. 

Let us also point out that the point \( p \) is always finite. Indeed, it follows from \eqref{cm2} and \eqref{em5} that \( b_1+b_2+a_1+a_2 = 0 \), see also \eqref{jac2} and \eqref{jac7}. Thus, if \( b_2 + b_1 - a_1 - a_2 = 0 \) were true, then we would have \( b_1=-b_2 \) and \( a_1=-a_2 \). In this case \( a_1b_1a_2 + a_1b_1b_2 + a_1a_2b_2 + b_1a_2b_2 = 0 \). However, this expression must be equal to \( -4 \) according to the above references.

\begin{proof}[Proof of Proposition~\ref{prop:Nt}]
Define 
\begin{equation}
\label{nt1}
\gamma(z):=\left(\frac{z-b_2}{z-a_2}\frac{z-b_1}{z-a_1}\right)^{1/4}, \quad z\in\overline\C\setminus J_t,
\end{equation}
where $\gamma(z)$ is holomorphic off $J_t=\supp(\mu_t)$ and the branch is chosen so that $\gamma(\infty)=1$. Further, set
\begin{equation}
\label{nt2}
A(z) = \frac{\gamma(z)+\gamma^{-1}(z)}2 \qandq B(z) := \frac{\gamma(z)-\gamma^{-1}(z)}{-2\ic}.
\end{equation}
Observe that the function \( A(z) \) was already defined in \eqref{alpha}. The functions \( A(z) \) and \( B(z) \) are holomorphic in \( \overline\C\setminus J_t \) and satisfy
\begin{equation}
\label{nt3}
A(\infty)=1, \quad B(\infty)=0, \qandq A_\pm(s) = \pm B_\mp(s), ~~ s\in J_t^\circ.
\end{equation}
Notice that the equation \( (AB)(z)=0 \) can be rewritten as \( \gamma^4(z) = 1 \) and has two solutions, namely, \( \infty \) and the point \( p \) from \eqref{tc4}. In fact, unless \( p\in J_t^\circ \), it a zero of \( B(z) \). Indeed, it is enough to show that \( \gamma(p)=1 \) in the latter case. Let \( L_i :=\gamma^4(J_{t,i})\), \( i\in\{1,2\} \), which are unbounded arcs connecting the origin to the point at infinity. Let \( L\subset\overline\C\setminus J_t  \) be an arc connecting  the point at infinity and \( p \). Then \( \gamma^4(L) \) is a closed curve that contains \( 1\) and does not intersect the arcs \( L_i \) and therefore does not wind around the origin. Thus, analytic continuation of the principal branch of the \( 1/4 \)-root from \( 1 \) along \( \gamma^4(L) \) leads back to the value \( 1 \) at the point \( 1 \). However, this continuation is exactly the continuation of \( \gamma(z) \) from the point at infinity to \( p \) along \( L \), which does imply that \( \gamma(p)=1 \) as claimed.

It follows from \eqref{nt3} that
\begin{equation}
\label{nt4}
\left\{
\begin{array}{rl}
(B/A)(z), & \z\in D^{(0)}, \medskip \\
-(A/B)(z), & \z\in D^{(1)},
\end{array}
\right.
\end{equation}
is a rational function on \( \RS \) with two simple zeros $\infty^{(0)}$ and $p^{(0)}$ and two simple poles $\infty^{(1)}$ and $p^{(1)}$ (if it happens that \( p\in J_t^\circ \), then we choose \( p^{(0)}\in\RS \) precisely in such a way that it is a zero of \eqref{nt4} and so \( p^{(1)} \) is a pole of \eqref{nt4}; it is, of course, still true that these points are distinct and \( \pi\big(p^{(k)}\big) = p \)). Therefore, Abel's theorem yields that
\begin{equation}
\label{nt5}
\left[\int_{p^{(0)}}^{\infty^{(1)}}\mathcal H\right]  = \left[\int_{p^{(1)}}^{\infty^{(0)}}\mathcal H\right]
\end{equation}
while the relations \eqref{tc4}, in particular, imply that
\begin{equation}
\label{nt6}
\left[\int_{p^{(0)}}^{\z_{n,0}}\mathcal H\right]  = \left[\int_{p^{(1)}}^{\z_{n,1}}\mathcal H\right].
\end{equation}
Let $\z_k$ be a topological limit of a subsequence $\{\z_{n_i,k}\}$. Holomorphy of the differential \( \mathcal H \) implies that
\[
\int_{p^{(k)}}^{\z_{n_i,k}} \mathcal H = \int_{p^{(k)}}^{\z_k} \mathcal H + \int_{\z_k}^{\z_{n_i,k}} \mathcal H \to \int_{p^{(k)}}^{\z_k} \mathcal H
\]
as \( i\to\infty \), where the integral from \( \z_k \) to \( \z_{n_i,k} \) is taken along the path that projects into a segment joining \( z_k \) and \( z_{n_i,k} \). The first claim of the proposition now follows from \eqref{nt5}, \eqref{nt6}, and the unique solvability of the Jacobi inversion problem on \( \RS \). Observe that
\begin{equation}
\label{nt7}
\left[\int_{p^{(0)}}^{p^{(1)}}\mathcal H\right]  = \left[\int_{\infty^{(1)}}^{\infty^{(0)}}\mathcal H\right] = \big[ \varsigma + \omega + \mathsf B\tau \big].
\end{equation}
by \eqref{nt5} and \eqref{residue}. Since
\[
s_{n,0}(N_n) = s_{n-1,1}(N_n) + \int_{p^{(1)}}^{p^{(0)}} \mathcal H + \varsigma + \omega + \mathsf B\tau
\]
the second conclusion of the proposition easily follows.
\end{proof}

\subsection{Proof of Proposition~\ref{prop:N0}}
\label{ss:N0}

As we shall show further below, the functions \( \Theta_n(z;t) \) from Proposition~\ref{prop:A} vanish at \( \z_{n,1} \) when it belongs to \( D^{(0)} \) and do not vanish at all when  \( \z_{n,1} \) does not belong to \( D^{(0)} \). Hence, the subsequences \( \N(\varepsilon)=\N(t,\varepsilon) \) from Proposition~\ref{prop:A} can be equivalently defined as 
\[
\N(\varepsilon) := \left\{ n\in\N:~~\z_{n,1} \not\in D^{(0)} \cap \pi^{-1}\big(\big\{|z|\geq 1/\varepsilon\big\}\big) \right\},
\]

Denote by \( r[s] \) the distance from \( 0 \) to \( [s] \), viewed as a lattice in \( \C \). That is,
\[
r[s]:= \min_{j,m\in\Z}\big|s+j+\mathsf Bm\big|,
\]
where $s$ is any representative of $[s]$. Further, for each \( r[s]>0 \), choose \( \varepsilon_s>0 \) small enough so that
\[
\big| \mathfrak a(\z)-\mathfrak a(\infty^{(0)}) \big|<r[s]/3, \quad \z\in U_{\varepsilon_s}, \quad U_\varepsilon := D^{(0)} \cap \pi^{-1}\big(\big\{|z| > 1/\varepsilon\big\}\big),
\]
where we also assume that \( \varepsilon_s \) is small enough so that always \( U_{\varepsilon_s} \subset\RS_{\boldsymbol\alpha,\boldsymbol\beta} \). Then
\begin{equation}
\label{7.3.0}
\big[\mathfrak a(\z_1) - \mathfrak a(\z_2) \big]  = [s] \neq [0] \quad \Rightarrow \quad \text{either } \z_1\notin U_{\varepsilon_s} \text{ or } \z_2 \notin U_{\varepsilon_s}
\end{equation} 
as otherwise it must hold \( 0<r[s] \leq |\mathfrak a(\z_1) - \mathfrak a(\z_2)|< (2/3) r[s] \).

Assume that either \( x \) or \( y \) is not an integer. Let \( \z_{n,1}^-,\z_{n,1},\z_{n,1}^+ \) be solutions of Jacobi inversion problem \eqref{tc4} with \( N \) equal to \( n-1,n,n+1 \), respectively. Then it follows from \eqref{tc4} that
\[
\big[\mathfrak a(\z_{n,1}^\pm) - \mathfrak a(\z_{n,1}) \big]  = \big[\mp \varsigma\big]. 
\]
As \( [\pm\varsigma] \neq [0] \) by assumption, the second assertion of the proposition follows from \eqref{7.3.0}.

Assume now that both \( x \) and \( y \) are integers. It readily follows from \eqref{tc4} that the definition of \( \z_{n,k} \) does not depend on the choice of \( N_n \) in this case. The first assertion of the proposition holds since
\[
\big[\mathfrak a(\z_{n+1,1})-\mathfrak a(\z_{n,1})\big]  = \big[\omega + \mathsf B\tau\big] \neq \big[0\big],
\]
where the last conclusion follows from the fact \( \omega=\mu_t(J_{t,1})\in(0,1) \) by the very definition in \eqref{tau-omega}.

\subsection{Proof of Proposition~\ref{prop:N}}
\label{ss:N}

Let \( k= N_{n+1} - N_n  \). It follows from \eqref{tc4} that
\begin{equation}
\label{7.3.1}
\big[\mathfrak a(\z_{n+1,1})-\mathfrak a(\z_{n,1})\big]  = \big[(1-k)\varsigma + \omega + \mathsf B\tau\big].
\end{equation}
Set \( \mathbb K  := \{k\in\Z:r[(1-k)\varsigma + \omega + \mathsf B\tau]=0\} \). We have already shown that \( r[\omega + \mathsf B\tau]>0 \). It also holds that \( r[\varsigma+\omega + \mathsf B\tau]>0 \) since  \( [\varsigma + \omega + \mathsf B\tau] \neq [0] \) by \eqref{residue} and the unique solvability of the Jacobi inversion problem. The following lemma easily follows from \eqref{7.3.0}.

\begin{lemma}
Assume that \( |n-N_n|\leq N_* \) and let \( \varepsilon_k:=\varepsilon_{(1-k)\varsigma + \omega + \mathsf B\tau} \).
\begin{itemize}
\item[(1)] If \( \mathbb K=\varnothing \), then at least one of the integers \( n,n+1 \) belongs to \( \N(\varepsilon) \) for all \( \varepsilon\leq \min_{|k|\leq 2N_*+1}\varepsilon_k \).
\item[(2)] If \( \mathbb K=\{k^\prime\} \), there exists an infinite subsequence \( \{ n_l \} \) such that \( N_{n_l+1}-N_{n_l}\neq k^\prime\) and therefore at least one of the integers \( n_l,n_l+1 \) belongs to \( \N(\varepsilon) \) for all \( \varepsilon\leq \min_{|k|\leq 2N_*+1,k\neq k^\prime}\varepsilon_k \). 
\item[(3)] If there exists an infinite subsequence \( \{ n_l \} \) such that \( N_{n_l+1}-N_{n_l}\in \{ 0,1 \} \), then at least one of the integers \( n_l,n_l+1 \) belongs to \( \N(\varepsilon) \) for all \( \varepsilon\leq \min \{\varepsilon_0,\varepsilon_1\} \). 
\end{itemize}
\end{lemma}

Inclusion \( k\in\mathbb K \) means that \( [(1-k)\varsigma+\omega+\mathsf B\tau]=[0]\).  Hence, both triples \( \omega,x,1 \) and \( \tau,y,1 \) are rationally dependent, \( \varsigma=x+\mathsf By \). Assume that \( k^\prime,k^{\prime\prime}\in\mathbb K \), \( k^\prime\neq k^{\prime\prime} \). Then it follows from \eqref{7.3.1} that \( [\omega+\mathsf B\tau] = [(k^\prime-1)\varsigma] \) and \( [(k^{\prime\prime}-k^\prime)\varsigma] = [0] \). The latter relation implies the first representation in \eqref{degen-cond} while the former gives the other two. It is easy to see in this case that \( \mathbb K= k^\prime + d\Z \). That is, if \( \mathbb K \) has at least two elements, then it is an arithmetic progression, \( \omega,\tau \) are rational numbers, \( \varsigma \) has rational coordinates in the basis \( 1,\mathsf B \), and the second and third relations of \eqref{degen-cond} must be satisfied. Thus, we can claim the following.

\begin{lemma}
If one of the triples \( \omega,x,1 \) or \( \tau,y,1 \) is rationally independent, then \( \mathbb K = \varnothing \). If not all numbers \( \omega,\tau,x,y \) are rational or they all rational but the second and third relations of \eqref{degen-cond} do not hold, then either \( \mathbb K=\varnothing \) or \( \mathbb K=\{k^\prime\}  \).
 \end{lemma}

Assume now that all three relations of \eqref{degen-cond} take place. That is, \( [\omega+\mathsf B\tau] = [(k-1)\varsigma] \) and \( [d\varsigma] = [0] \) for some integers \( k,d \). It follows from \eqref{nt5} that
\[
\int_{\boldsymbol b_2}^{p^{(1)}}\mathcal H = \frac12\int_{p^{(0)}}^{p^{(1)}}\mathcal H = \frac12\left(\int_{\infty^{(1)}}^{\infty^{(0)}}\mathcal H + j+\mathsf Bm\right) = \int_{\boldsymbol b_2}^{\infty^{(0)}}\mathcal H + \frac{j+\mathsf Bm}2
\]
for some \(j,m\in\Z \), where we use involution-symmetric paths of integration. Notice that \( j,m \) cannot be simultaneously even as this would contradict unique solvability of the Jacobi inversion problem. In fact, it holds that
\begin{equation}
\label{7.3.3}
\left[\int_{p^{(1)}}^{\infty^{(0)}}\mathcal H\right] = \left[ \int_{\boldsymbol b_2}^{ \infty^{(0)}} \mathcal H + \int_{\boldsymbol b_2}^{ p^{(0)}} \mathcal H \right] = \left[ \int_{\boldsymbol b_2}^{\boldsymbol b_1} \mathcal{H} \right] = \left[\frac{ 1 +\mathsf B}2\right],
\end{equation}
i.e., both $j, m$ must be odd. Indeed, the first equality follows from the symmetry $\mathfrak{a}(\boldsymbol z^*) = -\mathfrak a (\boldsymbol z)$. Next, observe that the function $(-1)^k\gamma^{-2}\big(z^{(k)}\big) - 1$ is meromorphic on $\RS$ with zeros at $\infty^{(0)}$ and $p^{(0)}$ and poles at $\boldsymbol b_1$ and $\boldsymbol b_2$, all simple. Hence, the second equality is a consequence of the same symmetry and Abel's theorem. Further, consider an involution-symmetric cycle represented by the anti-diagonal on Figure \ref{planar-model}, oriented so that it proceeds towards $\boldsymbol b_1$ in $D^{(0)}$. It is clearly homologous to $\boldsymbol \alpha + \boldsymbol \beta$. Then the symmetry of $\mathfrak a (\boldsymbol z)$ and the normalization in \eqref{hol-diff} and \eqref{B} yield the final equality in \eqref{7.3.3}. Notice that the conclusion of \eqref{7.3.3} a posteriori holds with the bounds of integration in left most integral of \eqref{7.3.3} in flipped. Therefore, adding \( \int_{\infty^{(0)}}^{\boldsymbol b_2}\mathcal H \) to both sides of \eqref{tc4} gives us
\begin{equation}
\label{7.3.2}
\left[\int_{\infty^{(0)}}^{\z_{n,1}}\mathcal H\right] = \left[\frac{{1}+\mathsf B}2 + (n-N_n)\varsigma + (\omega+\mathsf B\tau) n\right] = \left[\frac{{1}+\mathsf B}2 + (nk-N_n)\varsigma\right].
\end{equation}
Since \( \varsigma \) has rational coordinates in the basis \( 1,\mathsf B \) with denominator \( d \), the right-hand side of \eqref{7.3.2} has at most \( d \) distinct values that depend only on \( \varrho\in\{0,\ldots,d-1\} \), the remainder of the division of \( nk - N_n  \) by \( d \). Let \( \z_\varrho \), \( \varrho\in\{0,\ldots,d-1\} \), be such that
\[
\left[\int_{\infty^{(0)}}^{\z_\varrho}\mathcal H\right] = \left[\frac{{1}+\mathsf B}2 + \varrho\varsigma\right].
\]
Clearly,  \( \{\z_{n,1}\}_{n\in\N}\subseteq\{ \z_\varrho \}_{\varrho=0}^{d-1}\). Thus, it only remains to investigate when \( \z_\varrho=\infty^{(0)} \), or equivalently, when \( [({1}+\mathsf B)/2+\varrho\varsigma]=[0] \). From the representation of \( \varsigma \) in \eqref{degen-cond}, it must simultaneously hold that \( \varrho=d(2l_j-1)/(2i_j) \), \( j\in\{1,2\} \), for some \( l_1,l_2\in\Z \). Since one of the pairs \( (i_j,d) \) is co-prime and \( \varrho\in\{0,\ldots,d-1\} \), this is possible only if \( \varrho=0 \) or \( \varrho=d/2 \) (in the second case, of course, \( d \) must be even and \( i_j=2l_j-1\) must be odd), and the former is ruled out since $[(1 + \mathsf B)/2] \neq [0]$ (or equivalently, since \( 0 \neq 2l_j-1 \)).

\begin{lemma}
If all three relations of \eqref{degen-cond} take place, then Jacobi inversion problem \eqref{tc4} for \( \z_{n,1} \) has only finitely many distinct solutions and \( \infty^{(0)} \) is one of them if and only if \( d \) is even, \( i_1,i_2 \) are odd, and \( nk-N_n \) is divisible by \( d/2 \) but not by \( d \).
\end{lemma}

\subsection{Proof of Proposition~\ref{prop:A}}
\label{ss:A}

Recall that Abel's map \eqref{Abel-map} is essentially carried out by the function \( \mathfrak a(\z) \) defined in \eqref{ab1}. We shall consider the extension \( \tilde{\mathfrak a}(\z) \) of \( \mathfrak a(\z) \) to the whole surface \( \RS \) defined by setting \( \tilde{\mathfrak a}(\s) := \mathfrak a_+(\s) \) for \( \s\in \boldsymbol \alpha \) and \( \s\in\boldsymbol\beta\setminus\{\boldsymbol b_1\} \). Given such an extension and \eqref{tc4}, there exist unique integers $j_{n,k},m_{n,k}$ such that
\begin{equation}
\label{ab3}
\tilde{\mathfrak a}(\z_{n,k}) = \tilde{\mathfrak a}\big(p^{(k)}\big) + (n-N_n)\varsigma + (\omega + \mathsf B\tau)n + j_{n,k} + \mathsf Bm_{n,k}, \quad k\in\{0,1\}.
\end{equation}

Recall the definition of  $\theta(u)$ in \eqref{Rtheta}. The function $\theta(u)$ is holomorphic in $\C$ and enjoys the following periodicity properties:
\begin{equation}
\label{ab4}
\theta(u+j+\mathsf Bm) = \exp\left\{-\pi\ic \mathsf Bm^2-2\pi\ic um\right\}\theta(u), \qquad j,m\in\Z.
\end{equation}
It is also known that $\theta(u)$ vanishes only at the points of the lattice \( \left[\frac{\mathsf B+1}2\right] \). Let
\begin{equation}
\label{ab5}
\Theta_{n,k}(\z) = \exp\left\{-2\pi\ic\big(m_{n,k}+\tau n\big)\mathfrak a(\z)\right\} \frac{\theta\left(\mathfrak a(\z) - \tilde{\mathfrak a}(\z_{n,k}) - \frac{\mathsf B+1}2\right)}{\theta\left(\mathfrak a(\z) - \tilde{\mathfrak a}\big(p^{(k)}\big) - \frac{\mathsf B+1}2\right)}.
\end{equation}
Then we define \( \Theta_n(z;t) \) from Proposition~\ref{prop:A} as well as the other functions appearing in the statement of Theorem~\ref{thm:beta} by
\begin{equation}
\label{all-thetas}
\Theta_n(z;t) := \Theta_{n,1}^{(0)}(z), \;\; \vartheta_n(z) := \frac{\Theta_{n,1}^{(0)}(z)}{\Theta_{n,1}^{(0)}(\infty)}, \;\; \vartheta_n^*(z) := \frac{\Theta_{n,1}^{(1)}(z)}{\Theta_{n,1}^{(0)}(\infty)}, \;\; \text{and} \;\; \widetilde\vartheta_n(z) := \frac{\Theta_{n,0}^{(0)}(z)}{\Theta_{n,0}^{(1)}(\infty)},
\end{equation}
where \( F^{(i)}(z) \), \( i\in\{0,1\} \), stands for the pull-back under \( \pi(\z) \) of a function \( F(\z) \) from \( D^{(i)} \) into \( \overline\C\setminus J_t \). 

The functions $\Theta_{n,k}(\z)$ are meromorphic on \( \RS_{\boldsymbol\alpha,\boldsymbol\beta} \) with exactly one pole, which is simple and located at $p^{(k)}$, and exactly one zero, which is also simple and located at \( \z_{n,k} \) (observe that these functions can be analytically continued as multiplicatively multivalued functions on the whole surface \( \RS \); thus, we can talk about simplicity of a pole or zero regardless whether it belongs to the cycles of a homology basis or not). Moreover, according to \eqref{ab2}, \eqref{ab3}, and \eqref{ab4}, they possess continuous traces on $\boldsymbol\alpha,\boldsymbol\beta$ away from \( \boldsymbol b_1 \) that satisfy
\begin{equation}
\label{ab6}
\Theta_{n,k+}(\s) = \Theta_{n,k-}(\s)\left\{
\begin{array}{rl}
\exp\big\{-2\pi\ic(\omega n+(n-N_n)\varsigma)\big\}, & \boldsymbol s\in\boldsymbol\alpha\setminus\{\boldsymbol b_1\}, \medskip \\
\exp\big\{-2\pi\ic\tau n\big\}, & \boldsymbol s\in\boldsymbol\beta\setminus\{\boldsymbol b_1\}.
\end{array}
\right.
\end{equation}

Recall functions \( A(z),B(z) \) from \eqref{nt2} and \eqref{nt3}. To discuss boundedness properties of \( \Theta_{n,k}(\z) \) and for the asymptotic analysis in the following section it will be convenient to define
\begin{equation}
\label{ab7}
M_{n,0}(\z) = \Theta_{n,0}(\z) \left\{ \begin{array}{r} B(z), ~~ \z\in D^{(0)}, \medskip \\ A(z), ~~ \z\in D^{(1)}, \end{array}\right. \qandq  M_{n,1}(\z) = \Theta_{n,1}(\z) \left\{ \begin{array}{r} A(z), ~~ \z\in D^{(0)}, \medskip \\ -B(z), ~~ \z\in D^{(1)}.\end{array}\right.
\end{equation}
These functions are holomorphic on \( \RS\setminus\{\boldsymbol\alpha \cup \boldsymbol\beta \cup \boldsymbol\Delta \} \) since the pole of \( \Theta_{n,k}(\z) \) is canceled by the zero of \( B(z) \). Each function \( M_{n,k}(\z) \) has exactly two zeros, namely, \( \z_{n,k} \) and \( \infty^{(k)} \). It follows from \eqref{nt3} and \eqref{ab6} that
\begin{equation}
\label{Mn-jumps}
 \left\{
\begin{array}{rl}
M_{n,k\pm}^{(0)}(s)=\mp M_{n,k\mp}^{(1)}(s), & s\in J_{t,2}^\circ, \medskip \\
M_{n,k\pm}^{(0)}(s)=\mp e^{-2\pi\ic\tau n}M_{n,k\mp}^{(1)}(s), & s\in J_{t,1}^\circ, \medskip \\
M_{n,k\pm}^{(i)}(s) = e^{(-1)^i2\pi\ic(n\omega+(n-N_n)\varsigma)}M_{n,k\mp}^{(i)}(s), & s\in I_t^\circ.
\end{array}
\right.
\end{equation}
It further follows from \eqref{nt1} and \eqref{nt2} that \( |M_{n,k}(\z)|\sim|z-e|^{-1/4} \) as \( \z\to \boldsymbol e\in \boldsymbol E \) unless \( \z_{n,k} \) coincides with \( \boldsymbol e \) in which case the exponent becomes \( 1/4 \). Assume now that there exists \( N_*\geq0 \) such that \( |n-N_n|\leq N_* \) for all \( n\in \N \). Then for each \( \delta>0 \) there exists \( C(\delta,N_*) \) independent of \(  n \) such that
\begin{equation}
\label{MN-up-bound}
|M_{n,k}(\z)| \leq C(\delta,N_*), \quad \z\in O_\delta:=\RS\setminus \cup_{\boldsymbol e\in \boldsymbol E}\pi^{-1}\{|z-e|<\delta\}.
\end{equation}
Indeed, observe that \( \{\log|M_{n,k}(\z)|\} \) is a family of subharmonic functions in \( O_\delta\setminus\boldsymbol\alpha \) (the jump of $M_{n,k}(\z)$ is unimodular on $\boldsymbol\beta$). By the maximum principle for subharmonic functions, \( \log|M_{n,k}(\z)| \) reaches its maximum on $\partial(O_\delta\setminus\boldsymbol\alpha)$, where the maximum is clearly finite. Since the sequence \( \{ n-N_n \} \) is bounded by assumption and the range of $\tilde{\mathfrak a}(\z)$ is bounded by construction, so are the sequences \( \{m_{n,k} + \tau n\} \) and \( \{j_{n,k}+\omega n\} \), see \eqref{ab3} and recall that  \( j_{n,k}+\omega n \) are real and \( \mathrm{Im}(\mathsf B)>0 \). Thus, any limit point of  \( \{\log|M_{n,k}(\z)|\} \) is obtained by taking simultaneous limit points of \( \{ n-N_n \} \), \( \{m_{n,k} + \tau n\} \), and \( \{j_{n,k}+\omega n\} \), computing the corresponding solution \( \z_k \) of the Jacobi inversion problem \eqref{ab3} and plugging all of these quantities into the right-hand side of \eqref{ab5}. Hence, all these limit functions are also bounded above on the closure of $O_\delta\setminus\boldsymbol\alpha$, which proves \eqref{MN-up-bound}. Finally, it holds that
\begin{equation}
\label{MN-low-bound}
\big|M_{n,k}\big(\infty^{(1-k)}\big)\big| \geq c_\varepsilon, \quad n\in \N(\varepsilon),
\end{equation}
for some constant \( c_\varepsilon>0 \) by a similar compactness argument combined with the definition of \( \N(\varepsilon) \) in Proposition~\ref{prop:A}, the observation that \( M_{n,k}(\z) \) is non-zero at \( \infty^{(1-k)} \) when \(  \infty^{(1-k)} \neq \z_{n,k} \), and the last conclusion of Proposition~\ref{prop:Nt}.

\subsection{Proof of Proposition~\ref{prop:l}}
\label{ss:propl}

Define a function \( F(\z) \) on \( \RS\setminus\{\boldsymbol\alpha\cup\boldsymbol\beta\} \) by setting \( F\big(z^{(1)}\big) := 1/F\big(z^{(0)}\big) \)  and 
\[
F\big(z^{(0)}\big) := \mathcal D^{-1}(z) \exp\big\{ V(z)/2 + \mathcal Q(z)\big\},
\]
see \eqref{D} and \eqref{mQ}. It follows from \eqref{Djumps} and \eqref{gjumps} that
\begin{equation}
\label{new1}
F_+(\s) = F_-(\s)\left\{
\begin{array}{rl}
\exp\big\{2\pi\ic(\omega+\varsigma)\big\}, & \boldsymbol s\in\boldsymbol\alpha\setminus\{\boldsymbol b_1\}, \medskip \\
\exp\big\{2\pi\ic\tau\big\}, & \boldsymbol s\in\boldsymbol\beta\setminus\{\boldsymbol b_1\},
\end{array}
\right.
\end{equation}
where one needs to recall that \( I_t \) is oriented towards \( a_2 \) while \( \boldsymbol\alpha\cap \RS^{(0)} \) is oriented towards \( \boldsymbol b_1 \). Moreover, it follows from \eqref{oc5} that
\begin{equation}
\label{new2}
F\big(z^{(0)}\big) = e^{\ell_*/2}\mathcal D^{-1}(\infty) z + \mathcal O(1)
\end{equation}
as \( z\to\infty \). On the other hand, we get from \eqref{nt7} that there exist integers  \( j,m \) such that
\begin{equation}
\label{new5}
\tilde{\mathfrak a}\big(p^{(1)}\big) - \tilde{\mathfrak a}\big(p^{(0)}\big) = \varsigma + \omega + \mathsf B\tau + j + \mathsf Bm.
\end{equation}
Thus, we get similarly to \eqref{ab5} and \eqref{ab6} that the function
\begin{equation}
\label{new3}
\Theta(\z) := \exp\big\{-2\pi\ic(m+\tau)\mathfrak a(\z)\big\} \frac{\theta\left(\mathfrak a(\z) - \tilde{\mathfrak a}\big(p^{(1)}\big) - \frac{\mathsf B+1}2\right)}{\theta\left(\mathfrak a(\z) - \tilde{\mathfrak a}\big(p^{(0)}\big) - \frac{\mathsf B+1}2\right)}
\end{equation}
is meromorphic in \( \RS\setminus\{\boldsymbol\alpha\cup\boldsymbol\beta\} \), with a simple zero at \( p^{(1)} \), a simple pole at \( p^{(0)} \), and otherwise non-vanishing and finite, whose traces on the cycles of the homology basis satisfy
\begin{equation}
\label{new4}
\Theta_+(\s) = \Theta_-(\s)\left\{
\begin{array}{rl}
\exp\big\{-2\pi\ic(\omega+\varsigma)\big\}, & \boldsymbol s\in\boldsymbol\alpha\setminus\{\boldsymbol b_1\}, \medskip \\
\exp\big\{-2\pi\ic\tau\big\}, & \boldsymbol s\in\boldsymbol\beta\setminus\{\boldsymbol b_1\}.
\end{array}
\right.
\end{equation}
Combining \eqref{new1}, \eqref{new2}, and \eqref{new4} yields that the function \( (F\Theta)(\z) \) is rational on \( \RS \) with the divisor  \( p^{(1)} + \infty^{(1)} - p^{(0)} - \infty^{(0)} \). That is, this function is a constant multiple of the reciprocal of \eqref{nt4}. Observe that
\[
A(\infty) = 1 \quad \text{and} \quad B(z) = \frac{S_1}{4\ic} \frac1z + \mathcal O\left( z^{-2} \right)
\]
as \( z\to\infty \). Then we get from \eqref{nt4}, \eqref{new2}, and the symmetry \( F(\z)F(\z^*) \equiv 1 \) that
\begin{equation}
\label{new6}
e^{\ell_*} = \mathcal D^2(\infty) \frac{16}{S_1^2} \frac{\Theta\big(\infty^{(1)}\big)}{\Theta\big(\infty^{(0)}\big)}.
\end{equation}
Thus, it only remains to obtain the expression for the last fraction. We get from \eqref{7.3.3} that
\[
\mathfrak a\big( \infty^{(0)} \big) - \tilde{\mathfrak a}\big(p^{(1)}\big) = \frac{1+\mathsf B}2 + u + \mathsf Bv
\]
for some integers \( u,v \). Hence, we get from \eqref{residue}, \eqref{new5}, and \eqref{ab4} that
\[
\Theta\big(\infty^{(0)}\big)  =  \frac{e^{-\pi\ic(m+\tau)(\varsigma+\omega+\mathsf B\tau)}\theta(\mathsf Bv)}{\theta(\varsigma+\omega+\mathsf B\tau+ \mathsf B(m+v))} = \frac{e^{\pi\ic(2v+m-\tau)(\varsigma+\omega+\mathsf B\tau)+\pi\ic\mathsf B(2v+m)}\theta(0)}{\theta(\varsigma+\omega+\mathsf B\tau)}
\]
and similarly
\[
\frac1{\Theta\big(\infty^{(1)}\big)}  =  \frac{e^{-\pi\ic(m+\tau)(\varsigma+\omega+\mathsf B\tau)}\theta(\mathsf B(m+v))}{\theta(\mathsf Bv-\varsigma - \omega - \mathsf B\tau)} = \frac{e^{-\pi\ic(2v+m+\tau)(\varsigma+\omega+\mathsf B\tau)-\pi\ic\mathsf B(2v+m)}\theta(0)}{\theta(-\varsigma-\omega-\mathsf B\tau)}.
\]
Since \( \theta(\cdot) \) is an even function, this finishes the proof of the proposition.

\section{Asymptotic Analysis}
\label{s:aa}

\subsection{Initial Riemann-Hilbert Problem}

As agreed before, we omit the dependence on \( t \) whenever it does not cause ambiguity. In what follows, it will be convenient to set
\[
\boldsymbol I:=\left(\begin{matrix} 1 & 0 \\ 0 & 1 \end{matrix}\right), \quad \sigma_1:=\left(\begin{matrix} 0 & 1 \\ 1 & 0 \end{matrix}\right), \quad \text{and} \quad \sigma_3:=\left(\begin{matrix} 1 & 0 \\ 0 & -1 \end{matrix}\right).
\]
We are seeking solutions of the following sequence of Riemann-Hilbert problems for $2\times2$ matrix functions (\rhy):
\begin{itemize}
\label{rhy}
\item[(a)] $\boldsymbol Y(z)$ is analytic in $\C\setminus\Ga$ and $\lim_{\C\setminus\Ga\ni z\to\infty}\boldsymbol Y(z)z^{-n\sigma_3}=\boldsymbol I$;
\item[(b)] $\boldsymbol Y(z)$ has continuous traces on $\Gamma\setminus\{a_1,b_1,a_2,b_2\}$ that satisfy
\[
\boldsymbol Y_+(s) = \boldsymbol Y_-(s) \left(\begin{matrix}1& e^{-N_nV(s)}\\0&1\end{matrix}\right),
\]
where, as before, \( V(z) \) is given by \eqref{cm2} and the sequence \( \{N_n\} \) is such that \( |n-N_n|\leq N_* \) for some \( N_*>0 \).
\end{itemize}

The connection of \hyperref[rhy]{\rhy} to orthogonal polynomials was first demonstrated by Fokas, Its, and Kitaev in \cite{FIK} and lies in the following. If the solution of \hyperref[rhy]{\rhy} exists, then it is necessarily of the form
\begin{equation}
\label{rh1}
\boldsymbol Y(z) = \left(\begin{matrix}
P_n(z) & \big(\mathcal{C}P_ne^{-N_nV}\big)(z) \medskip \\
-\frac{2\pi\ic}{h_{n-1}}P_{n-1}(z) & -\frac{2\pi\ic}{h_{n-1}}\big(\mathcal{C}P_{n-1}e^{-N_nV}\big)(z)
\end{matrix}\right),
\end{equation}
where $P_n(z)=P_n(z;t,N_n)$ are the polynomial satisfying orthogonality relations \eqref{cm3}, $h_n=h_n(t,N_n)$ are the constants defined in \eqref{cm5a}, and $\mathcal{C}f(z)$ is the Cauchy transform of a function $f$ given on $\Gamma$, i.e.,
\[
(\mathcal{C}f)(z) := \frac1{2\pi\ic}\int_\Gamma\frac{f(s)}{s-z}\dd s.
\]

Below, we show the solvability of \hyperref[rhy]{\rhy} for all $n\in\N(t,\varepsilon)$ large enough following the framework of the steepest descent analysis introduced by Deift and Zhou \cite{DZ}. The latter lies in a series of transformations which reduce \hyperref[rhy]{\rhy} to a problem with jumps asymptotically close to identity.

\subsection{Renormalized Riemann-Hilbert Problem}

Suppose that $\boldsymbol Y(z)$ is a solution of \hyperref[rhy]{\rhy}. Put
\begin{equation}
\label{rh2}
\boldsymbol T(z) := e^{n\ell_*\sigma_3/2}\boldsymbol Y(z) e^{-n(g(z)+\ell_*/2)\sigma_3},
\end{equation}
where the function $g(z)$ is defined by \eqref{oc4} and \( \ell_* \) was introduced in \eqref{g1}. Then
\[
\boldsymbol T_+(s) = \boldsymbol T_-(s)\left(\begin{matrix} e^{-n(g_+(s)-g_-(s))} & e^{n(g_+(s)+g_-(s)-V(s)+\ell_*)+(n-N_n)V(s)} \\ 0 & e^{-n(g_-(s)-g_+(s))} \end{matrix}\right),
\]
$s\in\Gamma$, and therefore we deduce from \eqref{oc5} and \eqref{g4}--\eqref{g6} that $\boldsymbol T(z)$ solves \rht:
\begin{itemize}
\label{rht}
\item[(a)] $\boldsymbol T(z)$ is analytic in $\C\setminus\Gamma$ and $\lim_{\C\setminus\Gamma\ni z\to\infty}\boldsymbol T(z)=\boldsymbol I$;
\item[(b)] $\boldsymbol T(z)$ has continuous traces on $\Gamma\setminus\{a_1,b_1,a_2,b_2\}$ that satisfy
\[
\boldsymbol T_+(s) = \boldsymbol T_-(s)  \left\{
\begin{array}{rl}
\left(\begin{matrix} 1 & e^{n(2\pi\ic\tau+\phi_{a_1}(s))+(n-N_n)V(s)}\\0&1\end{matrix}\right), & s\in\Gamma(e^{\ic\pi}\infty,a_1), \medskip \\
\left(\begin{matrix} 1 & e^{n\phi_{b_2}(s)+(n-N_n)V(s)} \\ 0 & 1 \end{matrix}\right), & s\in\Gamma(b_2,e^{\pi\ic/3}\infty), \medskip \\
\left(\begin{matrix} e^{2\pi\ic\omega n} & e^{n\phi_{a_2}(s)+(n-N_n)V(s)}\\ 0 & e^{-2\pi\ic\omega n} \end{matrix}\right), & s\in\Gamma(b_1,a_2) ,
\end{array}
\right.
\]
and
\[
\boldsymbol T_+(s) = \boldsymbol T_-(s)  \left\{
\begin{array}{rl}
\left(\begin{matrix} e^{-n\phi_{b_2+}(s)} & e^{(n-N_n)V(s)} \\ 0 & e^{-n\phi_{b_2-}(s)} \end{matrix}\right), & s\in \Gamma(a_2,b_2), \medskip \\
\left(\begin{matrix} e^{-n(\phi_{b_2+}(s)-2\pi\ic\tau)}& e^{2\pi\ic\tau n+(n-N_n)V(s)} \\ 0 & e^{-n(\phi_{b_2-}(s)-2\pi\ic\tau)} \end{matrix}\right), & s\in \Gamma(a_1,b_1).
\end{array}
\right.
\]
\end{itemize}

Clearly, if \hyperref[rht]{\rht} is solvable and $\boldsymbol T(z)$ is the solution, then by inverting \eqref{rh2} one obtains a matrix $\boldsymbol Y(z)$ that solves \hyperref[rhy]{\rhy}.

\subsection{Lens Opening} 

As usual in the steepest descent analysis of matrix Riemann-Hilbert problems for orthogonal polynomials, the next step is based on the identity
\begin{multline*}
\left(\begin{matrix}e^{-n(\phi_{b_2+}(s)-C)}& e^{nC+(n-N_n)V(s)} \\ 0 &e^{-n(\phi_{b_2-}(s)-C)}\end{matrix}\right) =\left(\begin{matrix} 1 & 0 \\ e^{-n\phi_{b_2-}(s)-(n-N_n)V(s)} & 1 \end{matrix}\right) \times \\ \left(\begin{matrix} 0 & e^{nC+(n-N_n)V(s)} \\ -e^{-nC-(n-N_n)V(s)} & 0 \end{matrix}\right)\left(\begin{matrix} 1 & 0 \\ e^{-n\phi_{b_2+}(s)-(n-N_n)V(s)} & 1 \end{matrix}\right)
\end{multline*}
that follows from \eqref{g4}, where  $C=2\pi\ic\tau$ when  $s\in\Gamma(a_1,b_1)$ and $C=0$ when  $s\in\Gamma(a_2,b_2)$. To carry it out, we shall introduce two additional system of arcs.  
\begin{figure}[!h]
\includegraphics[scale=.25]{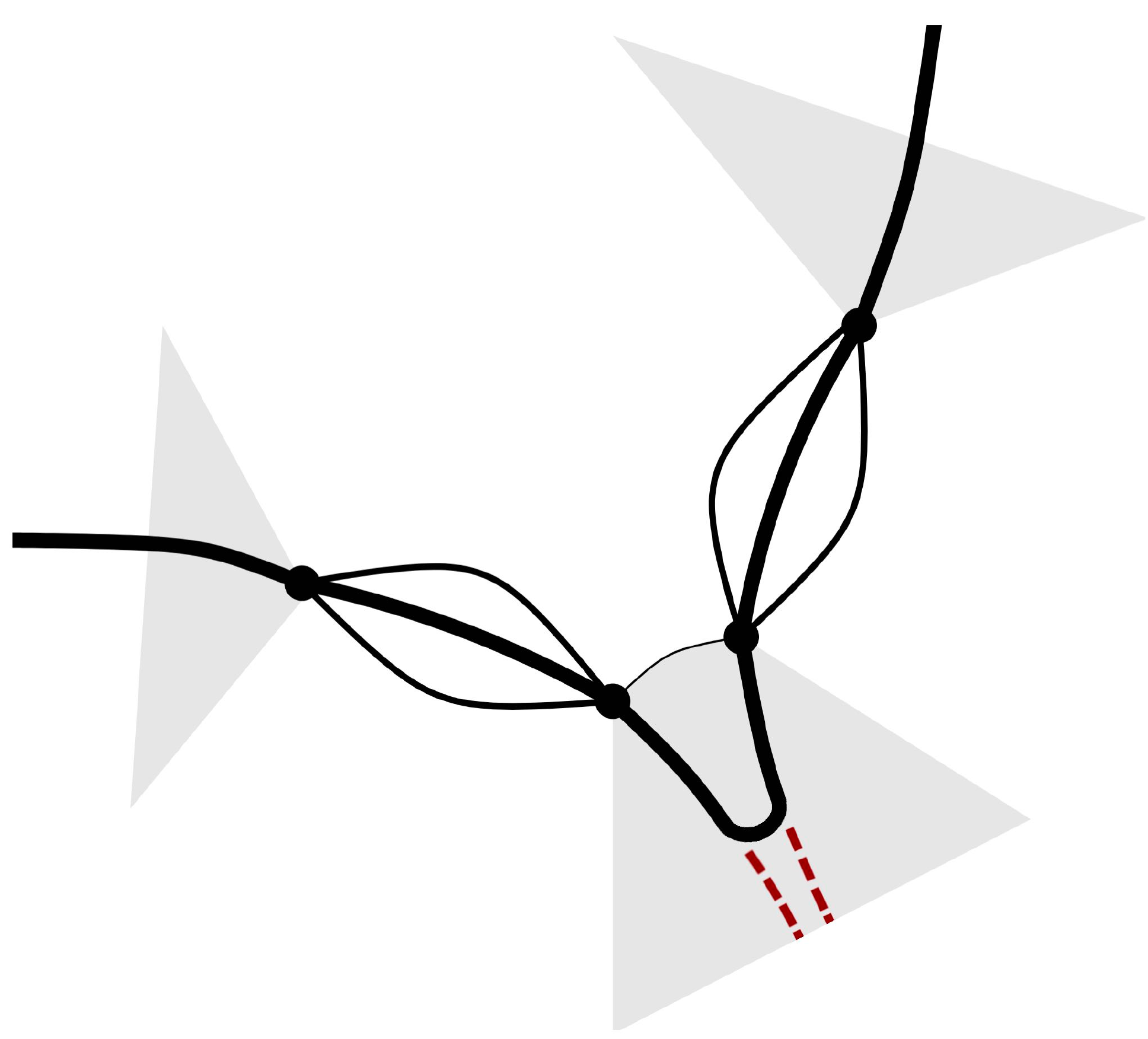}
	\begin{picture}(0,0)
	\put (-85, 60){$J_+$}
	\put (-67, 67){$J_+$}
	\put (-35, 50){$J_-$}
	\put(-92,33){$J_-$}
	\end{picture}
  \caption{\small The thick curves represent $\Gamma$ and thinner black curves represent $J_\pm$. The shaded part represents regions where $\re(\phi_e(z))<0$ (they are the same for each $e\in\{a_1,b_1,a_2,b_2\}$).}
\label{lens}
 \end{figure}
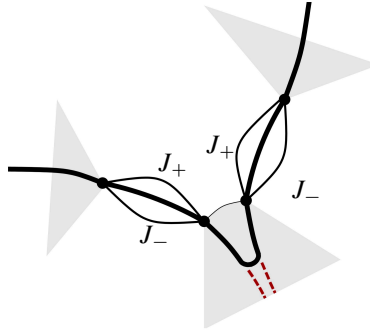
 
 Denote by \( J_\pm\) smooth homotopic deformations of $J_t$ within the region $\re(\phi_{b_2}(z))>0$ such that $J_+$ lies to the left and $J_-$ to the right of $J_t$, see Figure~\ref{lens} (both connected subarcs of \( J_+ \), resp. \( J_- \), are oriented from \( a_i \) to \( b_i \), \( i\in\{1,2\}\)). We shall fix the way these arcs emanate from $e\in\{a_1,b_1,a_2,b_2\}$. Namely, let $U_e$ be given by \eqref{g7} and \( (-\phi_e)^{2/3}(z) \) be as in \eqref{g12}.  Then we require that
\begin{equation}
\label{rh3}
\arg\big((-\phi_e)^{2/3}(z)\big) = \pm\nu_e(2\pi/3), \quad z\in U_e\cap J_\pm,
\end{equation}
where $\nu_e$ is defined by \eqref{g11}. This requirement always can be fulfilled due to conformality \( (-\phi_e)^{2/3}(z) \) in \( U_e \) and the choice of the branch in \eqref{g12}. 

Denote by $O_\pm$ the open sets delimited by $J_\pm$ and $J_t$. Set
\begin{equation}
\label{rh4}
\boldsymbol S(z) := \boldsymbol T(z) \left\{
\begin{array}{rl}
\left(\begin{matrix} 1 & 0 \\ \mp e^{-n\phi_{b_2}(z)-(n-N_n)V(z)} & 1 \end{matrix}\right), & z\in O_\pm, \medskip \\
\boldsymbol I, & \text{otherwise}.
\end{array}
\right.
\end{equation}
Then, if $\boldsymbol T(z)$ solves \hyperref[rht]{\rht}, $\boldsymbol S(z)$ solves \rhs:
\begin{itemize}
\label{rhs}
\item[(a)] $\boldsymbol S(z)$ is analytic in $\C\setminus(\Gamma\cup J_+\cup J_-)$ and $\lim_{\C\setminus\Gamma\ni z\to\infty}\boldsymbol S(z)=\boldsymbol I$;
\item[(b)] $\boldsymbol S(z)$ has continuous traces on $\Gamma\setminus\{a_1,b_1,a_2,b_2\}$ that satisfy \hyperref[rht]{\rht}(b) on $\Gamma(e^{\pi\ic}\infty,a_1)$, $\Gamma(b_1,a_2)$, and $\Gamma(b_2,e^{\pi\ic/3}\infty)$, as well as
\[
\boldsymbol S_+(s) = \boldsymbol S_-(s) \left\{
\begin{array}{rl}
\left(\begin{matrix} 0 & e^{(n-N_n)V(s)} \\ -e^{-(n-N_n)V(s)} & 0 \end{matrix}\right), & s\in \Gamma(a_2,b_2), \medskip \\
\left(\begin{matrix} 0 & e^{2\pi\ic\tau n+(n-N_n)V(s)} \\ -e^{-2\pi\ic\tau n-(n-N_n)V(s)} & 0 \end{matrix}\right), & s\in \Gamma(a_1,b_1), \medskip \\
\left(\begin{matrix}1& 0 \\ e^{-n\phi_{b_2}(s)-(n-N_n)V(s)} &1\end{matrix}\right), & s\in J_\pm.
\end{array}
\right.
\]
\end{itemize}
As before, since transformation \eqref{rh4} is invertible, a solution of \hyperref[rhs]{\rhs} yields a solution of \hyperref[rht]{\rht}.

\subsection{Global Parametrix} 

The Riemann-Hilbert problem for the global parametrix is obtained from \hyperref[rhs]{\rhs} by removing the quantities that are asymptotically zero from the jump matrices in \hyperref[rhs]{\rhs}(b). The latter can be easily identified with the help of \eqref{g3} and by recalling that the constant $\tau$ is real. Thus, we are seeking the solution of \rhn:
\begin{itemize}
\label{rhn}
\item[(a)] $\boldsymbol N(z)$ is analytic in $\overline\C\setminus\Gamma[a_1,b_2]$ and $\boldsymbol N(\infty)=\boldsymbol I$;
\item[(b)] $\boldsymbol N(z)$ has continuous traces on $\Gamma(a_1,b_2)\setminus\{b_1,a_2\}$ that satisfy
\[
\boldsymbol N_+(s) = \boldsymbol N_-(s)  \left\{
\begin{array}{rl} 
\left(\begin{matrix} 0 & e^{(n-N_n)V(s)} \\ -e^{-(n-N_n)V(s)} & 0 \end{matrix}\right), & s\in \Gamma(a_2,b_2), \medskip \\
\left(\begin{matrix} 0 & e^{2\pi\ic\tau n+(n-N_n)V(s)} \\ -e^{-2\pi\ic\tau n-(n-N_n)V(s)} & 0 \end{matrix}\right), & s\in \Gamma(a_1,b_1), \medskip \\
\left(\begin{matrix} e^{2\pi\ic\omega n} & 0 \\ 0 & e^{-2\pi\ic\omega n} \end{matrix}\right), & s\in \Gamma(b_1,a_2).
\end{array}
\right.
\]
\end{itemize}
We shall solve this problem only for \(n\in\N(\varepsilon)=\N(t,\varepsilon) \) from Proposition~\ref{prop:A}. In fact, to solve \hyperref[rhn]{\rhn} we only need to exclude indices \( n \) for which \( \z_{n,1}=\infty^{(0)}\) (in this case \( \z_{n,0}=\infty^{(1)}\), see Proposition~\ref{prop:Nt}) as should be clear from \eqref{rh5}. However, for further analysis we shall need estimate \eqref{MN-low-bound}, where it is crucial that \( n\in\N(\varepsilon) \) for some \( \varepsilon>0 \).

Let the functions \( M_{n,k}(\z) \) be given by \eqref{ab7} and \( \mathcal D(z)=\mathcal D(z;t) \) be defined by \eqref{D}. With the notation introduced right after \eqref{ab5}, a solution of \hyperref[rhn]{\rhn} is given by
\begin{equation}
\label{rh5}
\boldsymbol N(z) = \boldsymbol M^{-1}(\infty)\boldsymbol M(z), \quad \boldsymbol M(z) := \left(\begin{matrix} M_{n,1}^{(0)}(z) & M_{n,1}^{(1)}(z) \medskip \\ M_{n,0}^{(0)}(z) & M_{n,0}^{(1)}(z) \end{matrix}\right) \mathcal D^{(N_n-n)\sigma_3}(z).
\end{equation}
Indeed,  \hyperref[rhn]{\rhn}(a) follows from holomorphy of \( \mathcal D(z) \) and \( M_{n,k}(\z) \) discussed in Proposition~\ref{prop:Szego} and right after \eqref{ab7}. Fulfillment of \hyperref[rhn]{\rhn}(b) can be checked by using \eqref{Djumps} and \eqref{Mn-jumps}. Observe also that \( \det(\boldsymbol N(z))\equiv 1\). Indeed, as the jump matrices in \hyperref[rhn]{\rhn}(b) have unit determinants,  $\det(\boldsymbol N(z))$ is holomorphic through $\Gamma(a_1,b_1)$, $\Gamma(b_1,a_2)$, and $\Gamma(a_2,b_2)$. It also has at most square root singularities at $\{a_1,b_1,a_2,b_2\}$ as explained right after \eqref{Mn-jumps}. Thus, it is holomorphic throughout $\overline\C$ and therefore is a constant. The normalization at infinity implies that this constant is \( 1 \).

\subsection{Local Parametrices} 

The jumps discarded in \hyperref[rhn]{\rhn} are not uniformly close to the identity around each point $e\in\{a_1,b_1,a_2,b_2\}$. The goal of this section is to solve \hyperref[rhs]{\rhs} in the disks \( U_e \), see \eqref{g7}, with a certain matching condition on the boundary of the disks.  More precisely, we are looking for a matrix functions $\boldsymbol P_e(z)$ that solves \rhp:
\begin{itemize}
\label{rhp}
\item[(a)] $\boldsymbol P_e(z)$ has the same analyticity properties as $\boldsymbol S(z)$ restricted to $U_e$, see \hyperref[rhs]{\rhs}(a);
\item[(b)] $\boldsymbol P_e(z)$ satisfies the same jump relations as $\boldsymbol S(z)$ restricted to $U_e$, see \hyperref[rhs]{\rhs}(b);
\item[(c)] $\boldsymbol P_e(z)=\boldsymbol N(z)\big(\boldsymbol I+\boldsymbol{\mathcal{O}}(n^{-1})\big)$ holds uniformly on $\partial U_e$ as $n\to\infty$.
\end{itemize}
Again, we shall solve \hyperref[rhp]{\rhp} only for $n\in\N(\varepsilon)$ with \( \boldsymbol{\mathcal{O}}(\cdot)=\boldsymbol{\mathcal{O}}_\varepsilon(\cdot)\) in \hyperref[rhp]{\rhp}(c).

Let $U_e$, $J_e$, and $I_e$, \( e\in\{a_1,b_1,a_2,b_2\} \), be as in \eqref{g7} and \eqref{g8}. Further, let $\boldsymbol A(\zeta)$ be the Airy matrix \cite{DKMVZ,DKMVZ2}. That is, it is analytic in $\C\setminus\big((-\infty,\infty)\cup L_-\cup L_+\big)$,  $L_\pm:=\big\{\zeta:~\arg(\zeta)=\pm2\pi/3\big\}$, and satisfies
\[
\boldsymbol A_+(s) = \boldsymbol A_-(s) \left\{
\begin{array}{rl}
\left(\begin{matrix} 0 & 1 \\ -1 & 0 \end{matrix}\right), & s\in (-\infty,0), \medskip \\
\left(\begin{matrix} 1 & 0 \\ 1 & 1 \end{matrix}\right), & s\in L_\pm, \medskip \\
\left(\begin{matrix} 1 & 1 \\ 0 & 1 \end{matrix}\right), & s\in (0,\infty),
\end{array}
\right.
\]
where the real line is oriented from $-\infty$ to $\infty$ and the rays $L_\pm$ are oriented towards the origin. It is known that $\boldsymbol A(\zeta)$ has the following asymptotic expansion\footnote{As usual, we say that a function \( F(p) \), which might depend on other variables as well, admits an asymptotic expansion \( F(p) \sim \sum_{k=1}^\infty c_k \psi_k(p) \), where the functions \( \psi_k(p) \) depend only on \( p \) while the coefficients \( c_k \) might depend on other variables but not \( p \), if for any natural number \( K \) it holds that \( F(p)-\sum_{k=0}^{K-1}c_k \psi_k(p) = \mathcal O_K(\psi_K(p)) \).} at infinity:
\begin{equation}
\label{rh6}
\boldsymbol A(\zeta)e^{\frac23\zeta^{3/2}\sigma_3} \sim \frac{\zeta^{-\sigma_3/4}}{\sqrt2}\sum_{k=0}^\infty \left(\begin{matrix} s_k & 0 \\ 0 & t_k \end{matrix}\right)\left(\begin{matrix} (-1)^k & \mathrm i \\ (-1)^k \mathrm i & 1 \end{matrix}\right)  \left(\frac23\zeta^{3/2}\right)^{-k},
\end{equation}
where the expansion holds uniformly in $\C\setminus\big((-\infty,\infty)\cup L_-\cup L_+\big)$, and
\[
s_0=t_0=1, \quad s_k=\frac{\Gamma(3k+1/2)}{54^kk!\Gamma(k+1/2)}, \quad t_k=-\frac{6k+1}{6k-1}s_k, \quad k\geq1.
\]

Let us write $\boldsymbol A_e:=\boldsymbol A$ if \( e\in\{b_1,b_2\} \) and $\boldsymbol A_e:=\sigma_3\boldsymbol A\sigma_3$ if \( e\in\{a_1,a_2\} \). It can be easily checked that $\sigma_3\boldsymbol A\sigma_3$ has the same jumps as $\boldsymbol A$ only with the reversed orientation of the rays. Moreover, one needs to replace $\ic$ by $-\ic$ in \eqref{rh6} when describing the behavior of $\sigma_3\boldsymbol A\sigma_3$ at infinity. Let \( \zeta_e(z) := \big[-n(3/4)\phi_e(z)\big]^{2/3} \), which is conformal in \( U_e \), see \eqref{g12}. Further, put
\[
\boldsymbol J_e(z) := e^{(N_n-n)V(z)\sigma_3/2}\left\{
\begin{array}{rl}
\boldsymbol I, & e=b_2, \\
e^{\pi\ic(\pm\omega)n\sigma_3}, & e=a_2, \\
e^{\pi\ic(\pm \omega-\tau)n\sigma_3}, & e=b_1, \\
e^{-\pi\ic\tau n\sigma_3}, & e=a_1,
\end{array}
\right.
\]
where we use \( \omega \) if $z$ lies to the left of $\Gamma$ and use \( -\omega \) if $z$ lies to the right of $\Gamma$. Then it can be readily verified by using \eqref{g3}  that
\begin{equation}
\label{rh7}
\boldsymbol P_e(z) := \boldsymbol E_e(z)\boldsymbol A_e\big(\zeta_e(z)\big)e^{(2/3)\zeta_e^{3/2}(z)\sigma_3}\boldsymbol J_e(z),
\end{equation} 
satisfies \hyperref[rhp]{\rhp}(a,b) for any matrix \( \boldsymbol E_e(z) \) holomorphic in \( U_e \). It follows immediately from \eqref{rh6} that \hyperref[rhp]{\rhp}(c) will be satisfied if
\begin{equation}
\label{rh8}
\boldsymbol E_e(z) := \left(\boldsymbol{NJ}_e^{-1}\right)(z)\left(\begin{matrix} 1 & -\nu_e\ic \\ -\nu_e\ic & 1 \end{matrix}\right) \frac{\zeta_e^{\sigma_3/4}(z)}{\sqrt2},
\end{equation}
provided this matrix function is holomorphic in \( U_e \), where \( \nu_e \) was defined in \eqref{g11}. By using \hyperref[rhn]{\rhn}(b) and \eqref{g13} one can readily check that \( \boldsymbol E_e(z) \) is holomorphic in \( U_e\setminus\{e\} \). Since \( \zeta_e(z) \) has a simple zero at \( e \), it also follows from \eqref{rh5} and the claim after \eqref{Mn-jumps} that  \( \boldsymbol E_e(z) \) can have at most square root singularity at \( e \) and therefore is in fact holomorphic in the entire disk \( U_e \) as needed.

In fact, it follows from \eqref{rh6}--\eqref{rh8} that
\begin{equation}
\label{rh9}
\boldsymbol P_e(z) \sim  \boldsymbol N(z) \left(\boldsymbol I + \frac1n\sum_{k=0}^\infty \frac{\boldsymbol P_{e,k}(z)}{n^k} \right),
\end{equation}
where the expansion inside the parentheses holds uniformly on \( \partial U_e \) and locally uniformly for \( t\in O_\mathsf{two-cut} \), and
\begin{equation}
\label{rh10}
\boldsymbol P_{e,k-1}(z) = \boldsymbol J_e^{-1}(z)\left(\begin{matrix} 1 & -\nu_e \ic \\ -\nu_e \ic & 1 \end{matrix}\right) \left(\begin{matrix} s_k & 0 \\ 0 & t_k \end{matrix}\right) \left(\begin{matrix} (-1)^k & \nu_e \ic \\ \nu_e(-1)^k \ic & 1 \end{matrix}\right) \boldsymbol J_e(z) \left(-\frac{\phi_e(z)}2\right)^{-k}, \quad k\geq1.
\end{equation}

\subsection{RH Problem with Small Jumps}

Set
\( \Sigma := \left(\big[(\Gamma\setminus J_t)\cup J_+\cup J_-\big]\cap D\right)\cup \left(\cup_e \partial U_e\right) \),  \( D:= \C\setminus\cup_e \overline U_e \) We shall show that for all $n\in\N(\varepsilon)$ large enough there exists a matrix function $\boldsymbol R(z)$ that solves the following Riemann-Hilbert problem (\rhr):
\begin{itemize}
\label{rhr}
\item[(a)] $\boldsymbol R(z)$ is holomorphic in $\C\setminus\Sigma$ and $\lim_{\C\setminus\Gamma\ni z\to\infty}\boldsymbol R(z)=\boldsymbol I$;
\item[(b)] $\boldsymbol R(z)$ has continuous traces on $\Sigma^\circ$ (points with the well defined tangent) that satisfy
\[
\boldsymbol R_+(s) =  \boldsymbol R_-(s) \left\{
\begin{array}{ll}
\boldsymbol P_e(s)\boldsymbol N^{-1}(s), & s\in\partial U_e, \medskip \\
\boldsymbol N(s)  \left(\begin{matrix} 1 & 0 \\  e^{-n\phi_{b_2}(s)-(n-N_n)V(s)} & 1 \end{matrix}\right) \boldsymbol N^{-1}(s), & s\in J_{\pm}\cap D,
\end{array}
\right.
\]
where $\partial U_e$ is oriented clockwise, and
\[
\boldsymbol R_+(s) = \boldsymbol R_-(s) 
\left\{
\begin{array}{rl}
\boldsymbol N(s)\left(\begin{matrix} 1 & e^{n(2\pi\ic\tau+\phi_{a_1}(s))+(n-N_n)V(s)}\\0&1\end{matrix}\right)\boldsymbol N^{-1}(s), & s\in\Gamma(e^{\ic\pi}\infty,a_1)\cap D, \medskip \\
\boldsymbol N_-(s)\left(\begin{matrix} e^{2\pi\ic\omega n} & e^{n\phi_{a_2}(s)+(n-N_n)V(s)}\\ 0 & e^{-2\pi\ic\omega n} \end{matrix}\right)\boldsymbol N_+^{-1}(s), & s\in\Gamma(b_1,a_2)\cap D , \medskip \\
\boldsymbol N(s)\left(\begin{matrix} 1 & e^{n\phi_{b_2}(s)+(n-N_n)V(s)} \\ 0 & 1 \end{matrix}\right)\boldsymbol N^{-1}(s), & s\in\Gamma(b_2,e^{\pi\ic/3}\infty)\cap D.
\end{array}
\right.
\]
\end{itemize}

Observe that \hyperref[rhr]{\rhr} is a well posed problem as $\det(\boldsymbol N(z))\equiv 1$, as explained after \eqref{rh5}, and therefore the matrix is invertible. Recall also that the entries of $\boldsymbol N(z)$ and $\boldsymbol N^{-1}(z)$ are uniformly bounded on \( \Sigma \) for $n\in \N(\varepsilon)$ according to \eqref{MN-up-bound} and \eqref{MN-low-bound}.

To prove solvability of \hyperref[rhr]{\rhr}, let us show that the jump matrices in \hyperref[rhr]{\rhr}(b) are close to the identity.  To this end, set
\begin{equation}
\label{rh11}
\boldsymbol \Delta(s) := (\boldsymbol R_-^{-1}\boldsymbol R_+)(s) - \boldsymbol I, \quad s\in\Sigma.
\end{equation}
Since the entries of $\boldsymbol N(z)$ are uniformly bounded on each $\partial U_e$ with respect to \( n\in\N(\varepsilon) \), it holds by \hyperref[rhp]{\rhp}(c) and \eqref{rh9} that 
\begin{equation}
\label{rh12}
\boldsymbol \Delta(s) \sim \frac1n\sum_{k=0}^\infty \frac{\big(\boldsymbol N \boldsymbol P_{e,k} \boldsymbol N^{-1}\big)(s)}{n^k},
\end{equation}
where this pseudo expansion\footnote{We say that a function \( F(p) \)  admits a pseudo  expansion \( F(p) \sim \sum_{k=1}^\infty c_k \psi_k(p) \) if the coefficients \( c_k \) do depend on \( p \), but the estimates \( F(p)-\sum_{k=0}^{K-1}c_k \psi_k(p) = \mathcal O_K(\psi_K(p)) \) remain valid for any \( K \).}  is valid uniformly on $\partial U_e$. Thus, it holds that
\begin{equation}
\label{rh13}
\|\boldsymbol\Delta\|_{L^\infty(\cup_e\partial U_e)} = \mathcal O_\varepsilon\big(n^{-1}\big).
\end{equation}
Moreover, it follows from \eqref{g9} and the sentence right after that there exists a constant $C_D<1$, depending on the radii of the disks $U_e$, such that $|e^{\phi_{b_2}(s)}|<C_D$ for $s\in\Gamma(b_2,e^{\pi\ic/3}\infty)\cap D$, $|e^{\phi_{a_1}(s)}|<C_D$ for \( s\in\Gamma(e^{\pi\ic}\infty,a_1)\cap D \), and  $|e^{\phi_{a_2}(s)}|<C_D$ for \( s\in\Gamma(b_1,a_2)\cap D \). Therefore,
\begin{equation}
\label{rh14}
\boldsymbol \Delta(s) =\boldsymbol N_-(s)\left(\begin{matrix} 0 & e^{n(\phi_e(s)+C)+(n-N_n)V(s)} \\ 0 & 0 \end{matrix}\right)\boldsymbol N_+^{-1}(s) = \boldsymbol{\mathcal{O}}\left(C_D^n\right)
\end{equation}
on $(\Sigma\cap D)\setminus(J_+\cup J_-)$ since $C$ is either zero or purely imaginary, the entries of \( \boldsymbol N(z) \) are bounded and so is the sequence \( \{n-N_n\} \), the subscripts $\pm$ are needed only on $\Gamma(b_1,a_2)\cap D$, and we used the fact
\[
\boldsymbol N_-(s) e^{2\pi\ic\omega n\sigma_3} \boldsymbol N_+^{-1}(s) = \boldsymbol I, \quad s\in\Gamma(a_1,b_2),
\]
see \hyperref[rhn]{\rhn}(b). Similarly, we get that
\begin{equation}
\label{rh15}
\boldsymbol \Delta(s) =\boldsymbol N(s)\left(\begin{matrix} 0 & 0 \\ e^{-n\phi_{b_2}(s)-(n-N_n)V(s)} & 0 \end{matrix}\right)\boldsymbol N^{-1}(s) = \boldsymbol{\mathcal{O}}\left(C_D^n\right)
\end{equation}
on \( J_\pm\cap D \) for a possibly adjusted constant \( C_D \), where we used the fact that \( \mathrm{Re}(\phi_{b_2}(s))>0 \) for \( s\in J_\pm\setminus E \), see Figure~\ref{lens}. 

Equations \eqref{rh13}, \eqref{rh14}, and \eqref{rh15} show that $\boldsymbol \Delta(s)$ is uniformly close to zero. Since the entries of $\boldsymbol N(z)$ are holomorphic at infinity and $ e^{n\phi_e(s)}$ is geometrically small as $\Gamma\ni s\to\infty$, $\boldsymbol \Delta(s)$ is close to zero in $L^2$-norm as well. Then it follows from the same analysis as in \cite[Corollary~7.108]{Deift} that $\boldsymbol R(z)$ exists for all $n\in\N(\varepsilon)$ large enough and it holds uniformly in $\C$ that
\begin{equation}
\label{rh16}
\boldsymbol R(z) = \boldsymbol I + \boldsymbol{\mathcal O}_\varepsilon\big(n^{-1}\big).
\end{equation}

\subsection{Solution of \hyperref[rhy]{\rhy}}

Given $\boldsymbol R(z)$, \( \boldsymbol N(z) \), and \( \boldsymbol P_e(z) \), solutions of \hyperref[rhr]{\rhr}, \hyperref[rhn]{\rhn}, and \hyperref[rhp]{\rhp}, respectively, it is a trivial verification to check that \hyperref[rhs]{\rhs} is solved by
\begin{equation}
\label{rh17}
\boldsymbol S(z) = \left\{
\begin{array}{ll}
(\boldsymbol R\boldsymbol N)(z) & \text{in} \quad D\setminus[(\Gamma\setminus J_t)\cup J_+\cup J_-], \medskip \\
(\boldsymbol R\boldsymbol P_e)(z) & \text{in} \quad U_e,~~~e\in\{a_1,b_1,a_2,b_2\}.
\end{array}
\right.
\end{equation}

Let $K$ be a compact set in $\C\setminus\Gamma$. We can always arrange so that the set $K$ lies entirely within the unbounded component of the contour $\Sigma$. Then it follows from \eqref{rh2}, \eqref{rh4}, and \eqref{rh17} that
\begin{equation}
\label{rh18}
\boldsymbol Y(z) = e^{-n\ell^*\sigma_3/2}(\boldsymbol{RN})(z)e^{n(g(z)+\ell^*/2)\sigma_3}, \quad z\in K.
\end{equation}
Subsequently, by using \eqref{rh1} and \eqref{g1}, we see that 
\[
P_n(z) = [\boldsymbol Y(z)]_{11} = \big([\boldsymbol R(z)]_{11}[\boldsymbol N(z)]_{11}+[\boldsymbol R(z)]_{12}[\boldsymbol N(z)]_{21}\big)e^{n\mathcal Q(z) + \frac n2(V(z)-\ell_*)}.
\]
Therefore, it follows from \eqref{rh5} and \eqref{rh16} that
\begin{equation}
\label{rh19}
\psi_n(z)e^{-n\mathcal Q(z)} = \left( 1 + \mathcal O_\varepsilon\big(n^{-1}\big) \right) D^{N_n-n}(z) \frac{M_{n,1}^{(0)}(z)}{M_{n,1}^{(0)}(\infty)} + \mathcal O_\varepsilon\big(n^{-1}\big) \frac{D^{N_n-n}(z)}{\mathcal D^{2(n-N_n)}(\infty)} \frac{M_{n,0}^{(0)}(z)}{M_{n,0}^{(1)}(\infty)}.
\end{equation}
Asymptotic formula \eqref{tc17} now follows from \eqref{MN-up-bound} and \eqref{MN-low-bound}, boundedness of \( \{N_n-n\} \), and definitions of \( M_{n,1}(\z) \) in \eqref{ab7}, of \( \Theta_n(z) \) right before \eqref{ab5}, and of \( \vartheta_n(z) \) right after Proposition~\ref{prop:A}.

Now, let \( K \) be any compact set in \( \C\setminus J_t \). Write \( K=K_1\cup K_2 \), where \( K_1,K_2 \) are compact, \( K_1 \) does not intersect \( \Gamma \) and \( K_2 \) lies entirely within the region \( \{\re(\phi_{b_2}(z))<0\} \), see Figures~\ref{fig:tc} and~\ref{lens}. Again, the lens \( \Sigma \) can be adjusted so that \( K_1 \) lies in the unbounded component of the complement of \( \Sigma \). Hence, the estimate \eqref{tc17} on \( K_1 \) follows as before. To obtain it on \( K_2 \), recall that we had a lot of freedom in choosing \( \Ga \) away from \( J_t \). That is, \( \Ga \) can be deformed into \( \Ga^\prime \) that avoids \( K_2 \) and belongs to \( \{\re(\phi_{b_2}(z))<0\} \) away from \( J_t \). Then \hyperref[rhy]{\rhy}, formulated on \( \Ga^\prime \), can be solved exactly as before since estimates \eqref{rh14} and \eqref{rh15} remain the same (with a possibly modified constant \( C_D \)), and therefore \eqref{tc17} can be shown via \eqref{rh18}--\eqref{rh19}.

Finally, take $K\subset J_t^\circ$. It again follows from \eqref{rh1}, \eqref{rh2}, \eqref{rh4}, and \eqref{rh17} that
\begin{multline*}
P_n(s) = [\boldsymbol Y(s)]_{11} = \left([\boldsymbol R(s)]_{11}\left([\boldsymbol N(s)]_{11+} + [\boldsymbol N(s)]_{12+}e^{-n\phi_{b_2+}(s)-(n-N_n)V(s)}\right)\right. + \\ \left.[\boldsymbol R(s)]_{12}\left( [\boldsymbol N(s)]_{21+} +  [\boldsymbol N(s)]_{22+}e^{-n\phi_{b_2+}(s)-(n-N_n)V(s)} \right)\right)e^{ng_+(s)}, \quad s\in K.
\end{multline*}
Now, \eqref{g5} and \hyperref[rhn]{\rhn}(b) yield that
\[
[\boldsymbol N(s)]_{i2+} e^{-n\phi_{b_2+}(s)-(n-N_n)V(s)+ng_+(s)} = [\boldsymbol N(s)]_{i1-} e^{ng_-(s)}
\]
for \( s\in J_t^\circ \) and \( i\in\{1,2\} \). Hence, we get from \eqref{g1} and \eqref{rh16} that
\begin{multline*}
\psi_n(s) = \left( 1 + \mathcal O_\varepsilon\big(n^{-1}\big) \right) \left( [\boldsymbol N(s)]_{11+}e^{n\mathcal Q_+(s)} + [\boldsymbol N(s)]_{11-}e^{n\mathcal Q_-(s)} \right) + \\ \mathcal O_\varepsilon\big(n^{-1}\big)\left( [\boldsymbol N(s)]_{21+}e^{n\mathcal Q_+(s)} + [\boldsymbol N(s)]_{21-}e^{n\mathcal Q_-(s)} \right)
\end{multline*}
for \( s\in K \). Since the traces \( 2\mathcal Q_{\pm}(s)=\phi_{b_2\pm}(s) \), \( s\in J_t\), are purely imaginary by \eqref{g4}, the above asymptotic formula yields \eqref{tc18} in the same way \eqref{rh19} yielded \eqref{tc17}.

\section{Recurrence Coefficients} 
\label{s:ae}

In this section, we will use the notation $P_n(z;t, N) = z^n +  \sum_{k = 0}^{n - 1} (P_{n,t,N})_k z^k$ and 
\[
\boldsymbol K(z)  =  \boldsymbol I + z^{-1} \boldsymbol K_1(n,t,N) + z^{-2} \boldsymbol K_2(n,t,N) + \mathcal{O}\big(z^{-3}\big) \quad \text{as} \quad z \to \infty,
\]
where \(  \boldsymbol K \in \{ \boldsymbol M, \boldsymbol N, \boldsymbol R, \boldsymbol T ,\boldsymbol Yz^{-n\sigma_3}  \}  \). To begin, recall \eqref{cm4} and \eqref{cm5b}. That is,
\begin{equation}
\label{rc1}
\left\{
\begin{array}{lll}
\ga_n^2(t,N) & = & h_n(t,N)/h_{n-1}(t,N), \medskip \\
\beta_n(t,N) & = & (P_{n,t,N})_{n - 1} - (P_{n+1,t,N})_n,
\end{array}
\right.
\end{equation} 
where  for the second formula  in \eqref{rc1} to be valid it must hold that \( \deg(P_{n+1}(\cdot;t,N))=n+1 \), see the discussion after \eqref{cm5b}. It follows from the orthogonality relations \eqref{cm3} and \eqref{cm5a} that
\begin{eqnarray*}
z^n\big(\mathcal{C}P_ne^{-NV}\big)(z) &=& -\frac1z\frac1{2\pi\ic} \int_\Ga s^nP_ne^{-NV}\dd s - \frac1{z^2}\frac1{2\pi\ic} \int_\Ga s^{n+1}P_ne^{-NV}\dd s + \mathcal O\left( z^{-3}\right) \\
& = & \frac{h_n(t,N)}{2\pi\ic}\left(- \frac1z + \frac{ (P_{n+1,t,N})_n}{z^2}\right) + \mathcal O\left( z^{-3}\right),
\end{eqnarray*}
where the second expression for the coefficient next to \( z^{-2} \) is obtained under the condition that \( \deg(P_{n+1}(\cdot;t,N))=n+1 \) (the first expression, of course, is valid regardless this condition) and is based on the observation
\[
\int_\Ga s^{n+1}P_ne^{-NV}\dd s = \int_\Ga \left(s^{n + 1} - P_{n + 1}\right) P_ne^{-NV}\dd s = (P_{n + 1,N})_n \int s^nP_ne^{-NV}\dd s,
\]
where all polynomials correspond to the same parameter \( N \). Hence, in the considered case, it follows from \eqref{rh1} that 
\begin{equation}
\label{rc2}
{\boldsymbol Y}_1(n,t,N) = \begin{pmatrix} (P_{n,t,N})_{n - 1} & - \frac{h_n(t,N)}{2\pi \ic} \medskip \\  - \frac{2\pi \ic}{h_{n - 1}(t,N)} & -(P_{n,t,N})_{n - 1}  \end{pmatrix} \quad \text{and} \quad \boldsymbol Y_2(n,N) = \begin{pmatrix} * &  \frac{h_n(t,N)}{2\pi\ic}(P_{n+1,t,N})_n \medskip \\ * & *  \end{pmatrix}.
\end{equation}

This, in turn, yields the formulae
\begin{equation}
\label{rc3}
\begin{cases}
h_n(t, N) & = -2\pi \ic [\boldsymbol Y_1(n,t, N)]_{12} , \smallskip \\
\gamma_n^2 (t, N) & = [\boldsymbol Y_1(n,t, N)]_{12} [\boldsymbol Y_1(n, t,N)]_{21}, \smallskip \\
\beta_n (t, N) & = [\boldsymbol Y_1(n, t,N)]_{11} + [\boldsymbol Y_2(n, t,N)]_{12}[\boldsymbol Y_1(n, t,N)]_{12}^{-1}.
\end{cases}
\end{equation}

Assume now that \( \deg(P_{n+1}(\cdot;t^*,N)) = n \) for some value \( t=t^* \). As explained after \eqref{cm5b}, \( \beta_n(t^*,N)=\infty \) in this case. Moreover, we also pointed out after \eqref{cm4} that if \( \deg(P_k(\cdot;t,N)) = k \) for \( k\in\{n,n+1 \} \), then \( \beta_n(t,N) \) is finite.  Hence, for \( n\in \N(t,\varepsilon) \), \( P_{n+1}(z;t,N) \) degenerates if and only if \( \beta_n(t,N) \) is infinite. Thus, due to meromorphic dependence of \( \beta_n(t,N) \) on the parameter \( t \), if \( P_{n+1}(z;t^*,N) \) degenerates, then \( \deg(P_{n+1}(\cdot;t,N)) = n+1 \) in some small punctured neighborhood of \( t^* \). Moreover, it follows from Theorem~\ref{geometry2} and \eqref{tc4} that \( \z_{n,1}(t) \) is real analytic in real and imaginary parts of \( t \). Hence, if \( n\in \N(t^*,\varepsilon) \), then \( n\in \N(t,\varepsilon) \) for some sequence of values \( t \) converging to \( t^* \). Clearly, along this sequence the third relation in \eqref{rc3} remains valid. Since the entries of \( \boldsymbol Y(z) \) meromorphically depend on \( t \), then the third relation in \eqref{rc3} must remain valid at \( t^* \) as well (this simply means that \( \beta_n(t,N)=\infty\) corresponds to \( [\boldsymbol Y_2(n, t, N)]_{12}\neq 0 \) and \( [\boldsymbol Y_1(n, t, N)]_{12} =0 \)). Hence, \eqref{rc3} holds regardless the degree of \( P_{n+1}(z;t,N) \) and we drop indicating the dependence on \( t \) altogether.

As before, let $\mu_1=\mu_1(t)$ be the first moment of the equilibrium measure \( \mu_t \) (hopefully this slight abuse of notation will not cause a confusion). Then \eqref{oc4} implies that
\begin{equation}
\label{rc4}
e^{-ng(z) \sigma_3} = z^{-n\sigma_3} \left[ \boldsymbol I + \dfrac{n \mu_1}{z}\sigma_3  + \boldsymbol{\mathcal O} \left(\dfrac{1}{z^2}\right) \right], 
\end{equation}
where \( \boldsymbol{\mathcal O}(\cdot) \) is a diagonal matrix. Combining \eqref{rc4} with \eqref{rh2} yields
\begin{equation} 
\label{Y-T}
\begin{cases}
	[\boldsymbol T_1(n, N)]_{12} &= e^{n \ell_*} [\boldsymbol Y_1(n, N)]_{12}, \medskip \\
	[\boldsymbol T_1(n, N)]_{21} &= e^{-n \ell_*} [\boldsymbol Y_1(n, N)]_{21},
\end{cases}
\;\;
\begin{cases}
	[\boldsymbol T_1 (n, N)]_{11} &= [\boldsymbol Y_1 (n, N)]_{11} + n \mu_1, \medskip \\
	[\boldsymbol T_2(n, N)]_{12} &= e^{n \ell_*} \left( [\boldsymbol Y_2(n, N)]_{12} - n \mu_1 [\boldsymbol Y_1(n, N)]_{12} \right). 
\end{cases}
\end{equation}
This, in turn, leads to
\begin{equation}
\label{rc6}
\left\{   
\begin{array}{lll}
	h_n(N) & = & -2\pi \ic e^{-n \ell_*} [\boldsymbol T_1(n, N)]_{12} , \medskip \\
	\gamma_n^2 (N) & = & [\boldsymbol T_1(n, N)]_{12} [\boldsymbol T_1(n,N)]_{21},\medskip \\ 
	\beta_n (N) & = & [\boldsymbol T_1(n, N)]_{11} + [\boldsymbol T_2(n, N)]_{12}[\boldsymbol T_1(n,N)]_{12}^{-1}.
\end{array}
\right.
\end{equation}
It was shown in \cite[Theorem 7.10]{DKMVZ2} that $\boldsymbol R(z)$ admits an asymptotic expansion 
\begin{equation}
	\label{r-expansion}
	\boldsymbol R(z) = \boldsymbol I + \sum_{k = 1}^{\infty} \dfrac{\boldsymbol R^{(k)}(z)}{n^k}, \quad n \to \infty \quad  \text{ and } \quad z \in \C \setminus \cup_e \partial U_e. 
\end{equation}
The matrices $\boldsymbol R^{(k)}(z)$ satisfy an additive Riemann-Hilbert Problem: 
\begin{enumerate}
	\item[(a)] $\boldsymbol R^{(k)}(z)$ is analytic in $\C \setminus \cup_e \partial U_e$; 
	\item[(b)] $\boldsymbol R_+^{(k)}(z) = \boldsymbol R_-^{(k)}(z) + \sum_{j = 1}^{k} \boldsymbol R_-^{(k - j)}(z) \boldsymbol \Delta_j(z)$ for $z \in \cup_e \partial U_e$; 
	\item[(c)] as $z \to \infty,$ $\boldsymbol R^{(k)}(z)$ admits an expansion of the form
	\begin{equation}
		\label{rk-expansion}
		\boldsymbol R^{(k)}(z) = \dfrac{\boldsymbol R_1^{(k)}}{z} + \dfrac{\boldsymbol R_2^{(k)}}{z^2} + \mathcal{O}\left( \dfrac{1}{z^3} \right) .
	\end{equation}
\end{enumerate}
Recall that $\boldsymbol T(z) = (\boldsymbol R \boldsymbol N)(z)$ for $z$ in the vicinity of infinity by \eqref{rh2} and \eqref{rh18}. It also follows from \eqref{rh5}, \eqref{MN-up-bound} and \eqref{MN-low-bound} that matrices \( \boldsymbol N(z) \) form a normal family in \( n\in\N(\varepsilon) \) in the vicinity of infinity. Plugging expansions \eqref{r-expansion} and \eqref{rk-expansion} into the product $\boldsymbol T(z) = (\boldsymbol R \boldsymbol N)(z)$ yields
\begin{equation}
\label{rc7}
\begin{cases}
	\boldsymbol T_1(n, N) &= \boldsymbol N_1(n, N) + \dfrac{\boldsymbol R^{(1)}_1}{n} + \dfrac{\boldsymbol R^{(2)}_1}{n^{2}} + \boldsymbol{\mathcal O}_\varepsilon\left( \dfrac{1}{n^3} \right), \medskip \\
	\boldsymbol T_2(n, N) &= \boldsymbol N_2(n, N) + \dfrac{\boldsymbol R_1^{(1)} \boldsymbol N_1(n, N) + \boldsymbol R_2^{(1)}}{n} + \dfrac{\boldsymbol R_1^{(2)}\boldsymbol N_1(n, N)  + \boldsymbol R_2^{(2)}}{n^2} +  \boldsymbol{\mathcal O}_\varepsilon\left( \dfrac{1}{n^3} \right),
\end{cases}
\end{equation}
where we used the normality of \( \boldsymbol N(z) \) to deduce uniform boundedness of the entries of \(  \boldsymbol N_k(n, N) \), \( k\in\{1,2\} \), to maintain \( \boldsymbol{\mathcal O}_\varepsilon ( n^{-3}) \) as the order of the error terms (it is, of course, now depends on \( \varepsilon \) as do the bounds for the entries \(  \boldsymbol N_k(n, N) \)).

To use \eqref{rc7}, we move on to computing $\boldsymbol N_k(n,N)$. To that end, it follows from \eqref{D-exp} that  
\[
\mathcal D^{(N - n) \sigma_3}(z)  =  \mathcal{D} ^{(N - n) \sigma_3}(\infty) \left( \boldsymbol I + \dfrac{ (N - n) D_1}{3z}\sigma_3 + \boldsymbol{\mathcal O}\left(\dfrac{1}{z^2} \right) \right)
\]
as \( z\to\infty \), where \( \boldsymbol{\mathcal O}(z^{-2}) \) is a diagonal matrix and \( D_1 := t^2 + \mu_1 - 3\varsigma d_2/2\). Further, recall the function \( \gamma(z) \) from \eqref{nt1}. Then
\[
\gamma(z) = 1 - \frac{S_1}{4z} + \frac{C-3S_1^2/8}{4z^2} + \mathcal O\big(z^{-3} \big) \quad \text{and} \quad \gamma^{-1}(z) = 1 + \frac{S_1}{4z} - \frac{C-5S_1^2/8}{4z^2} + \mathcal O\big(z^{-3} \big)
\]
as \( z\to\infty \), where \( S_k = \sum_{i=1}^2\big( b_i^k-a_i^k\big) \) and \( C = a_1^2+a_2^2 + b_1b_2 + a_1a_2 - \sum_{i,j=1}^2 b_ia_j \) . Then it follows from \eqref{nt2} that
\[
A(z) = 1 + \mathcal O\big(z^{-2} \big) \quad \text{and} \quad B(z) = \frac{S_1}{4\ic}\frac1z +  \frac{S_2}{8\ic}\frac1{z^2} + \mathcal O\big(z^{-3} \big) 
\]
as \( z\to\infty \) since \( S_1^2-2C = S_2 \). The above power series expansions, along with \eqref{ab7} and \eqref{rh5}, yield that
\begin{equation}
\label{N-1}
\boldsymbol N_1(n, N) = \frac{(N - n)D_1}{3} \sigma_3 ~+~  \begin{pmatrix}  
\dfrac{\dd}{\dd z} \left( \log \Theta^{(0)}_{n, 1}(1/z) \right)\biggl|_{z = 0} & - \mathcal{D}^{2(n - N)}(\infty) \dfrac{S_1}{4 \ic}\dfrac{\Theta^{(1)}_{n, 1}(\infty)}{\Theta^{(0)}_{n, 1} (\infty)} \medskip \\  
	\mathcal{D}^{2(N- n)}(\infty) \dfrac{S_1}{4 \ic}\dfrac{\Theta_{n, 0}^{(0)}(\infty)}{\Theta_{n, 0}^{(1)} (\infty)} &  \dfrac{\dd}{\dd z} \left( \log \Theta_{n, 0}^{(1)}(1/z) \right)\biggl|_{z = 0}
\end{pmatrix}
\end{equation}
and 
\begin{equation}
\label{N-2}
[\boldsymbol N_2(n, N)]_{12} =   \frac{\mathcal{D}^{2(n - N) }(\infty)}{ \Theta_{n, 1}^{(0)}(\infty) } \frac{S_1}{4\ic} \left[\left(\frac{(N - n)D_1}{3} - \frac{S_2}{2S_1} \right) \Theta^{(1)}_{n, 1}(\infty) \right. \\  \left. - \dfrac{\dd}{\dd z} \left( \Theta^{(1)}_{n, 1}(1/z) \right)\biggl|_{z = 0} \right].
\end{equation}
Using \eqref{N-1} and \eqref{rc7} in \eqref{rc6} gives
\[
h_n(N_n) = \dfrac{\pi S_1}{2} e^{-n\ell^*}  \mathcal{D}^{2(n - N)}(\infty)\dfrac{\Theta_{n, 1}^{(1)}(\infty)}{\Theta_{n, 1}^{(0)}(\infty)} + \mathcal{O}_{\varepsilon}\left( \dfrac{1}{n} \right)
\]
as \( \N(\varepsilon) \ni n \to \infty \), which proves the first formula in \eqref{hn-asymp} (recall \eqref{all-thetas}). Similarly,  
\[
\gamma_n^2(t, N) = \dfrac{S_1^2}{16} \dfrac{\Theta_{n, 1}^{(1)}(\infty) \Theta_{n, 0}^{(0)}(\infty)}{\Theta_{n, 1}^{(0)}(\infty) \Theta_{n, 0}^{(1)}(\infty)}  + \mathcal{O}_{\varepsilon}\left( \dfrac{1}{n} \right) \qasq \N(\varepsilon) \ni n \to \infty,
\]
which finishes the proof of \eqref{hn-asymp}. Define \( \Theta_{n-1,1}(\z;N_n) \) as in \eqref{ab5} but with \( \z_{n-1,1}(N_{n-1}) \) replaced by \( \z_{n-1,1}(N_n) \). It follows from \eqref{ab6}, the last claim of Proposition~\ref{prop:Nt}, and \eqref{new4} that \( \Theta_{n,0}(\z) \) is a constant multiple of \( \Theta(\z)\Theta_{n-1,1}(\z;N_n) \). Then we get from \eqref{new6} that
\begin{equation}
\label{rc8}
\gamma_n^2(t, N) = e^{-\ell_*}\mathcal D^2(\infty) \dfrac{\Theta_{n, 1}^{(1)}(\infty) \Theta_{n-1, 1}^{(0)}(\infty;N_n)}{\Theta_{n, 1}^{(0)}(\infty) \Theta_{n-1, 1}^{(1)}(\infty;N_n)}  + \mathcal{O}_{\varepsilon}\left( \dfrac{1}{n} \right) \qasq \N(\varepsilon) \ni n \to \infty.
\end{equation}
Finally,  \eqref{rc6} and \eqref{rc7} yield
\[
\beta_n(N_n) = \frac{[\boldsymbol N_1]_{11}[\boldsymbol N_1]_{12}+[\boldsymbol N_2]_{12} + \mathcal O_\varepsilon\big( n^{-1}\big)}{ [\boldsymbol N_1]_{12} + \mathcal O_\varepsilon\big( n^{-1}\big) },
\]
which, in view of \eqref{N-1} and \eqref{N-2} gives
\[
\beta_n(N_n) = \frac{\Theta_{n,1}^{(1)}(\infty) \left( \dfrac{S_2}{2S_1} + \dfrac{\dd}{\dd z} \left(\log \Theta^{(0)}_{n, 1}(1/z) + \log \Theta^{(1)}_{n, 1}(1/z) \right)\biggl|_{z = 0}\right) + \mathcal O_\varepsilon\big( n^{-1}\big)}{\Theta_{n,1}^{(1)}(\infty) + \mathcal O_\varepsilon\big( n^{-1}\big)},
\]
as claimed in \eqref{betan-asymp} (we need to divide both numerator and denominator by \( \Theta_{n,1}^{(0)}(\infty) \)).

\end{document}